\documentclass[12pt,oneside,reqno,a4paper]{amsart}

\usepackage[latin1]{inputenc}
\usepackage{amsthm}
\usepackage{amsmath}
\usepackage{amsfonts}
\usepackage{amssymb}
\usepackage{bigints}
\usepackage{graphicx}
\usepackage[colorlinks]{hyperref}
\usepackage{cleveref} 

\usepackage[square,numbers]{natbib}
\usepackage[paper=a4paper, top=3cm,bottom=2.5cm,left=2.5cm,right=2.5cm]{geometry}
\usepackage[normalem]{ulem}

\usepackage{thmtools}
\usepackage{thm-restate}

\usepackage{xcolor}

\hypersetup{
    colorlinks = true,
    linkcolor = blue,
    citecolor = cyan
}

\newcommand{\Weizsacker}{Weizs\"{a}cker}
\newcommand{\TFW}{\textnormal{TFW}}
\newcommand{\loc}{\textnormal{loc}}
\newcommand{\unif}{\textnormal{unif}}

\newcommand{\id}{\textnormal{ d}}

\newcommand{\wt}{\widetilde}

\newcommand{\R}{\mathbb{R}^{3}}

\newcommand{\Holder}{H\"{o}lder}

\newcommand{\smcb}{ \textstyle{\frac{1}{|\cdot|}}  }
\newcommand{\smfrac}[2]{ \textstyle{ \frac{#1}{#2} } }

\newcommand{\ou}{\overline{u}}
\newcommand{\ophi}{\overline{\phi}}
\newcommand{\om}{\overline{m}}

\newcommand{\spt}{\textnormal{spt}}

\newcounter{mydef}
\newcounter{defcount}

\newtheorem{theorem}[mydef]{Theorem}
\newtheorem{definition}[defcount]{Definition}
\newtheorem{lemma}[mydef]{Lemma}
\newtheorem{proposition}[mydef]{Proposition}

\newtheorem*{assumption*}{Assumption}
\newtheorem{corollary}[mydef]{Corollary}
\newtheorem*{theorem*}{Theorem}
\newtheorem*{proposition*}{Proposition}
\newtheorem*{lemma*}{Lemma}

\theoremstyle{definition}
\newtheorem{remark}{Remark}

\numberwithin{theorem}{section}
\numberwithin{lemma}{section}
\numberwithin{corollary}{section}
\numberwithin{proposition}{section}
\numberwithin{equation}{section}

\def\XXint#1#2#3{{\setbox0=\hbox{$#1{#2#3}{\int}$}
\vcenter{\hbox{$#2#3$}}\kern-.5\wd0}}

\title[Convergence from Yukawa to Coulomb in the TFW Model]{Convergence Rates from Yukawa to Coulomb Interaction in the Thomas--Fermi--von Weizs\"{a}cker Model}

\author{F. Q. Nazar}
\date{\today}

\thanks{FQN is funded by the MASDOC doctoral training centre, EPSRC grant
  EP/H023364/1.}

\address{F. Q. Nazar \\ Mathematics Institute \\ Zeeman Building \\
  University of Warwick \\ Coventry CV4 7AL \\ UK}
\email{F.Q.Nazar@warwick.ac.uk}

\begin{document}

\begin{abstract}
We establish uniform convergence, with explicit rate, of the solution to the Thomas--Fermi--von Weizs\"{a}cker (TFW) Yukawa model to the solution of the TFW Coulomb model, for general condensed nuclear configurations. As a consequence, we show the convergence of forces from the Yukawa to the Coulomb model. These results rely on an extension of Nazar \& Ortner (2015) to the Yukawa setting. Auxiliary results of independent interest shown also include new existence, uniqueness and stability results for the Yukawa ground state.

\end{abstract} 

\maketitle

%
%
%

\section{Introduction}

One of the challenges in molecular simulation is treating the interaction of charged particles using the Coulomb potential. Due to the long-range of the Coulomb potential $\smfrac{1}{|x|}$, the Yukawa potential $Y_{a}(x) = \smfrac{e^{-a|x|}}{|x|}$, for $a > 0$, is often used as a short-ranged approximation \cite{C/LB/L, Lahbabi-MeanField, Lahbabi-Defects, Blanc-DefinitionofGroundState, Rowlinson}. The Yukawa potential also appears in the Thomas--Fermi theory of impurity screening, where the parameter $a > 0$ represents the inverse screening length of a metal \cite{Kaxiras, Kittel, Ashcroft}. 

The aim of this paper is to establish the uniform convergence of the Yukawa ground state to the Coulomb ground state, in the Thomas--Fermi--von \Weizsacker\, (TFW) model. The main technical result estimates the rate of convergence. A rigorous statement is given in \Cref{Theorem - Yukawa Coulomb Comparison}.
\begin{theorem*}
Let $m \in L^\infty(\R)$ represent a nuclear charge
  distribution satisfying
  \begin{align*}
    m \geq 0 \qquad \text{and} \qquad
    \lim_{R \to \infty} \frac{1}{R} \inf_{x \in \R} \int_{B_{R}(x)} m(z) \id z = + \infty.
  \end{align*}
  Let the corresponding Coulomb ground state electron density and electrostatic
  potential, denoted by $u, \phi : \R \to \mathbb{R}$, satisfy the TFW equations,
  \begin{align*}
    - \Delta u &+ \frac{5}{3} u^{7/3} - \phi u = 0,  \\
    - \Delta \phi &= 4\pi (m - u^{2}),
  \end{align*}
and for $a > 0$, let the corresponding Yukawa ground state, denoted by $(u_{a},\phi_{a})$, satisfy the TFW Yukawa equations
    \begin{align*}
      &- \Delta u_{a} + \frac{5}{3} u_{a}^{7/3} - \phi_{a} u_{a} = 0,  \\
      &- \Delta \phi_{a} + a^{2} \phi_{a} = 4\pi (m - u_{a}^{2}).
    \end{align*}
  
  Then there exists $C > 0$ such that
  \begin{align}
\| u_{a} - u \|_{W^{2,\infty}(\R)} + \| \phi_{a} - \phi \|_{W^{2,\infty}(\R)} \leq C a^{2}. \label{goal1}
  \end{align}
\end{theorem*}

To the best of the author's knowledge, this is the first result that provides a rate of convergence for ground states from Yukawa to Coulomb interaction, for any electronic structure model.

An important consequence of (\ref{goal1}) is an estimate for the rate of convergence of forces in the TFW model, when passing from the Yukawa to Coulomb interaction. Given a countable collection of nuclei $Y = (Y_{j})_{j \in \mathbb{N}} \subset \R$ and $a > 0$, the TFW Yukawa and Coulomb energy densities, $\mathcal{E}_{a}(Y;x)$ and $\mathcal{E}(Y;x)$ respectively, can be defined. It follows from (\ref{goal1}) that
\begin{align}
\left| \int_{\R} \left( \frac{\partial \mathcal{E}_{a}}{\partial Y_{k}} - \frac{\partial \mathcal{E}}{\partial Y_{k}} \right)(Y;x) \id x \right| &\leq C a^{2}. \label{goal2}
\end{align}
A rigorous statement of this result is given in \Cref{Theorem - Forcing Yuk and Coulomb Comparison}.


In a forthcoming article \cite{Paper2}, the aim will be to generalise
the analysis of variational problems for the mechanical
response to defects in an infinite crystal \cite{EOS} to electronic structure models, using the TFW model with Coulomb interaction. The uniform convergence of forces from Yukawa to Coulomb suggests that one could construct an approximate mechanical response problem using the Yukawa interaction. This could be more efficient for the purposes of numerical simulations as it replaces the long-range Coulomb interaction with the short-ranged Yukawa interaction. The result (\ref{goal2}) suggests that the error in the electron density may propagate into an $O(a^{2})$ error in the equilibrium configuration. This will be explained in future work.

The remainder of this article is organised as follows: In Section \ref{Section2} the definition of the TFW model is recalled and the relevant existing results are summarised. In Section \ref{Section3} the main technical results are stated, including the rigorous statement of the convergence result \eqref{goal1}. Applications are presented in Section \ref{Section4}, followed by the detailed proofs of the results in \Cref{Section5}. An additional technical argument is given in the Appendix, that extends uniqueness of the Yukawa ground state to all $a > 0$.

\begin{remark}
The analytical approach presented closely follows and adapts the study of the TFW equations in \cite{C/LB/L, Paper1}. An overview of the TFW equations can be found in \cite{Paper1} and \cite{Rowlinson} provides a background on the Yukawa potential and its various applications.

To the best of the author's knowledge, the closest existing result to (\ref{goal1}) in the literature is \cite[Proposition 2.30]{C/LB/L}, which shows $u_{a} \to u$ strongly in $H^{1}_{\loc}(\R)$ as $a \to 0$, for periodic and neutral TFW systems, but does not estimate the rate. \qed
\end{remark}

\subsection*{Acknowledgements}
The author thanks Virginie Ehrlacher and Xavier Blanc for helpful
discussions about the TFW model in the Yukawa setting.

\section{The TFW Yukawa Model}
\label{Section2}

For $p \in [1,\infty]$ define the function spaces
\begin{align*}
L^{p}_{\loc}(\R) &:= \{ \, f: \R \to \mathbb{R} \, | \, \forall \, K \subset \R \text{ compact}, f \in L^{p}(K) \, \} \quad \text{ and} \\
L^{p}_{\unif}(\R) &:= \{ \, f \in L^{p}_{\loc}(\R) \, | \, \sup_{x \in \R} \| f \|_{L^{p}(B_{1}(x))} < \infty \, \}.
\end{align*}
For $k \in \mathbb{N}$, $H^{k}_{\loc}(\R), H^{k}_{\unif}(\R)$ are defined analogously. For a multi-index $\alpha = (\alpha_{1},\alpha_{2},\alpha_{3})$, define the partial derivative $\partial^{\alpha} = \partial_{1}^{\alpha_{1}} \partial_{2}^{\alpha_{2}} \partial_{3}^{\alpha_{3}}$. Throughout this paper, $\alpha, \beta$ denote three-dimensional multi-indices.

The Coulomb interaction, for $f, g \in L^{6/5}(\R)$, is given by
\begin{align*}
D_{0}(f,g) = \int_{\R} \int_{\R} \frac{f(x)g(y)}{|x-y|} \id x \id y = \int_{\R} \left( f * \smcb \right)(y) g(y) \id y.
\end{align*}
and is finite due to the Hardy--Littlewood--Sobolev estimate \cite{Aubin_HLS}. The Yukawa interaction is a short-range approximation to the Coulomb interaction, with the Yukawa potential $Y_{a}(x) = \smfrac{e^{-a|x|}}{|x|}$, for $a > 0$, replacing the Coulomb potential $\smfrac{1}{|x|}$. The parameter $a > 0$ controls the range of the interaction, in particular one formally recovers the long-ranged Coulomb interaction as $a \to 0$. The Yukawa interaction, for $a > 0$ and $f, g \in L^{2}(\R)$, is given by
\begin{align*}
D_{a}(f,g) = \int_{\R} \int_{\R} \frac{f(x)e^{-a|x-y|}g(y)}{|x-y|} \id x \id y = \int_{\R} \left( f * Y_{a} \right)(y) g(y) \id y,
\end{align*}
which is finite as Cauchy-Schwarz' and Young's inequality for convolutions imply
\begin{align*}
|D_{a}(f,g)| \leq \| Y_{a} \|_{L^{1}(\R)} \| f \|_{L^{2}(\R)} \| g \|_{L^{2}(\R)} \leq C a^{-2} \| f \|_{L^{2}(\R)} \| g \|_{L^{2}(\R)}.
\end{align*}
Let $a > 0$ and $m \in L^{2}(\R), m \geq 0,$ denote the charge density of a finite nuclear cluster, then the corresponding TFW Yukawa energy functional is defined, for $v \in H^{1}(\R)$,
by
\begin{align}
E^{\TFW}_{a}(v,m) = C_{\textnormal{W}} \int_{\R} |\nabla v|^{2} + C_{\textnormal{TF}} \int_{\R} v^{10/3} + \frac{1}{2} D_{a}( m - v^{2}, m - v^{2}). \label{eq:TFW-energy-Yuk-def,with-constants}
\end{align}
The function $v$ corresponds to the positive square root of the electron
density. The first two terms of (\ref{eq:TFW-energy-Yuk-def,with-constants}) model the kinetic energy of the
electrons while the third term models the Coulomb energy. This definition of the
Coulomb energy is only valid for smeared nuclei. The energy (\ref{eq:TFW-energy-Yuk-def,with-constants}) can be rescaled to
ensure that $C_{\textnormal{W}} = C_{\textnormal{TF}} = 1$.

To construct the electronic ground state for an infinite arrangement of
  nuclei (e.g., crystals), it is necessary to restrict admissible nuclear charge densities to
  $m \in L^{1}_{\unif}(\R), m \geq 0$, satisfying
\begin{align*}
\textnormal{(H1) }& \sup_{x \in \R} \int_{B_{1}(x)} m(z) \id z < \infty, \\
\textnormal{(H2) }& \lim_{R \to \infty} \inf_{x \in \R} \frac{1}{R} \int_{B_{R}(x)} m(z) \id z = \infty.
\end{align*} 
The property (H1) guarantees that no clustering of infinitely many nuclei occurs
at any point in space whereas (H2) ensures that there are no large regions that
are devoid of nuclei.

For each $m$ satisfying (H1)--(H2), \cite[Theorem 6.10]{Paper1} guarantees the existence and uniqueness of a ground state $(u,\phi)$ satisfying
\begin{subequations}
\label{eq:u-phi-eq-pair}
\begin{align}
- \Delta u &+ \frac{5}{3} u^{7/3} - \phi u = 0, \label{eq:u-inf-eq} \\
- \Delta \phi &= 4\pi (m - u^{2}), \label{eq:phi-inf-eq}
\end{align}
\end{subequations}
Similarly, as remarked in \cite[Chapter 6]{C/LB/L}, it also follows that for sufficiently small $a > 0$, the existence and uniqueness of the Yukawa ground state $(u_{a},\phi_{a})$, solving
\begin{subequations}
\label{eq: u phi Yuk eq pair}
\begin{align}
&- \Delta u_{a} + \frac{5}{3} u^{7/3}_{a} - \phi_{a} u_{a} = 0, \label{eq:u-Yuk-inf-eq} \\
&- \Delta \phi_{a} + a^{2} \phi_{a} = 4\pi (m - u^{2}_{a}), \label{eq:phi-Yuk-inf-eq}
\end{align}
\end{subequations}
The equation (\ref{eq:phi-inf-eq}) arises from the Coulomb interaction, as $\smfrac{1}{4\pi |\cdot|}$ is the Green's function for the Laplacian on $\R$, while (\ref{eq:phi-Yuk-inf-eq}) is obtained for the Yukawa problem, as $\smfrac{1}{4\pi} Y_{a}$ is the Green's function for $- \Delta + a^{2}$ on $\R$, $a > 0$.

\begin{definition}
In this article, for any nuclear configuration $m$ satisfying \textnormal{(H1)--(H2)}, the ground state corresponding to $m$ refers to the unique solution $(u,\phi)$ to \textnormal{(\ref{eq:u-phi-eq-pair})}. For $a > 0$, the Yukawa ground state corresponding to $m$ refers to the unique solution $(u_{a},\phi_{a})$ to \textnormal{(\ref{eq: u phi Yuk eq pair})}. \qed
\end{definition}

\section{Main Results}
\label{Section3}
\subsection{Regularity estimates}
This section generalises the TFW pointwise stability estimate and its consequences \cite{Paper1} from the Coulomb to the Yukawa setting.

The proofs of the main results in the next section require uniform regularity
estimates for Yukawa systems refining those shown in \cite{C/LB/L}, provided that $a \in (0,a_{0}]$ for some $a_{0} > 0$.

The main regularity estimate (\ref{Proposition - Yukawa Regularity Est All a}) relies on uniform variants of (H1)--(H2), so the class of nuclear configurations $\mathcal{M}_{L^{2}}$, defined in \cite{Paper1}, is used. Given $M, \omega_{0}, \omega_{1} > 0$, let $\omega = (\omega_{0},\omega_{1})$ and define
\begin{align}
\mathcal{M}_{L^{2}}(M,\omega) = \bigg\{ \, m \in L^{2}_{\unif}(\R) \, \bigg| \,\,\, &\| m \|_{L^{2}_{\unif}(\R)} \leq M, \nonumber \\ \, & \forall \, R > 0 \,\, \inf_{x \in \R} \int_{B_{R}(x)} m(z) \id z \geq \omega_{0} R^{3} - \omega_{1} \, \bigg\}. \label{eq:M-L2-def}
\end{align}
As each nuclear distribution $m \in \mathcal{M}_{L^{2}}(M,\omega)$ satisfies (H1)--(H2), \cite[Chapter 6]{C/LB/L} guarantees the existence of corresponding ground states $(u_{a},\phi_{a})$ for sufficiently small $a$. The proof of \cite[Proposition 2.2, Chapter 6]{C/LB/L} is adapted to extend existence and uniqueness of Yukawa ground states to all $a > 0$. In addition, the uniformity in upper and lower bounds on $m \in \mathcal{M}_{L^2}(M, \omega)$ yields regularity estimates and lower bounds on these ground states which are also uniform.

\begin{proposition}
\label{Proposition - Yukawa Regularity Est All a}
Let $a_{0} > 0$ and $m \in \mathcal{M}_{L^{2}}(M,\omega)$, then for any $0 < a \leq a_{0}$
there exists $(u_{a},\phi_{a})$ solving \textnormal{(\ref{eq: u phi Yuk eq pair})}, satisfying $u_{a} \geq 0$ and
\begin{align}
&\| u_{a} \|_{H^{4}_{\unif}(\R)} + \| \phi_{a} \|_{H^{2}_{\unif}(\R)} \leq C(a_{0},M), \label{eq:Yuk-u-H4-phi-H2-unif-est-all-a}
\end{align}
where the constant $C(a_{0},M)$ is increasing in both $a_{0}$ and $M$.
\end{proposition}
\Cref{Proposition - Yukawa Regularity Est All a} can be generalised to obtain existence of Yukawa ground states corresponding to finite nuclear configurations, for sufficiently small $a > 0$. The following result will be used in \Cref{Proposition - Infinite Finite Ground state comparison} to compare the Yukawa ground state with its finite approximation.

\begin{proposition}
\label{Proposition - Yukawa Regularity Est}
For any nuclear distribution $m:\R \to \mathbb{R}_{\geq 0}$, satisfying
\begin{align*}
\| m \|_{L^{2}_{\unif}(\R)} \leq M,
\end{align*}
there exists $a_{0} = a_{0}(m) > 0$ such that for all $0 < a \leq a_{0}$, there exists $(u_{a},\phi_{a})$ solving \textnormal{(\ref{eq: u phi Yuk eq pair})}, satisfying $u_{a} \geq 0$ and
\begin{align}
&\| u_{a} \|_{H^{4}_{\unif}(\R)} + \| \phi_{a} \|_{H^{2}_{\unif}(\R)} \leq C(M). \label{eq:Yuk-u-H4-phi-H2-unif-est}
\end{align}
If $\int_{B_{R_{0}}(x)} m \geq c_{0} > 0$ for some $x \in \R$ and $R_{0}, c_{0}$, then $a_{0} = a_{0}(R_{0},c_{0}) > 0$.
\end{proposition}




\begin{proposition}
\label{Proposition - Uniform Yuk inf u estimate all a}
Let $a_{0} > 0$ and $m \in \mathcal{M}_{L^{2}}(M,\omega)$, then for all $0 < a \leq a_{0}$ the corresponding Yukawa ground state $(u_{a},\phi_{a})$ is unique and there exists $c_{a_{0},M,\omega} > 0$ such that the electron density $u_{a}$ satisfies
\begin{align}
\inf_{x \in \R} u_{a}(x) &\geq c_{a_{0},M,\omega} > 0. \label{eq:u-Yuk-inf-est-all-a}
\end{align}


\end{proposition}

%
Assuming higher regularity of the nuclear distributions implies higher regularity of the ground state. Therefore define for $k \in \mathbb{N}_{0}$
\begin{align*}
\mathcal{M}_{H^{k}}(M,\omega) = \bigg\{ \, m \in H^{k}_{\unif}(\R) \, \bigg| \,\,\, &\| m \|_{H^{k}_{\unif}(\R)} \leq M, \\ \, \, & \forall \, R > 0 \,\, \inf_{x \in \R} \int_{B_{R}(x)} m(z) \id z \geq \omega_{0} R^{3} - \omega_{1} \, \bigg\}.
\end{align*}
Arguing by induction and applying the uniform lower bound (\ref{eq:u-Yuk-inf-est-all-a})
yields the following result.

\begin{corollary}
\label{Corollary - General Est Ck Yuk version}
Let $a_{0} > 0$, $k \in \mathbb{N}_{0}$ and $m \in \mathcal{M}_{H^{k}}(M,\omega)$, then for all $0 < a \leq a_{0}$ the corresponding Yukawa ground state $(u_{a},\phi_{a})$ satisfies
\begin{align}
\| u_{a} \|_{H^{k+4}_{\unif}(\R)} +
\| \phi_{a} \|_{H^{k+2}_{\unif}(\R)} \leq C(a_{0},k,M,\omega). \label{eq:corr-k,u-phi-extra-reg-Yuk-est}
\end{align}
\end{corollary}

\subsection{Uniform Yukawa estimates}
The main result of this article is a uniform estimate 
 comparing the Yukawa and Coulomb ground states corresponding to the same nuclear configuration. This result is essentially a consequence of \cite[Theorems 3.4 and 3.5]{Paper1}.

In the following, $(u,\phi) = (u_{0},\phi_{0})$ denotes the corresponding Coulomb ground state solving (\ref{eq:u-phi-eq-pair}), i.e the ground state with Yukawa parameter $a = 0$.

\begin{theorem}
\label{Theorem - Yukawa Coulomb Comparison}
Suppose $a_{0} > 0$, $k \in \mathbb{N}_{0}$, $m \in \mathcal{M}_{H^{k}}(M,\omega)$ and let $(u, \phi)$
denote the corresponding Coulomb ground state. For $0 < a \leq a_{0}$, let $(u_{a},\phi_{a})$ denote the corresponding Yukawa ground state, then there exists $C = C(a_{0},k,M,\omega) > 0$ such that 
\begin{align}
\| u_{a} - u \|_{W^{k+2,\infty}(\R)} + \| \phi_{a} - \phi \|_{W^{k+2,\infty}(\R)} \leq C a^{2}. \label{eq: Yukawa Coulomb pointwise k pos est}
\end{align}
\end{theorem}

\begin{remark}
The error term in (\ref{eq: Yukawa Coulomb pointwise k pos est}) arises from the additional term in the Yukawa equation (\ref{eq:phi-Yuk-inf-eq}), as opposed to due to a difference in nuclear distributions in \cite[Theorems 3.4 and 3.5]{Paper1}. For this reason, the author believes that an analogous result to \Cref{Theorem - Yukawa Coulomb Comparison} also holds for point charge nuclei. \qed
\end{remark}

\Cref{Theorem - Yukawa Coulomb Comparison} can be generalised to compare two Yukawa ground states $(u_{a_{1}},\phi_{a_{1}})$, $ (u_{a_{2}},\phi_{a_{2}})$ corresponding to the same nuclear configuration, where the parameters $a_{1}, a_{2}$ differ.
\begin{corollary}
\label{Corollary - Yukawa Coulomb Comparison 2}
Let $a_{0} > 0$, $k \in \mathbb{N}_{0}$, $m \in \mathcal{M}_{H^{k}}(M,\omega)$ and suppose $0 < a_{1} \leq a_{2} \leq a_{0}$, then let $(u_{a_{1}}, \phi_{a_{1}}), (u_{a_{2}}, \phi_{a_{2}})$ denote the corresponding Yukawa ground states. There exists $C = C(a_{0},k,M,\omega) > 0$ such that 
\begin{align}
\| u_{a_{1}} - u_{a_{2}} \|_{W^{k+2,\infty}(\R)} + \| \phi_{a_{1}} - \phi_{a_{2}} \|_{W^{k+2,\infty}(\R)} \leq C \left( a_{2}^{2} - a_{1}^{2} \right). \label{eq: Yukawa Yukawa pointwise k pos est}
\end{align}
\end{corollary}

\subsection{Pointwise Yukawa estimates}
Theorems \ref{Theorem - One inf pointwise stability estimate Yuk alt} and \ref{Theorem - Exponential Est Integral Yuk RHS} extend \cite[Theorems 3.4 and 3.5]{Paper1} to the Yukawa model and require the class of test functions
\begin{align}
H_{\gamma} &= \bigg \{ \, \xi \in H^{1}(\R) \, \bigg| \, |\nabla \xi(x)| \leq \gamma |\xi(x)| \, \forall \, x \in \R \, \bigg \} \label{eq:H-gamma-test-space}
\end{align}
for some $\gamma > 0$.
Observe that $e^{-\wt \gamma |\cdot|} \in H_{\gamma}$ for any $0 < \wt \gamma \leq \gamma$.
\begin{theorem}
\label{Theorem - One inf pointwise stability estimate Yuk alt}
Let $m_{1} \in \mathcal{M}_{L^{2}}(M,\omega)$, and let
$m_{2} : \R \to \mathbb{R}_{\geq 0}$ satisfy 
\begin{align*}
\| m_{2} \|_{L^{2}_{\unif}(\R)} \leq M',
\end{align*}
then there exists $a_{1} = a_{1}(\omega,m_{2}) > 0$ such that for all $0 < a \leq a_{1}$ there exist solutions $(u_{1,a},\phi_{1,a})$ and $(u_{2,a},\phi_{2,a})$
to \textnormal{(\ref{eq: u phi Yuk eq pair})} corresponding to $m_{1}, m_{2}$, where $(u_{2,a},\phi_{2,a})$ satisfies
$u_{2,a} \geq 0$ and
\begin{align}
\| u_{2,a} \|_{H^{4}_{\unif}(\R)} &+ \| \phi_{2,a} \|_{H^{2}_{\unif}(\R)} \leq C(M'), \label{eq: u2 phi2 reg est}
\end{align}
independently of $a$. Define
\begin{align*}
w = u_{1,a} - u_{2,a}, \quad \psi = \phi_{1,a} - \phi_{2,a}, \quad R_{m} = 4 \pi (m_{1} - m_{2}),
\end{align*}
then there exist $C = C(M,M',\omega), \gamma = \gamma(M,M',\omega) >0$, such that for any $\xi \in H_{\gamma}$
\begin{align}
\int_{\R} \bigg( \sum_{|\alpha_{1}| \leq 4} |\partial^{\alpha_{1}}w|^{2} + \sum_{|\alpha_{2}| \leq 2} |\partial^{\alpha_{2}} \psi|^{2} \bigg) \xi^{2} \leq C \int_{\R} R_{m} \xi^{2}. \label{eq: w and psi partial xi global onesided est}
\end{align}
In particular, for any $y \in \R$,
\begin{align}
\sum_{|\alpha| \leq 2} |\partial^{\alpha} w(y)|^{2} + |\psi(y)|^{2} \leq C \int_{\R} |R_{m}(x)|^{2} e^{-2\gamma |x - y|} \id x. \label{eq: w and psi pointwise rhs exp integral onesided est}
\end{align}
\end{theorem} 

\Cref{Theorem - One inf pointwise stability estimate Yuk alt} can be generalised to
obtain higher-order pointwise estimates, but this requires that $m_{1}, m_{2} \in \mathcal{M}_{H^{k}}(M,\omega)$ for some
$k \in \mathbb{N}_{0}$ to ensure that \emph{both}
$\inf u_{1}, \inf u_{2} > 0$.

\begin{theorem}
\label{Theorem - Exponential Est Integral Yuk RHS}
Let $a_{0} > 0$, $k \in \mathbb{N}_{0}$, $m_{1}, m_{2} \in \mathcal{M}_{H^{k}}(M,
\omega)$ and for $0 < a \leq a_{0}$, let $(u_{1,a},\phi_{1,a}), (u_{2,a},\phi_{2,a})$ denote the corresponding Yukawa ground states. Define
\begin{align*}
w = u_{1,a} - u_{2,a}, \quad \psi = \phi_{1,a} - \phi_{2,a}, \quad R_{m} = 4 \pi (m_{1} - m_{2}),
\end{align*}
then there exist $C = C(a_{0},k,M,\omega), \gamma = \gamma(a_{0},M,\omega) > 0$, independent of $a$, such that for any $\xi \in H_{\gamma}$
\begin{align}
\int_{\R} \bigg( \sum_{|\alpha_{1}| \leq k+4} |\partial^{\alpha_{1}} w|^{2} + \sum_{|\alpha_{2}| \leq k+2} |\partial^{\alpha_{2}} \psi|^{2} \bigg) \xi^{2} \leq C \int_{\R} \sum_{|\beta| \leq k} |\partial^{\beta} R_{m}|^{2} \xi^{2}. \label{eq: w and psi partial xi global est}
\end{align}
In particular, for any $y \in \R$,
\begin{align}
\sum_{|\alpha_{1}| \leq k+2} |\partial^{\alpha_{1}} w(y)|^{2} + \sum_{|\alpha_{2}| \leq k} |\partial^{\alpha_{2}} \psi(y)|^{2} \leq C \int_{\R} \sum_{|\beta| \leq k} |\partial^{\beta} R_{m}(x)|^{2} e^{-2\gamma |x - y|} \id x. \label{eq: w and psi Yuk pointwise rhs exp integral est}
\end{align}
\end{theorem}

\section{Applications}
\label{Section4}
\subsection{Yukawa and Coulomb forces}

Let $\eta \in C^{\infty}_{\textnormal{c}}(B_{R_{0}}(0))$ be radially symmetric
and satisfy $\eta \geq 0$ and $\int_{\R} \eta = 1$ describe the charge density
of a single (smeared) nucleus, for some fixed $R_{0} > 0$. For any countable
collection of nuclear coordinates
$Y = (Y_{j})_{j \in \mathbb{N}} \in (\R)^\mathbb{N}$, let the corresponding
nuclear configuration be defined by
\begin{align}
  m_{Y}(x) = \sum_{j \in \mathbb{N}} \eta(x - Y_{j}). \label{eq: m Y ass def}
\end{align}
A natural space of nuclear coordinates, related to the $\mathcal{M}_{H^k}$
spaces is
\begin{align}
  \mathcal{Y}_{L^{2}}(M,\omega) := \{ \, Y \in (\R)^\mathbb{N}  \, | \, m_{Y} \in \mathcal{M}_{L^{2}}(M,\omega)  \, \}. \label{eq: Y space def}
\end{align}

For any $Y \in \mathcal{Y}_{L^{2}}(M,\omega)$ and $a > 0$, there exists a unique Yukawa ground state $(u_{a},\phi_{a})$ corresponding to $m = m_{Y}$. Two definitions for the energy density for an infinite system are provided, for bounded $\Omega \subset \R$:
\begin{align}
\int_{\Omega} \mathcal{E}_{1,a}(Y;x) \id x &:= \int_{\Omega} |\nabla u_{a}|^{2} + \int_{\Omega} u_{a}^{10/3} + \frac{1}{2} \int_{\Omega} \phi_{a}(m-u_{a}^{2}), \label{eq: TFW energy 1} \\
\int_{\Omega} \mathcal{E}_{2,a}(Y;x) \id x &:= \int_{\Omega} |\nabla u_{a}|^{2} + \int_{\Omega} u_{a}^{10/3} + \frac{1}{8\pi} \left( \int_{\Omega} |\nabla \phi_{a}|^{2} + a^{2} \int_{\Omega} \phi_{a}^{2} \right),
\label{eq: TFW energy 2}
\end{align}
which satisfy
$\mathcal{E}_{1,a}(Y;\cdot), \mathcal{E}_{2,a}(Y;\cdot) \in L^{1}_{\unif}(\R)$.

Suppose now that $\Omega \subset \R$ is a charge-neutral volume
\cite{YuTrinkleBaderVolumes}, that is, if $n$ is the unit normal to $\partial \Omega$,
then $\nabla \phi_{a} \cdot n = 0$ on $\partial \Omega$. Recall (\ref{eq:phi-Yuk-inf-eq}),
\begin{align*}
- \Delta \phi_{a} + a^{2} \phi_{a} = 4 \pi ( m - u_{a}^{2} ),
\end{align*}
it then follows that
\begin{align*}
\frac{1}{8\pi} \left( \int_{\Omega} |\nabla \phi_{a}|^{2} + a^{2} \int_{\Omega} \phi_{a}^{2} \right) = \frac{1}{8\pi} \int_{\Omega} (- \Delta \phi_{a} + a^{2} \phi_{a}) \phi_{a} + \int_{\partial \Omega} \phi_{a} \nabla \phi_{a} \cdot n = \frac{1}{2} \int_{\Omega} \phi_{a} (m - u^{2}_{a}),
\end{align*}
hence
\begin{align*}
\int_{\Omega} \mathcal{E}_{1,a}(Y;x) \id x = \int_{\Omega} \mathcal{E}_{2,a}(Y;x) \id x.
\end{align*}
Similarly, for finite systems and $\Omega = \R$, the two energies (\ref{eq: TFW energy 1})--(\ref{eq: TFW energy 2}) agree. Thus $\mathcal{E}_{1,a}, \mathcal{E}_{2,a}$ are two energy densities which are well-defined for infinite configurations.

Given $Y \in \mathcal{Y}_{L^{2}}(M,\omega)$, similarly define the Coulomb energy densities $\mathcal{E}_{1}(Y;\cdot)$, $\mathcal{E}_{2}(Y;\cdot) \in L^{1}_{\unif}(\R)$ \cite{Paper1}
\begin{align}
\mathcal{E}_{1}(Y;\cdot) &:= |\nabla u|^{2} + u^{10/3} + \frac{1}{2} \phi(m-u^{2}), \label{eq: E1 def} \\
\mathcal{E}_{2}(Y;\cdot) &:= |\nabla u|^{2} + u^{10/3} + \frac{1}{8\pi} |\nabla \phi|^{2}. \label{eq: E2 def}
\end{align}
By comparing the Yukawa and Coulomb energy densities, (\ref{eq: TFW energy 1})--(\ref{eq: TFW energy 2}) with (\ref{eq: E1 def})--(\ref{eq: E2 def}) respectively, then applying \Cref{Theorem - Yukawa Coulomb Comparison} and \Cref{Proposition - Yukawa Regularity Est} yields the convergence of the energy densities: for all $0 < a \leq a_{0}$
\begin{align}
\left \| \mathcal{E}_{1,a} - \mathcal{E}_{1} \right \|_{L^{2}_{\unif}(\R)} + \| \mathcal{E}_{2,a} - \mathcal{E}_{2} \|_{H^{1}_{\unif}(\R)} \leq C(a_{0},M) a^{2}. \label{eq:Yukawa-Coulomb-energy-density-conv}
\end{align} 
In (\ref{eq:Yukawa-Coulomb-energy-density-conv}), the regularity of the difference $\mathcal{E}_{1,a} - \mathcal{E}_{1}$ is limited by the nuclear distribution $m \in L^{2}_{\unif}(\R)$, whereas this term does not apppear in $\mathcal{E}_{2,a} -\mathcal{E}_{2}$, hence the latter possesses additional regularity.

The next result shows that the force generated by a nucleus converges when passing from the Yukawa to the Coulomb model.
\begin{theorem}
\label{Theorem - Forcing Yuk and Coulomb Comparison}
Let $a_{0} > 0$, $Y \in \mathcal{Y}_{L^{2}}(M,\omega)$ and $i \in \{ 1,2\}$, then for all $0 < a \leq a_{0}$ and $k \in \mathbb{N}$, the Yukawa force density $\partial_{Y_k} \mathcal{E}_{i,a}(Y,\cdot) \in L^{1}(\R)$ exists and satisfies
\begin{align}
\int_{\R} \frac{\partial \mathcal{E}_{1,a}}{\partial Y_{k}}(Y;x) \id x = \int_{\R} \frac{\partial \mathcal{E}_{2,a}}{\partial Y_{k}}(Y;x) \id x = \int_{\R} \phi_{a}(x) \, \frac{\partial m_Y(x)}{\partial Y_k} \id x. \label{eq: Yuk force equivalence}
\end{align}
In addition, the Coulomb force density $\partial_{Y_k} \mathcal{E}_{i}(Y,\cdot) \in L^{1}(\R)$ also exists and there exists $C = C(a_{0},M,\omega) > 0$ such that for all $0 < a \leq a_{0}$
\begin{align}
\left| \int_{\R} \left( \frac{\partial \mathcal{E}_{i,a}}{\partial Y_{k}} - \frac{\partial \mathcal{E}_{i}}{\partial Y_{k}} \right)(Y;x) \id x \right| &\leq C a^{2}. \label{eq: Yuk and Coulomb force comparison est}
\end{align}
\end{theorem}
The expression (\ref{eq: Yuk force equivalence}) shows that the forces generated by the energy densities $\mathcal{E}_{1,a}$ and $\mathcal{E}_{2,a}$ are identical. Also, (\ref{eq: Yuk and Coulomb force comparison est}) establishes an $O(a^{2})$ convergence of forces when passing from the Yukawa to the Coulomb model. 

\subsection{Thermodynamic limit estimates}
\label{Subsection - Thermodynamic Limit Estimates}
The following result extends \cite[Proposition 4.1]{Paper1} to the Yukawa setting, providing an estimate for comparing the infinite Yukawa ground state
with its finite approximation, over compact sets, thus providing explicit rates
of convergence for the thermodynamic limit. This is discussed in
\Cref{rem:rdl-estimate}.

Interpreted differently, the result yields estimates on the decay of the
perturbation from the bulk electronic structure at a domain boundary.

\begin{proposition}
\label{Proposition - Infinite Finite Ground state comparison}
Let $m \in \mathcal{M}_{L^{2}}(M,\omega)$, $\Omega \subset \R$ be open and suppose there exists
$m_{\Omega}: \R \to \mathbb{R}_{\geq 0}$ such that 
$m_{\Omega} = m$ on $\Omega$ and $\| m_{\Omega} \|_{L^{2}_{\unif}(\R)} \leq M$
\textnormal{(}e.g., $m_\Omega = m \chi_\Omega$\textnormal{)}. Then there exists $a_{0} = a_{0}(\omega,m_{\Omega}) > 0$ such that for all $0 < a \leq a_{0}$ there exists a ground state $(u_{a},\phi_{a})$ corresponding to $m$ and $(u_{\Omega,a},\phi_{\Omega,a})$ solving \textnormal{(\ref{eq: u phi Yuk eq pair})} with
$m=m_{\Omega}$, $u_{\Omega,a} \geq 0$ and
$C = C(a_{0},M, \omega), \gamma = \gamma(a_{0},M,\omega) > 0$, independent of $a$ and $\Omega$, such
that for all $y \in \Omega$
\begin{align}
\sum_{|\alpha| \leq 2} |\partial^{\alpha} (u_{a} - u_{\Omega,a})(y)| + |(\phi_{a} - \phi_{\Omega,a})(y)| \leq C e^{-\gamma \textnormal{dist}(y,\partial \Omega)}.  \label{eq: w and psi Rbuf est}
\end{align}
\end{proposition}

\begin{remark}
  \label{rem:rdl-estimate}
  Let $R>0$ and $R_n \uparrow \infty$, then applying \Cref{Proposition -
      Infinite Finite Ground state comparison}, with $\Omega = B_{R_n}(0)$ and
    $m_\Omega = m_{R_n}$ and $0 < a \leq a_{0} = a_{0}(\omega)$ gives a rate of convergence for the finite
    approximation $(u_{a,R_{n}},\phi_{a,R_{n}})$, solving (\ref{eq: u phi Yuk eq pair}),
    to the ground state $(u_{a},\phi_{a})$
  \begin{align}
    \| u_{a} - u_{a,R_{n}} \|_{W^{2,\infty}(B_R(0))} 
    + \| \phi_{a} - \phi_{a,R_{n}} \|_{L^{\infty}(B_R(0))} \leq 
    C e^{- \gamma ( R_{n} - R )}. 
    \label{eq: w and psi Rn local exp conv est}
  \end{align}
  This strengthens the result that $(u_{a,R_{n}},\phi_{a,R_{n}})$ converges to
  $(u_{a},\phi_{a})$ pointwise almost everywhere along a subsequence \cite{C/LB/L}. \qed
\end{remark}

\subsection{Pointwise stability and neutrality estimates}

The following results extend \cite[Corollary 4.2, Theorem 4.3]{Paper1} to the Yukawa model. \Cref{Corollary - Exponential Estimates Yuk Consequences} shows that the decay properties of the nuclear perturbation are inherited by the response of the Yukawa ground state, and \Cref{Corollary - Neutrality Yuk Estimate} shows the neutrality of nuclear perturbations for the TFW equations in the Yukawa setting.

\begin{corollary}
\label{Corollary - Exponential Estimates Yuk Consequences}
Let $a_{0} > 0$, $k \in \mathbb{N}_{0}$, $m_{1}, m_{2} \in \mathcal{M}_{H^{k}}(M, \omega)$ and $0 < a \leq a_{0}$, then let $(u_{1,a},\phi_{1,a}), (u_{2,a},\phi_{2,a})$ denote the corresponding Yukawa ground states and define
\begin{align*}
w = u_{1,a} - u_{2,a}, \quad \psi = \phi_{1,a} - \phi_{2,a}, \quad R_{m} = 4 \pi (m_{1} - m_{2}).
\end{align*}
\begin{enumerate}
\item \textnormal{(Exponential Decay)} If $R_{m} \in H^{k}(\R)$ and
  $\spt(R_{m}) \subset B_{R}(0)$, or there exists $\gamma' > 0$ such that
$\sum_{|\beta| \leq k}| \partial^{\beta} R_{m}(x)| \leq C e^{-\gamma'|x|},$
then there exist $C=C(a_{0},k,M,\omega), \gamma = \gamma(a_{0},M,\omega) > 0$ depending
also on $R$ or $\gamma'$ such that
\begin{align}
\sum_{|\alpha_{1}| \leq k+2} |\partial^{\alpha_{1}} w(x)| + \sum_{|\alpha_{2}| \leq k} |\partial^{\alpha_{2}} \psi(x)| \leq C e^{- \gamma |x|}. \label{eq: w psi local exp est}
\end{align}

\item \textnormal{(Algebraic Decay)} If there exist $C, r > 0$ such that
$\sum_{|\beta| \leq k}| \partial^{\beta} R_{m}(x)| \leq C (1 + |x|)^{-r}$
then there exists $C = C(a_{0},r,k,M,\omega) > 0$ such that
\begin{align}
\sum_{|\alpha_{1}| \leq k+2} |\partial^{\alpha_{1}} w(x)| + \sum_{|\alpha_{2}| \leq k} |\partial^{\alpha_{2}} \psi(x)| &\leq C (1 + |x|)^{-r}. \label{eq: w psi decay est}
\end{align}

\item \textnormal{(Global Estimates)} If $R_{m} \in H^{k}(\R)$ 
then there exists $C = C(a_{0},k,M,\omega) > 0$ such that
\begin{align}
\| w \|_{H^{k+4}(\R)} + \| \psi \|_{H^{k+2}(\R)} \leq C \| R_{m} \|_{H^{k}(\R)}. \label{eq: w psi Sobolev est}
\end{align}
\end{enumerate}
\end{corollary}

\begin{corollary}
\label{Corollary - Neutrality Yuk Estimate}
Let $a_{0} > 0$, $m_{1}, m_{2} \in \mathcal{M}_{L^{2}}(M,\omega)$ and $0 < a \leq a_{0}$, then define $\rho_{12,a} := m_{1} - u_{1,a}^{2} - m_{2} + u_{2,a}^{2}$.
\begin{enumerate}
\item 
  If $\spt(m_{1} -m_{2}) \subset B_{R'}(0)$, or there exist $C, \wt \gamma > 0$ such that
$|(m_{1} - m_{2})(x)| \leq C e^{- \wt \gamma|x|}$,
then $\rho_{12,a} \in L^{1}(\R)$ and there exist $C, \gamma > 0$, independent of $a$, such that, for all $R > 0$, 
\begin{align}
\bigg | \int_{B_{R}(0)} \rho_{12,a} \bigg | &\leq C e^{- \gamma R}. \label{eq: local pert BR exp est}
\end{align}

\item 
  If there exists $C, r > 0$ such that
$|(m_{1}- m_{2})(x)| \leq C (1 + |x|)^{-r}$ 
then there exists $C > 0$, independent of $a$, such that, for all $R > 0$,
\begin{align}
\bigg | \int_{B_{R}(0)} \rho_{12,a} \bigg | &\leq C (1+ R)^{2 - r}. \label{eq: local pert BR alg est}
\end{align}

\item 
  If $m_{1} - m_{2} \in L^{2}(\R)$ \textnormal{(}e.g., $r > 3/2$ in \textnormal{(2)}\textnormal{)} then
  $\rho_{12,a} \in L^{2}(\R)$ and
\begin{align}
\lim_{\varepsilon \to 0} \frac{1}{|B_{\varepsilon}(0)|} \int_{B_{\varepsilon}(0)} \widehat{\rho}_{12,a}(k) \id k = 0, \label{eq: global neutrality Leb point eq}
\end{align}
where $\widehat{\rho}_{12,a}$ denotes the Fourier transform of $\rho_{12,a}$.
\end{enumerate}
\end{corollary}

\section{Proofs}
\label{Section5}

The following technical lemma is used in \Cref{Proposition - Finite Yukawa Regularity Est All a} to show $u_{a,R_{n}} > 0$ but will also be useful to show a uniform lower bound for the ground state electron density $u_{a}$ in \Cref{Lemma - Yukawa Finite Lower Bound} in the Appendix.

\begin{lemma}
\label{Lemma - Technical Lemma 2}
Let $0 < a_{1} \leq a_{2}$ and $m \in \mathcal{M}_{L^{2}}(M,\omega)$, then for $R_{n} > 0$ define $\psi_{R_{n}} \in C^{\infty}_{\rm c}(B_{4R_{n}}(0))$ satisfying $\psi_{R_{n}} \geq 0$ and $\psi_{R_{n}} = 1$ on $B_{2R_{n}}(0)$ and $m_{R_{n}} = m \cdot \chi_{B_{R_{n}}(0)}$. Then there exists $C_{0} = C_{0}(a_{1},a_{2},\omega) > 0$ and $R_{0} = R_{0}(a_{1},a_{2},\omega) > 0$ such that for all $a_{1} \leq a \leq a_{2}$ and $R_{n} \geq R_{0}$ 
\begin{align}
\int_{\R} |\nabla \psi_{R_{n}}|^{2} - D_{a}(m_{R_{n}}, \psi_{R_{n}}^{2}) \leq - C_{0} R_{n}^{3}. \label{eq: TL2}
\end{align}
\end{lemma}

\begin{proof}[Proof of \Cref{Lemma - Technical Lemma 2}]
Let $a_{1} \leq a \leq a_{2}$. By the construction of $\psi_{R_{n}}$
\begin{align}
\int |\nabla \psi_{R_{n}}|^{2} = \int_{B_{4R_{n}}(0) \smallsetminus B_{2R_{n}}(0)} |\nabla \psi_{R_{n}}|^{2} \leq C \int_{B_{4R_{n}}(0) \smallsetminus B_{2R_{n}}(0)} R_{n}^{-2} \leq C_{1} R_{n}. \label{eq: all a grad est}
\end{align}
Additionally, it follows that
\begin{align}
D_{a}( m_{R_{n}}, \psi_{R_{n}}^{2}) &= \int_{\R} \left( m_{R_{n}} * Y_{a} \right) \psi_{R_{n}}^{2} \geq
\int_{B_{2R_{n}}(0)} \left( m_{R_{n}} * Y_{a} \right)(x) \id x \nonumber \\
&= \int_{\R} \left( \int_{B_{2R_{n}}(0) \cap B_{R_{n}}(y)} m_{R_{n}}(x-y) \id x \right) \frac{e^{-a|y|}}{|y|} \id y \nonumber \\
&= \int_{\R} \left( \int_{B_{2R_{n}}(-y) \cap B_{R_{n}}(0)} m_{R_{n}}(x) \id x \right) \frac{e^{-a|y|}}{|y|} \id y. \label{eq: m kj all a conv}
\end{align}
First consider for $R' > 0$
\begin{align*}
\int_{B_{R'}(0)} \frac{e^{-a|y|}}{|y|} \id y = 4 \pi \int_{0}^{R'} r e^{-ar} \id r = \frac{4 \pi}{a^{2}} \left( 1 - e^{-aR'} (1 + a R')  \right),
\end{align*}
hence choosing $R' = (4a)^{-1}$ ensures that
\begin{align}
\int_{B_{1/4a}(0)} \frac{e^{-a|y|}}{|y|} \id y = \frac{4 \pi}{a^{2}} \left( 1 - \smfrac{5}{4}e^{-1/4} \right) =: C_{2}a^{-2}, \label{eq: Yuk exp all a int est}
\end{align}
where $C_{2} > 0$. Now choose $R_{n} \geq (4a)^{-1}$, then the triangle inequality implies for $|y| \leq (4a)^{-1}$, $B_{2R_{n}}(-y) \supset B_{R_{n}}(0)$, hence as $m \in \mathcal{M}_{L^{2}}(M,\omega)$
\begin{align}
\int_{B_{2R_{n}}(-y) \cap B_{R_{n}}(0)} m_{R_{n}}(x) \id x \geq \int_{B_{R_{n}}(0)} m(x) \id x \geq \omega_{0} R_{n}^{3} - \omega_{1}. \label{eq: m kj all a int lower bound}
\end{align}
Combining the inequalities (\ref{eq: m kj all a conv})--(\ref{eq: m kj all a int lower bound}) gives
\begin{align}
D_{a}( m_{R_{n}}, \psi_{R_{n}}^{2}) &= \int_{\R} \left( \int_{B_{2R_{n}}(-y) \cap B_{R_{n}}(0)} m_{R_{n}}(x) \id x \right) \frac{e^{-a|y|}}{|y|} \id y \nonumber \\
&\geq \int_{B_{1/4a}(0)} \left( \int_{B_{2R_{n}}(-y) \cap B_{R_{n}}(0)} m_{R_{n}}(x) \id x \right) \frac{e^{-a|y|}}{|y|} \id y \nonumber \\
&\geq \int_{B_{1/4a}(0)} \left( \int_{B_{R_{n}}(0)} m_{R_{n}}(x) \id x \right) \frac{e^{-a|y|}}{|y|} \id y \geq C_{2} a^{-2} (\omega_{0} R_{n}^{3} - \omega_{1}). \label{eq: all a neg est}
\end{align}
Now define $C_{0} = \smfrac{C_{2} \omega_{0}}{2a_{2}^{2}} > 0$ and $R_{n} \geq R_{0} := \max\{1, (4a_{1})^{-1}, (\smfrac{C_{1} + C_{2} \omega_{1} a_{1}^{-2}}{C_{0}} )^{1/2} \}$, then combining (\ref{eq: all a grad est}) and (\ref{eq: all a neg est}) yields the desired estimate (\ref{eq: TL2}) for any $a_{1} \leq a \leq a_{2}$ and $R_{n} \geq R_{0}$
\begin{align*}
\int |\nabla \psi_{R_{n}}|^{2} - D_{a}( m_{R_{n}}, \psi_{R_{n}}^{2}) &\leq \left( C_{1} R_{n} + C_{2} \omega_{1} a^{-2} \right) - 2 C_{0} R_{n}^{3} \nonumber \\
&\leq C_{0} R_{n}^{3} - 2C_{0} R_{n}^{3} = - C_{0}R_{n}^{3}. \qedhere
\end{align*}

\end{proof}

\subsection{Proof of regularity estimates}

\begin{proposition}
\label{Proposition - Finite Yukawa Regularity Est} 
Let $m : \mathbb{R} \to \mathbb{R}_{\geq 0}$ satisfy 
\begin{align*}
\| m \|_{L^{2}_{\unif}(\R)} \leq M,
\end{align*}
and $R_{n} \uparrow \infty$, then define the truncated nuclear distribution $m_{R_{n}} = m \cdot \chi_{B_{R_{n}}(0)}$. 
%
%
There exists $R_{0} = R_{0}(m), a_{0} = a_{0}(m) > 0$ such that for all $R_{n} \geq R_{0}$ and $0 < a \leq a_{0}$, the unique solution to the minimisation problem
\begin{align}
I^{\TFW}_{a}(m_{R_{n}}) = \inf \left\{ \, E^{\TFW}_{a}(v,m_{R_{n}})  \, \bigg| \, \nabla v \in L^{2}(\R), v \in L^{10/3}(\R), v \geq 0 \, \right \} \label{eq: I a Rn min problem}
\end{align}
yields a unique solution $(u_{a,R_{n}},\phi_{a,R_{n}})$ to
\begin{subequations}
\label{eq: u phi finite Yuk pair}
\begin{align}
&- \Delta u_{a,R_{n}} + \frac{5}{3} u_{a,R_{n}}^{7/3} - \phi_{a,R_{n}} u_{a,R_{n}} = 0, \label{eq: u a n eq} \\
&- \Delta \phi_{a,R_{n}} + a^{2}\phi_{a,R_{n}} = 4 \pi \left( m_{R_{n}} - u_{a,R_{n}}^{2} \right). \label{eq: phi a n eq}
\end{align}
\end{subequations}
which satisfy the following estimates, with constants independent of $R_{n}$\textnormal{:}
\begin{align}
\| u_{a,R_{n}} \|_{H^{4}_{\unif}(\R)} &\leq C(M), \label{eq: uRn a H4 unif est} \\
\| \phi_{a,R_{n}} \|_{H^{2}_{\unif}(\R)} &\leq C(M), \label{eq: phiRn a H2 unif est}
\end{align}
and $u_{a,R_{n}} > 0$ on $\R$ 
whenever
$m_{R_{n}} \not\equiv 0$. In particular, if $\int_{B_{R_{0}}(x)} m \geq c_{0} > 0$ for some $x \in \R$ and $R_{0}, c_{0} > 0$, then $a_{0} = a_{0}(R_{0},c_{0}) > 0$.
\end{proposition}

In the case $m \in \mathcal{M}_{L^{2}}(M,\omega)$, \Cref{Proposition - Finite Yukawa Regularity Est} can be extended to all $a > 0$. The following result will be used to prove \Cref{Proposition - Yukawa Regularity Est All a}.

\begin{proposition}
\label{Proposition - Finite Yukawa Regularity Est All a} 
Let $a_{0} > 0$, $m \in \mathcal{M}_{L^{2}}(M,\omega)$ and $R_{n} \uparrow \infty$, then define $m_{R_{n}} := m \cdot \chi_{B_{R_{n}}(0)}$. 
There exists $R_{0} = R_{0}(a_{0},\omega) > 0$ such that for all $0 < a \leq a_{0}$ and $R_{n} \geq R_{0}$, the minimisation problem \textnormal{(\ref{eq: I a Rn min problem})} yields a unique solution $(u_{a,R_{n}},\phi_{a,R_{n}})$ to \textnormal{(\ref{eq: u phi finite Yuk pair})}
which satisfy the following estimates, with constants independent of $a$ and $R_{n}$\textnormal{:}
\begin{align}
\| u_{a,R_{n}} \|_{H^{4}_{\unif}(\R)} &\leq C(a_{0},M), \label{eq: uRn all a H4 unif est} \\
\| \phi_{a,R_{n}} \|_{H^{2}_{\unif}(\R)} &\leq C(a_{0},M). \label{eq: phiRn all a H2 unif est}
\end{align}
\end{proposition}

\begin{remark}
The Coulomb minimisation problem \cite[Proposition 6.3]{Paper1} imposes a charge neutrality condition. Imposing a neutrality condition for the finite Yukawa problem introduces a Lagrange multiplier into (\ref{eq: u phi finite Yuk pair}) that weakens Theorem \ref{Theorem - Yukawa Coulomb Comparison} significantly. \qed
\end{remark}

The proof of \Cref{Proposition - Finite Yukawa Regularity Est} largely follows the proof of \cite[Proposition 6.3]{Paper1}. \Cref{Proposition - Finite Yukawa Regularity Est} is proved in four steps.

In Step 1, the minimisation problem (\ref{eq: I a Rn min problem}) is shown to be well-posed and defines a unique solution $(u_{a,R_{n}},\phi_{a,R_{n}})$ to (\ref{eq: u phi finite Yuk pair}), where $u_{a,R_{n}},\phi_{a,R_{n}}$ are continuous and decay at infinity.
The argument in Step 2 adapts the Solovej estimate for Yukawa systems to show: there exists $C_{S} > 0$ that for all $m \in \mathcal{M}_{L^{2}}(M,\omega)$ and $a, R_{n} > 0$
\begin{align}
\smfrac{10}{9} u_{a,R_{n}}^{4/3} &\leq \phi_{a,R_{n}} + C_{S} + a^{2}. \label{eq: Solovej Yuk est}
\end{align}
The aim of Step 3 is to show that there exists $a_{0} = a_{0}(\omega), R_{0} = R_{0}(\omega) > 0$ such that for all $0 < a \leq a_{0} \leq 1$ and $R_{n} \geq R_{0}$ 
\begin{align*}
u_{a,R_{n}} > 0 \, \text{ on } \, \R.
\end{align*}

Finally, in Step 4, the following estimate is established
\begin{align}
\| u_{R_{n}} \|_{L^{\infty}(\R)}^{4/3} + \| \phi_{R_{n}} \|_{L^{\infty}(\R)} \leq C(M) + a^{2} \leq C(M) + 1, \label{eq: u phi a Rn L inf est}
\end{align}
where the final constant is independent of $a$, $a_{0}$ and $R_{n}$. The desired estimates (\ref{eq: uRn a H4 unif est})-(\ref{eq: phiRn a H2 unif est}) then follow from standard elliptic regularity.

%

\begin{proof}[Proof of \Cref{Proposition - Finite Yukawa Regularity Est}]
If $m \equiv 0$, then for all $a > 0$ and $R_{n}$, clearly $u_{a,R_{n}} = \phi_{a,R_{n}} = m_{R_{n}} = 0$ satisfies (\ref{eq: u phi finite Yuk pair}) and (\ref{eq: uRn a H4 unif est})--(\ref{eq: phiRn a H2 unif est}).

If $m \not\equiv 0$, then $\int_{B_{R_{0}}(x)} m \geq c_{0} > 0$ for some $x \in \R$ and $R_{0}, c_{0} > 0$. Without loss of generality suppose $x = 0$ otherwise translate $m$.

\emph{Step 1}
For each $n \in \mathbb{N}$ define
\begin{align*}
m_{R_{n}}(x) &= m(x) \cdot \chi_{B_{R_{n}}}(x),
\end{align*}
and choosing $R_{n} \geq R_{0}$ ensures that $\int_{\R} m_{R_{n}} \geq c_{0} > 0$, hence $m_{R_{n}} \not\equiv 0$. Recall
\begin{align*}
E^{\TFW}_{a} (v,m_{R_{n}}) &= \int |\nabla v|^{2} + \int v^{10/3} + \frac{1}{2} D_{a}( m_{R_{n}} - v^{2}, m_{R_{n}} - v^{2}) \geq 0.
\end{align*}
For each $R_{n}$ and $a > 0$, recall the minimisation problem (\ref{eq: I a Rn min problem})
\begin{align*}
I^{\TFW}_{a}(m_{R_{n}}) = \inf \left\{ \, E^{\TFW}_{a}(v,m_{R_{n}})  \, \bigg| \, \nabla v \in L^{2}(\R), v \in L^{10/3}(\R), v \geq 0 \, \right \}.
\end{align*}
By the Gagliardo--Nirenberg--Sobolev embedding \cite{Evans}, $v \in L^{6}(\R)$ and $\| v \|_{L^{6}(\R)} \leq C \| \nabla v \|_{L^{2}(\R)}$, moreover $v \in L^{p}(\R)$ for $p \in [10/3,6]$. Consequently
\begin{align*}
0 \leq D_{a}(v^{2},v^{2}) \leq \| Y_{a} \|_{L^{1}(\R)} \| v \|^{4}_{L^{4}(\R)} \leq C \| v \|_{L^{10/3}(\R)}^{5/2} \| v \|_{L^{6}(\R)}^{3/2} \leq C \| v \|_{L^{10/3}(\R)}^{5/2} \| \nabla v \|_{L^{2}(\R)}^{3/2}.
\end{align*}
Observe that there are no charge constraints on the electron density as in general $v \not\in L^{2}(\R)$. This is chosen to ensure that no Lagrange multipliers appear in (\ref{eq: u phi finite Yuk pair}).

As $m_{R_{n}} \in L^{p_{1}}(\R), Y_{a} \in L^{p_{2}}(\R)$  for all $p_{1} \in [1,2], p_{2} \in [1,3)$, applying Young's inequality yields
\begin{align*}
D_{a}(m_{R_{n}},v^{2}) \leq \| Y_{a} \|_{L^{5/2}(\R)} \| m_{R_{n}} \|_{L^{1}(\R)} \| v^{2} \|_{L^{5/3}(\R)} \leq C \| v \|_{L^{10/3}(\R)}^{2} \leq C + \frac{1}{2} \| v \|_{L^{10/3}(\R)}^{10/3},
\end{align*}
it follows that
\begin{align*}
E^{\TFW}_{a}(v,m_{R_{n}}) \geq \frac{1}{2} \left( \| \nabla v \|_{L^{2}(\R)}^{2} + \| v \|_{L^{10/3}(\R)}^{10/3} + D_{a}(v^{2},v^{2}) \right) + \frac{1}{2} D_{a}(m_{R_{n}},m_{R_{n}}) - C.
\end{align*}
As the energy is bounded below, there exists a minimising sequence $v_{k}$ satisfying
\begin{align*}
\| \nabla v \|_{L^{2}(\R)}^{2} + \| v \|_{L^{10/3}(\R)}^{10/3} + D_{a}(v^{2},v^{2}) \leq C, 
\end{align*}
hence there exists $u_{a,R_{n}}$ such that $\nabla u_{a,R_{n}} \in L^{2}(\R), u_{a,R_{n}} \in L^{10/3}(\R)$. Moreover, along a subsequence $\nabla v_{k}$ converges to $\nabla u_{a,R_{n}}$ weakly in $L^{2}(\R)$, $v_{k}$ converges to $u_{a,R_{n}}$, weakly in $L^{6}(\R)$ and $L^{10/3}(\R)$, strongly in $L^{p}(B_{R}(0))$ for all $p \in [1,6)$ and $R > 0$ and pointwise almost everywhere. Consequently,
\begin{align*}
E^{\TFW}_{a}(u_{a,R_{n}},m_{R_{n}}) \leq \liminf_{k \to \infty} E^{\TFW}_{a}(v_{k},m_{R_{n}}) = I^{\TFW}_{a}(m_{R_{n}}),
\end{align*}
hence $u_{a,R_{n}}$ is a minimiser of (\ref{eq: I a Rn min problem}). Define the alternate minimisation problem
\begin{align}
\inf \left\{ \, E^{\TFW}_{a}(\sqrt \rho,m_{R_{n}})  \, \bigg| \, \nabla \sqrt \rho \in L^{2}(\R), \rho \in L^{5/3}(\R), \rho \geq 0 \, \right \}. \label{eq: rho min problem}
\end{align}
Due to the strict convexity of $\rho \mapsto E^{\TFW}_{a}(\sqrt \rho,m_{R_{n}})$, it follows that $\rho_{a,R_{n}} = u^{2}_{a,R_{n}}$ is the unique minimiser of (\ref{eq: rho min problem}), hence $u_{a,R_{n}}$ is the unique minimiser of (\ref{eq: I a Rn min problem}).

Define
\begin{align}
\phi_{a,R_{n}} = \left( m_{R_{n}} - u_{a,R_{n}}^{2} \right) * Y_{a}, \label{eq: phi a Rn def}
\end{align}
then it follows that $(u_{a,R_{n}},\phi_{a,R_{n}})$ is the unique distributional solution to (\ref{eq: u phi finite Yuk pair})
\begin{align*}
&- \Delta u_{a,R_{n}} + \frac{5}{3} u_{a,R_{n}}^{7/3} - \phi_{a,R_{n}} u_{a,R_{n}} = 0, \\
&- \Delta \phi_{a,R_{n}} + a^{2}\phi_{a,R_{n}} = 4 \pi \left( m_{R_{n}} - u_{a,R_{n}}^{2} \right).
\end{align*}
Moreover, as $m_{R_{n}} - u_{a,R_{n}}^{2} \in L^{2}(\R)$ and the Fourier transform of $Y_{a}$, $\widehat{Y_{a}}$, satisfies
\begin{align*}
\widehat{Y_{a}}(k) = \frac{1}{a^{2} + |k|^{2}},
\end{align*}
it follows that 
\begin{align*}
\int_{\R} |\widehat{\phi_{a,R_{n}}}(k)|^{2} (a^{2} + |k|^{2}) \id k &= \int_{\R} |\widehat{ \left(m_{R_{n}} - u_{a,R_{n}}^{2}\right)}(k)|^{2} | \widehat{ Y_{a} }(k)|^{2} (a^{2} + |k|^{2}) \id k \\
&= \int_{\R} \frac{ |\widehat{ \left(m_{R_{n}} - u_{a,R_{n}}^{2}\right)}(k)|^{2} }{ (a^{2} + |k|^{2}) } \id k \\
&= \int_{\R} \left( \left( m_{R_{n}} - u_{a,R_{n}}^{2} \right) * Y_{a} \right) \left( m_{R_{n}} - u_{a,R_{n}}^{2} \right) \\ &= D_{a}(m_{R_{n}} - u_{a,R_{n}}^{2},m_{R_{n}} - u_{a,R_{n}}^{2}).
\end{align*}
It follows that $\phi_{a,R_{n}} \in H^{1}(\R)$ and
\begin{align}
\int_{\R} |\nabla \phi_{a,R_{n}}|^{2} + a^{2} \int_{\R} \phi_{a,R_{n}}^{2} = D_{a}(m_{R_{n}} - u_{a,R_{n}}^{2},m_{R_{n}} - u_{a,R_{n}}^{2}). \label{eq: phi a Rn H1 est}
\end{align}
Additionally, by applying Young's inequality yields
\begin{align*}
\| \phi_{a,R_{n}} \|_{L^{\infty}(\R)} &\leq \| m_{R_{n}} \|_{L^{2}(\R)} \| Y_{a} \|_{L^{2}(\R)} + \leq \| u_{a,R_{n}}^{2} \|_{L^{3}(\R)} \| Y_{a} \|_{L^{3/2}(\R)} \\ &\leq \| m_{R_{n}} \|_{L^{2}(\R)} \| Y_{a} \|_{L^{2}(\R)} + \leq \| u_{a,R_{n}} \|_{L^{6}(\R)}^{2} \| Y_{a} \|_{L^{3/2}(\R)},
\end{align*}
hence by \cite[Lemma II.25]{Lieb/Simon_TF}, $\phi_{a,R_{n}}$ is a bounded, continuous function that decays uniformly at infinity. In addition, as $m_{R_{n}} \in L^{p}(\R)$ for all $p \in [1,2]$, $Y_{a} \in L^{1}(\R)$ and $u_{a,R_{n}} \in L^{10/3}(\R)$, it follows that
\begin{align*}
\| \phi_{a,R_{n}} \|_{L^{5/3}(\R)} &\leq \| m_{R_{n}} - u_{a,R_{n}}^{2} \|_{L^{5/3}}(\R) \| Y_{a} \|_{L^{1}(\R)} \\ &\leq C \left( \| m_{R_{n}} \|_{L^{5/3}}(\R) + \| u_{a,R_{n}}^{2} \|_{L^{5/3}(\R)} \right) \\
&\leq C \left( \| m_{R_{n}} \|_{L^{5/3}}(\R) + \| u_{a,R_{n}} \|_{L^{10/3}(\R)}^{2} \right).
\end{align*}
%
To bound $u_{a,R_{n}}$ above, recall that $u_{a,R_{n}}$ solves
\begin{align}
- \Delta u_{a,R_{n}} &=  -\frac{5}{3} u_{a,R_{n}}^{7/3} + \phi_{a,R_{n}} u_{a,R_{n}}, \label{eq: u a Rn eq}
\end{align}
and $u_{a,R_{n}} \in L^{10/3}(\R) \cap L^{6}(\R), \phi_{a,R_{n}} \in L^{5/3}(\R) \cap L^{\infty}(\R)$. It follows that the right-hand side of (\ref{eq: u a Rn eq}) belongs to $L^{2}(\R)$ and
\begin{align*}
\| -\smfrac{5}{3} u_{a,R_{n}}^{7/3} + \phi_{a,R_{n}} u_{a,R_{n}} \|_{L^{2}(\R)} &\leq \smfrac{5}{3} \| u_{a,R_{n}}^{7/3} \|_{L^{2}(\R)} + \| \phi_{a,R_{n}} u_{a,R_{n}} \|_{L^{2}(\R)} \\
&\leq \smfrac{5}{3} \| u_{a,R_{n}} \|_{L^{14/3}(\R)}^{7/3} + \| \phi_{a,R_{n}} \|_{L^{5}(\R)} \| u_{a,R_{n}} \|_{L^{10/3}(\R)} \\
&\leq \smfrac{5}{3} \| u_{a,R_{n}} \|_{L^{10/3}(\R)}^{5/6} \| u_{a,R_{n}} \|_{L^{6}(\R)}^{3/2} + \| \phi_{a,R_{n}} \|_{L^{5}(\R)} \| u_{a,R_{n}} \|_{L^{10/3}(\R)}.
\end{align*}
Then for any $x \in \R$ applying the elliptic regularity estimate \cite{Evans} yields
\begin{align*}
\| u_{a,R_{n}} \|_{H^{2}(B_{1}(x))} &\leq C ( \| \smfrac{5}{3} u_{a,R_{n}}^{7/3} - \phi_{a,R_{n}} u_{a,R_{n}} \|_{L^{2}(B_{2}(x))} + \| u_{a,R_{n}} \|_{L^{2}(B_{2}(x))} ) \nonumber \\
&\leq C ( \| \smfrac{5}{3} u_{a,R_{n}}^{7/3} - \phi_{a,R_{n}} u_{a,R_{n}} \|_{L^{2}(\R)} + \| u_{a,R_{n}} \|_{L^{10/3}(B_{2}(x))} ) \\
&\leq C ( \| \smfrac{5}{3} u_{a,R_{n}}^{7/3} - \phi_{a,R_{n}} u_{a,R_{n}} \|_{L^{2}(\R)} + \| u_{a,R_{n}} \|_{L^{10/3}(\R)} ),
\end{align*}
where the constant is independent of $x \in \R$. The Sobolev embedding $H^{2}(B_{1}(x)) \hookrightarrow C^{0,1/2}(B_{1}(x))$ implies that $u_{a,R_{n}}$ is continuous and bounded as 
\begin{align*}
\| u_{a,R_{n}} \|_{L^{\infty}(B_{1}(x))} \leq \| u_{a,R_{n}} \|_{C^{0,1/2}(B_{1}(x))} \leq C \| u_{a,R_{n}} \|_{H^{2}(B_{1}(x))},
\end{align*}
hence
\begin{align}
\| u_{a,R_{n}} \|_{L^{\infty}(\R)} = \sup_{x \in \R} \| u_{a,R_{n}} \|_{L^{\infty}(B_{1}(x))} \leq \sup_{x \in \R} C \| u_{a,R_{n}} \|_{H^{2}(B_{1}(x))} < \infty.
\end{align}
It remains to show that $u_{a,R_{n}}$ decays at infinity. Recall that $u_{a,R_{n}}$ solves (\ref{eq: u a Rn eq})
\begin{align*}
- \Delta u_{R_{n}} &=  -\frac{5}{3} u_{R_{n}}^{7/3} + \phi_{R_{n}} u_{R_{n}}
\end{align*}
and also that $u_{a,R_{n}} \in L^{10/3}(\R) \cap L^{\infty}(\R), \phi_{a,R_{n}} \in L^{5/3}(\R) \cap L^{\infty}(\R)$. Define
\begin{align}
g_{a,R_{n}} := \left( -\frac{5}{3} u_{a,R_{n}}^{7/3} + \phi_{a,R_{n}} u_{a,R_{n}} \right) * \smcb. \label{eq: g a Rn def}
\end{align}
Observe that $u_{a,R_{n}}^{7/3} \in L^{10/7}(\R) \cap L^{\infty}(\R)$ and applying \Holder's inequality gives
\begin{align*}
\| \phi_{a,R_{n}} u_{a,R_{n}} \|_{L^{10/9}(\R)} &\leq \| \phi_{a,R_{n}} \|_{L^{5/3}(\R)} \| u_{a,R_{n}} \|_{L^{10/3}(\R)},
\end{align*}
hence $\phi_{a,R_{n}} u_{a,R_{n}} \in L^{10/9}(\R) \cap L^{\infty}(\R)$. It follows that $-\smfrac{5}{3} u_{R_{n}}^{7/3} + \phi_{R_{n}} u_{R_{n}} \in L^{10/7}(\R) \cap L^{\infty}(\R)$. Decompose
\begin{align*}
g_{a,R_{n}} &= \left( -\frac{5}{3} u_{a,R_{n}}^{7/3} + \phi_{a,R_{n}} u_{a,R_{n}} \right) * \left( \smcb \chi_{B_{1}(0)} \right) + \left( -\frac{5}{3} u_{a,R_{n}}^{7/3} + \phi_{a,R_{n}} u_{a,R_{n}} \right) * \left( \smcb \chi_{B_{1}(0)^{c}} \right),
\end{align*}
then as $\smcb \chi_{B_{1}(0)} \in L^{p_{1}}(\R)$ for all $p_{1} \in [1,3)$, $\smcb \chi_{B_{1}^{c}(0)} \in L^{p_{2}}(\R)$ for all $p_{2} \in (3,\infty]$ applying Young's inequality yields
\begin{align*}
\left \| g_{a,R_{n}} \right \|_{L^{\infty}(\R)} &\leq \| \smfrac{5}{3} u_{a,R_{n}}^{7/3} - \phi_{a,R_{n}} u_{a,R_{n}} \|_{L^{2}(\R)} \left \| \smcb \chi_{B_{1}(0)} \right \|_{L^{2}(\R)} \\ & \quad + \| \smfrac{5}{3} u_{a,R_{n}}^{7/3} - \phi_{a,R_{n}} u_{a,R_{n}} \|_{L^{10/7}(\R)} \left \| \smcb \chi_{B_{1}(0)^{c}} \right \|_{L^{10/3}(\R)},
\end{align*}
hence \cite[Lemma II.25]{Lieb/Simon_TF} implies that $g_{a,R_{n}}$ is a continuous, bounded function vanishing at infinity. In addition, $g_{a,R_{n}}$ solves
\begin{align}
- \Delta g_{a,R_{n}} = -\frac{5}{3} u_{a,R_{n}}^{7/3} + \phi_{a,R_{n}} u_{a,R_{n}} \label{eq: g a Rn dist solution}
\end{align}
in distribution. Combining (\ref{eq: u a Rn eq}) and (\ref{eq: g a Rn dist solution}), it follows that
\begin{align*}
- \Delta ( u_{a,R_{n}} - g_{a,R_{n}}) = 0,
\end{align*}
in distribution, so by Weyl's Lemma $u_{a,R_{n}} - g_{a,R_{n}}$ is harmonic \cite{Gilbarg/Trudinger}. As $u_{a,R_{n}} - g_{a,R_{n}} \in L^{\infty}(\R)$, Liouville's Theorem implies $u_{a,R_{n}} - g_{a,R_{n}}$ is constant \cite{Gilbarg/Trudinger}. Suppose that $u_{a,R_{n}} - g_{a,R_{n}} = c \neq 0$, then as $g_{a,R_{n}}$ decays at infinity
\begin{align*}
\lim_{x \to \infty} u_{a,R_{n}}(x) = c \neq 0,
\end{align*}
which contradicts $u_{a,R_{n}} \in L^{10/3}(\R)$. It follows that $u_{a,R_{n}} = g_{a,R_{n}}$ hence $u_{a,R_{n}}$ decays uniformly at infinity.

\emph{Step 2}
The argument in \cite{Solovej_Universality} is now adapted to show the Solovej estimate for Yukawa systems (\ref{eq: Solovej Yuk est})
\begin{align*}
\smfrac{10}{9} u_{a,R_{n}}^{4/3} &\leq \phi_{a,R_{n}} + C_{S} + a^{2}.
\end{align*}
For convenience, in the following argument $u_{a,R_{n}},\phi_{a,R_{n}},m_{a,R_{n}}$ will be denoted as $u,\phi,m$. As $u$ solves (\ref{eq: u a n eq})
\begin{align*}
- \Delta u + \frac{5}{3} u^{7/3} - \phi u = 0,
\end{align*} 
following the proof of \cite[Proposition 8]{Solovej_Universality}, $w = u^{4/3}$ is non-negative and satisfies
\begin{align}
- \Delta w + \frac{4}{3} \left( \smfrac{5}{3}w - \phi \right)w \leq 0. \label{eq: w Yuk Sol est}
\end{align}
Let $\lambda \in (0,\smfrac{5}{3})$ and define
\begin{align*}
v(x) = \lambda u^{4/3} - \phi - ( C(\lambda) + a^{2} ),
\end{align*}
where $C(\lambda) = (9/4)\pi^{2} \lambda^{-2}(\smfrac{5}{3} - \lambda)^{-1} > 0$. The expression (\ref{eq: phi a n eq}) can be written as
\begin{align}
- \Delta \phi + a^{2} \phi = 4\pi ( m - w^{3/2}). \label{eq: phi Yuk est}
\end{align}
Combining (\ref{eq: w Yuk Sol est}) and (\ref{eq: phi Yuk est}), it follows that
\begin{align*}
\Delta v(x) \geq \frac{4 \lambda}{3} \left( \smfrac{5}{3}w - \phi  \right) w - 4\pi w^{3/2} + 4\pi m - a^{2} \phi.
\end{align*}
The aim is to prove that $v \leq 0$ by showing that $S = \{ x \, | \, v(x) > 0 \}$ is empty. As $u,\phi$ are continuous functions decaying at infinity, it follows that $v$ is continuous, $S$ is bounded, open and $v = 0$ on $\partial S$. Over $S$,
\begin{align*}
\Delta v &\geq \frac{4 \lambda}{3} \left( v + \smfrac{5}{3} w - \lambda w + (C(\lambda) + a^{2} ) \right) w - 4\pi w^{3/2} + 4\pi m - a^{2} \phi \\
&\geq \frac{4 \lambda}{3} \left( \smfrac{5}{3} w - \lambda w + C(\lambda) + a^{2} \right) w - 4\pi w^{3/2} + 4\pi m - a^{2} \phi \\
&= \left( \frac{4 \lambda (\smfrac{5}{3} - \lambda)}{3} w - 4\pi w^{1/2} + \frac{4 \lambda}{3} C(\lambda) \right) w + \frac{4 \lambda}{3} a^{2} w  + 4\pi m - a^{2} \phi.
\end{align*}
The value of $C(\lambda)$ is chosen to ensure that
\begin{align*}
\frac{4 \lambda (\smfrac{5}{3} - \lambda)}{3} w - 4\pi w^{1/2} + \frac{4 \lambda}{3} C(\lambda) \geq 0,
\end{align*}
hence as $m$ is non-negative and $v \geq 0$ in $S$
\begin{align*}
\Delta v &\geq \frac{4 \lambda}{3} a^{2} w  + 4\pi m - a^{2} \phi \\
&\geq a^{2} ( \lambda w - \phi) = a^{2}( v + (C(\lambda) + a^{2} )) \geq a^{2} (C(\lambda) + a^{2} ) \geq 0.
\end{align*}
As $v$ satisfies
\begin{align*}
- \Delta v &\leq 0 \quad \text{ in } S, \\
v &= 0 \quad \text{ on } \partial S,
\end{align*}
it follows that both $v \leq 0$ and $v > 0$ on $S$, hence $S$ is non-empty and $v \leq 0$ on $\R$. So for all $\lambda \in (0,\smfrac{5}{3})$ and all $x \in \R$
\begin{align*}
\lambda u^{4/3}(x) \leq \phi(x) + C(\lambda) + a^{2}.
\end{align*}
The right-hand side is minimised by choosing $\lambda = \smfrac{10}{9}$, which yields the desired estimate (\ref{eq: Solovej Yuk est}).

\emph{Step 3} The aim is to show that there exists $a_{0} = a_{0}(\omega), R_{0} = R_{0}(\omega) > 0$ such that for all $0 < a \leq a_{0}$ and $R_{n} \geq R_{0}$, $u_{a,R_{n}} > 0$ on $\R$, by following the argument used in \cite[Proposition 2.2]{C/LB/L}.

First recall the energy minimisation problem (\ref{eq: I a Rn min problem})
\begin{align*}
I^{\TFW}_{a}(m_{R_{n}}) = \inf \left\{ \, E^{\TFW}_{a}(v,m_{R_{n}})  \, \bigg| \, \nabla v \in L^{2}(\R), v \in L^{10/3}(\R), v \geq 0 \, \right \}
\end{align*}
where
\begin{align}
E^{\TFW}_{a} (v,m_{R_{n}}) &= \int_{\R} |\nabla v|^{2} + \int_{\R} v^{10/3} + \frac{1}{2} D_{a}( m_{R_{n}} - v^{2}, m_{R_{n}} - v^{2}). \label{eq:E-TFW-a-energy-def}
\end{align}
By showing that for large $R_{n}$ and small $a > 0$
\begin{align}
I^{\TFW}_{a}(m_{R_{n}}) = E^{\TFW}_{a}(u_{a,R_{n}},m_{R_{n}}) < E^{\TFW}_{a}(0,m_{R_{n}}), \label{eq:uaRn-0-comparison-est}
\end{align}
it follows that $u_{a,R_{n}} \not\equiv 0$, hence by the Harnack inequality $u_{a,R_{n}} > 0$ on $\R$ \cite{Gilbarg/Trudinger}.
An admissible test function $\varphi_{a}$ is constructed to satisfy: for sufficiently large $R_{n} $
\begin{align*}
I^{\TFW}_{a}(m_{R_{n}}) \leq E^{\TFW}_{a}(\varphi_{a_{0}},m_{R_{n}}) < E^{\TFW}_{a}(0,m_{R_{n}}) = \frac{1}{2} D_{a}( m_{R_{n}} , m_{R_{n}}).
\end{align*}

For $\varepsilon > 0$, let $\varphi_{a} = \varepsilon \psi_{a}$ and consider the difference
\begin{align}
E^{\TFW}_{a} &(\varepsilon \psi_{a},m_{R_{n}}) - E^{\TFW}_{a} (0,m_{R_{n}}) \nonumber \\ &= \varepsilon^{2} \left( \int |\nabla \psi_{a}|^{2} - D_{a}( m_{R_{n}}, \psi_{a}^{2}) \right) + \frac{\varepsilon^{4}}{2} D_{a}( \psi_{a}^{2}, \psi_{a}^{2}) + \varepsilon^{10/3} \int \psi_{a}^{10/3}. \label{eq:E-TFW-energy-diff-eq}
\end{align}
For small $\varepsilon > 0$, the right-hand side of (\ref{eq:E-TFW-energy-diff-eq}) is shown to be negative by first proving that there exists $a_{0}, C_{0} > 0$ such that for all $0 < a \leq a_{0}$ 
\begin{align}
\int_{\R} |\nabla \psi_{a}|^{2} - D_{a}( m_{R_{n}}, \psi_{a}^{2}) \leq - \frac{C_{0}}{2} \, a < 0. \label{eq:Nabla-a-neg-eq}
\end{align}

Let $\psi_{0} \in C^{\infty}_{c}(B_{1}(0))$ satisfy $\psi_{0} \geq 0$, and $\psi_{0} = 1$ on $B_{1/2}(0)$, then define $\psi_{a}(x) = a^{3/2}\psi_{0}(ax)$, for $a \in (0,1]$. 

Using the definition of $\psi_{a}$ gives
\begin{align}
D_{a}( m_{R_{n}}, \psi_{a}^{2}) &= \int_{\R} \left( m_{R_{n}} * Y_{a} \right) \psi_{a}^{2} \geq
\frac{a^{3}}{4} \int_{B_{1/2a}(0)} \left( m_{R_{n}} * Y_{a} \right)(x) \id x \nonumber \\
&= a^{3} \int_{\R} \left( \int_{B_{1/2a}(0) \cap B_{R_{n}}(y)} m_{R_{n}}(x-y) \id x \right) \frac{e^{-a|y|}}{|y|} \id y \nonumber \\
&= a^{3} \int_{\R} \left( \int_{B_{1/2a}(-y) \cap B_{R_{n}}(0)} m_{R_{n}}(x) \id x \right) \frac{e^{-a|y|}}{|y|} \id y. \label{eq: m kj conv}
\end{align}
First consider for $R' > 0$
\begin{align*}
\int_{B_{R'}(0)} \frac{e^{-a|y|}}{|y|} \id y = 4 \pi \int_{0}^{R'} r e^{-ar} \id r = \frac{4 \pi}{a^{2}} \left( 1 - e^{-aR'} (1 + a R')  \right),
\end{align*}
hence choosing $R' = (4a)^{-1}$ ensures that
\begin{align}
\int_{B_{1/4a}(0)} \frac{e^{-a|y|}}{|y|} \id y = \frac{4 \pi}{a^{2}} \left( 1 - \smfrac{5}{4}e^{-1/4} \right) \geq \frac{\pi}{10 a^{2}}. \label{eq: Yuk exp int est}
\end{align}
Now choose $a^{*} = \min\{1, (4R_{0})^{-1}\}$ and suppose $R_{n} \geq R_{0}$. Then for all $y \in B_{1/4a}(0)$, it follows from the triangle inequality that $B_{R_{0}}(0) \subset B_{1/2a}(-y) \cap B_{R_{n}}(0)$, hence
\begin{align}
\int_{B_{1/2a}(-y) \cap B_{R_{n}}(0)} m_{R_{n}}(x) \id x \geq \int_{B_{R_{0}}(0)} m(x) \id x \geq c_{0} > 0. \label{eq: m kj int lower bound}
\end{align}
Applying (\ref{eq: Yuk exp int est})--(\ref{eq: m kj int lower bound}) to (\ref{eq: m kj conv}), it follows that for all $0 < a \leq a^{*}$ and $R_{n} \geq R_{0}$
\begin{align}
D_{a}( m_{R_{n}}, \psi_{a}^{2}) &= \int_{\R} \left( m_{R_{n}} * Y_{a} \right) \psi_{a}^{2} \nonumber \\ &\geq a^{3} \int_{\R} \left( \int_{B_{1/2a}(-y) \cap B_{R_{n}}(0)} m_{R_{n}}(x) \id x \right) \frac{e^{-a|y|}}{|y|} \id y \nonumber \\
&= c_{0} a^{3} \int_{B_{1/4a}(0)} \frac{e^{-a|y|}}{|y|} \id y \geq \frac{c_{0} \pi}{10} a =: C_{0} \, a. \label{eq: Technical Lemma Application}
\end{align}

Using a change of variables
\begin{align}
\int_{B_{1/a}(0)} |\nabla \psi_{a}|^{2} = a^{2} \int_{B_{1}(0)} |\nabla \psi_{0}|^{2} =: C_{1} a^{2}. \label{eq:C2-eq}
\end{align}
Now define $a_{0} = \min\{a^{*}, \smfrac{C_{0}}{2C_{1}}\}$,
then for any $0 < a \leq a_{0}$ and $R_{n} \geq R_{0}$, combining (\ref{eq: Technical Lemma Application})--(\ref{eq:C2-eq}) yields (\ref{eq:Nabla-a-neg-eq})
\begin{align*}
\int |\nabla \psi_{a}|^{2} - D_{a}( m_{R_{n}}, \psi_{a}^{2}) \leq C_{1} a^{2} - C_{0} a \leq \frac{C_{0}}{2} \, a - C_{0} \, a = - \frac{C_{0}}{2} \, a < 0.
\end{align*}
Using that $a_{0}, \varepsilon \in (0,1]$, the remaining terms in (\ref{eq:E-TFW-energy-diff-eq}) can be estimated using a change of variables
\begin{align}
\frac{\varepsilon^{4}}{2} D_{a}( \psi_{a_{0}}^{2}, \psi_{a_{0}}^{2}) &+ \varepsilon^{10/3} \int \psi_{a_{0}}^{10/3} = \frac{\varepsilon^{4} a_{0}}{2} D_{0}( \psi_{0}^{2}, \psi_{0}^{2}) + \varepsilon^{10/3} a_{0}^{7} \int \psi_{0}^{10/3} \nonumber \\
&\leq \left( \frac{1}{2} D_{0}( \psi_{0}^{2}, \psi_{0}^{2}) + \int \psi_{0}^{10/3} \right) \varepsilon^{4} a_{0} =: C_{2} \varepsilon^{4} a_{0}. \label{eq:C3-epsilon4-est}
\end{align}
Applying the estimates (\ref{eq:Nabla-a-neg-eq})-(\ref{eq:C3-epsilon4-est}) to (\ref{eq:E-TFW-energy-diff-eq}) and choosing $0 < \varepsilon \leq \min\{1, (\smfrac{C_{0}}{3C_{2}})^{1/2}\}$ yields the desired result (\ref{eq:uaRn-0-comparison-est})
\begin{align*}
E^{\TFW}_{a} &(\varepsilon \psi_{a},m_{R_{n}}) - E^{\TFW}_{a} (0,m_{R_{n}}) \nonumber \leq \left( C_{2} \varepsilon^{2} - \frac{C_{0}}{2} \right) \varepsilon^{2} a_{0} < 0.
\end{align*}
\emph{Step 4} The aim is to show a uniform upper bound for $\phi_{a,R_{n}}$, which together with (\ref{eq: Solovej Yuk est}) yields the uniform estimate (\ref{eq: u phi a Rn L inf est})
\begin{align*}
\| u_{a,R_{n}} \|_{L^{\infty}(\R)}^{4/3} + \| \phi_{a,R_{n}} \|_{L^{\infty}(\R)} \leq C(M) + a^{2} \leq C(M) + 1,
\end{align*}
where the constant is independent of $a$ and $R_{n}$. This will be proved by adapting the argument used to show uniform regularity for finite systems with Coulomb interaction \cite{Paper1, C/LB/L}.

As $u_{a,R_{n}} \geq 0$, re-arranging the Solovej estimate (\ref{eq: Solovej Yuk est}) gives the uniform lower bound
\begin{align}
\phi_{a,R_{n}} \geq - ( C_{S} + a^{2} ). \label{eq: phi aRn Solovej bound below}
\end{align}
If $\phi_{a,R_{n}}$ is non-positive, then (\ref{eq: u phi a Rn L inf est}) holds as
\begin{align*}
\| u_{a,R_{n}} \|_{L^{\infty}(\R)}^{4/3} + \| \phi_{a,R_{n}} \|_{L^{\infty}(\R)} \leq 2 ( C_{S} + a^{2} ) \leq 2 ( C_{S} + 1 ).
\end{align*}
Instead, suppose that $\phi^{+}_{a,R_{n}}$ is non-zero at some point in $\R$. As shown in Step 1, $\phi_{a,R_{n}}$ is a continuous function that decays at infinity, hence there exists $x_{a,R_{n}} \in \R$ such that
\begin{align}
\phi_{a,R_{n}}^{+}(x_{a,R_{n}}) = \| \phi_{a,R_{n}}^{+} \|_{L^{\infty}(\R)} > 0. \label{eq: phi Rn sup xn = 0}
\end{align}
Without loss of generality, assume that $x_{a,R_{n}} = 0$.

In Step 1, it was shown that $u_{a,R_{n}}, \phi_{a,R_{n}} \in L^{\infty}(\R), \nabla u_{a,R_{n}} \in L^{2}(\R), \phi_{a,R_{n}} \in H^{1}(\R)$. Consequently, applying \cite[Lemma 6.1]{Paper1} implies that \newline $L_{a,R_{n}} = - \Delta + \frac{5}{3} u_{a,R_{n}}^{4/3} - \phi_{a,R_{n}}$ is a non-negative operator.

Choose $\varphi \in C^{\infty}_{c}(B_{1}(0))$ satisfying $0 \leq \varphi \leq 1$, $\varphi = 1$ on $B_{1/2}(0)$ and $\int_{\R} \varphi^{2} = 1$, then for $y \in \R$, define $\varphi_{y} \in C^{\infty}_{c}(B_{1}(y))$ by $\varphi_{y} = \varphi(\cdot - y)$. As $L_{a,R_{n}}$ is non-negative
\begin{align*}
\langle \varphi_{y}, L_{a,R_{n}} \varphi_{y} \rangle = \int_{\R} |\nabla \varphi_{y}|^{2} + \int_{\R} \left( \frac{5}{3} u_{a,R_{n}}^{4/3} - \phi_{a,R_{n}} \right) \varphi_{y}^{2} \geq 0,
\end{align*}
which can be re-arranged and expressed using convolutions as
\begin{align}
\frac{5}{3} \left( u_{a,R_{n}}^{4/3} * \varphi^{2} \right) &\geq  \left( \phi_{a,R_{n}} * \varphi^{2} - \int_{\R} |\nabla \varphi|^{2} \right)_{+} \nonumber \\
&= \left( \phi_{a,R_{n}} * \varphi^{2} - C \right)_{+} \label{eq: phi Rn first convolution est}
\end{align}
Observe that $\phi_{a,R_{n}} * \varphi^{2}$ solves
\begin{align}
- \Delta \left( \phi_{a,R_{n}} * \varphi^{2} \right) + a^{2}\left( \phi_{a,R_{n}} * \varphi^{2} \right) &= 4 \pi \left( m_{R_{n}} * \varphi^{2} - u^{2}_{a,R_{n}} * \varphi^{2}  \right). \label{eq: Delta phi Rn convolution est}
\end{align}
The first term can be estimated uniformly
\begin{align}
\left(m_{R_{n}} * \varphi^{2} \right)(x) &= \int_{B_{1}(x)} m_{R_{n}}(y) \varphi^{2}(x - y) \id y \nonumber \\ &\leq \int_{B_{1}(x)} m(y) \id y \leq C_{0} \| m \|_{L^{2}_{\unif}(\R)} \leq C_{0} M. \label{eq: m Rn convolution est}
\end{align}
For the second term, using the convexity of $t \mapsto t^{3/2}$ for $t \geq 0$ and the fact that $\int \varphi^{2} = 1$, applying Jensen's inequality and (\ref{eq: phi Rn first convolution est}) implies that
\begin{align}
4 \pi \, u_{a,R_{n}}^{2} * \varphi^{2}(x) &\geq \frac{5}{3} u_{a,R_{n}}^{2} * \varphi^{2}(x) \nonumber \\
&= \frac{5}{3} \int_{\R} u_{a,R_{n}}^{2}(x-y) \varphi^{2}(y) \id y \nonumber \\
&= \frac{5}{3} \int_{\R} \left(u_{a,R_{n}}^{4/3}(x-y)\right)^{3/2} \varphi^{2}(y) \id y \nonumber \\
&\geq \frac{5}{3} \left( \int_{\R} u_{a,R_{n}}^{4/3}(x-y) \varphi^{2}(y) \id y \right)^{3/2} \nonumber \\
&= \frac{5}{3}(u^{4/3}_{a,R_{n}} * \varphi^{2})^{3/2} \geq \left( \phi_{a,R_{n}} * \varphi^{2} - C  \right)^{3/2}_{+}. \label{eq: phi Rn convolution est}
\end{align}
Combining the estimates (\ref{eq: Delta phi Rn convolution est})--(\ref{eq: phi Rn convolution est}) yields
\begin{align*}
- \Delta \left( \phi_{a,R_{n}} * \varphi^{2} \right) + a^{2}\left( \phi_{a,R_{n}} * \varphi^{2} \right) + \left( \phi_{a,R_{n}} * \varphi^{2} - C  \right)^{3/2}_{+} \leq C_{0} M.
\end{align*}
Observe that as $\phi_{a,R_{n}}$ is a continuous function that decays at infinity, $\phi_{a,R_{n}} * \varphi^{2}$ also shares these properties. Now consider the set
\begin{align*}
S = \{ \, x \in \R \, | \, \phi_{a,R_{n}} * \varphi^{2} - C > 0 \, \},
\end{align*}
it follows that $S$ is open and bounded and that $\phi_{a,R_{n}} * \varphi^{2} - C = 0$ on $\partial S$. Observe that the constant function $h = (C_{0}M)^{2/3}$ satisfies
\begin{align*}
- \Delta h + a^{2}(h + C) + h^{3/2}_{+} \geq h^{3/2}_{+} = C_{0}M \quad &\text{ on } S, \\
0 = \, \phi_{a,R_{n}} * \varphi^{2} - C \leq h \quad &\text{ in } \partial S,
\end{align*}
so by the maximum principle $\phi_{a,R_{n}} * \varphi^{2} \leq C(1 + M^{2/3})$
over $S$, and also on $S^{c}$, hence
\begin{align}
\phi_{a,R_{n}} * \varphi^{2} \leq C(1 + M^{2/3}). \label{eq: phi Rn M convolution est}
\end{align}
Applying (\ref{eq: phi aRn Solovej bound below}), it follows that
\begin{align}
\phi_{a,R_{n}}^{+} * \varphi^{2} &= \phi_{a,R_{n}}^{-} * \varphi^{2} + \phi_{a,R_{n}} * \varphi^{2} \leq C_{S} + a^{2} + C(1 + M^{2/3})  = C(1 + M^{2/3}) + a^{2}. \label{eq: phi aRn plus conv est}
\end{align}
Additionally,
\begin{align}
-\Delta \phi_{a,R_{n}}^{+} &\leq -\Delta \phi_{a,R_{n}}^{+} + a^{2} \phi_{a,R_{n}}^{+} = \left( -\Delta \phi_{a,R_{n}} + a^{2} \phi_{a,R_{n}} \right) \chi_{\{\phi_{a,R_{n}} > 0\}} \nonumber \\ &= 4 \pi \left( m_{R_{n}} - u_{a,R_{n}}^{2} \right) \chi_{\{\phi_{a,R_{n}} > 0\}} \leq 4 \pi m_{R_{n}} \chi_{\{\phi_{a,R_{n}} > 0\}}  \leq 4 \pi m_{R_{n}}. \label{eq: phi aRn plus Lap est}
\end{align}
From this point onwards, following the proof of \cite[Proposition 6.2]{Paper1} verbatim with the estimates (\ref{eq: phi aRn plus conv est})--(\ref{eq: phi aRn plus Lap est}) gives
\begin{align}
\| \phi_{a,R_{n}}^{+} \|_{L^{\infty}(\R)} \leq C(1 + M) + a^{2}. \label{eq: phi aRn plus Linf est}
\end{align}
Combining (\ref{eq: phi aRn Solovej bound below})--(\ref{eq: phi aRn plus Linf est}) with the Solovej estimate (\ref{eq: Solovej Yuk est}), yields the desired estimate (\ref{eq: u phi a Rn L inf est})
\begin{align*}
\| u_{a,R_{n}} \|_{L^{\infty}(\R)}^{4/3} + \| \phi_{a,R_{n}} \|_{L^{\infty}(\R)} \leq C(1+M) + a^{2} \leq C(1+M).
\end{align*}
Then, as in the proof of \cite[Proposition 6.2]{Paper1}, applying elliptic regularity estimates to the system (\ref{eq: u phi finite Yuk pair}) yields the desired estimates (\ref{eq: uRn a H4 unif est})--(\ref{eq: phiRn a H2 unif est}).
\begin{align*}
\| u_{a,R_{n}} \|_{H^{4}_{\unif}(\R)} &\leq C(M), \\
\| \phi_{a,R_{n}} \|_{H^{2}_{\unif}(\R)} &\leq C(M). \qedhere
\end{align*}
\end{proof}


\begin{proof}[Proof of \Cref{Proposition - Finite Yukawa Regularity Est All a}]
The proof follows the steps used to show \Cref{Proposition - Finite Yukawa Regularity Est}. Steps 1, 2 and 4 hold verbatim and Step 3 is modified to instead show that for any $a_{0} > 0$ and $m \in \mathcal{M}_{L^{2}}(M,\omega)$, there exists $R_{0} = R_{0}(a_{0},\omega) > 0$ such that for any $0 < a \leq a_{0}$ and $R_{n} \geq R_{0}$, the unique minimiser $u_{a,R_{n}}$ of (\ref{eq: I a Rn min problem}) satisfies
\begin{align}
u_{a,R_{n}} > 0 \, \text{ on } \, \R. \label{eq: u pos all a}
\end{align}
Recall the energy minimisation problem (\ref{eq: I a Rn min problem})
\begin{align*}
I^{\TFW}_{a}(m_{R_{n}}) = \inf \left\{ \, E^{\TFW}_{a}(v,m_{R_{n}})  \, \bigg| \, \nabla v \in L^{2}(\R), v \in L^{10/3}(\R), v \geq 0 \, \right \}
\end{align*}
where
\begin{align*}
E^{\TFW}_{a} (v,m_{R_{n}}) &= \int_{\R} |\nabla v|^{2} + \int_{\R} v^{10/3} + \frac{1}{2} D_{a}( m_{R_{n}} - v^{2}, m_{R_{n}} - v^{2}). 
\end{align*}
A family of test functions $\varphi_{R_{n}}$ is now constructed to satisfy: for large $R_{n}$ 
\begin{align}
I^{\TFW}_{a}(m_{R_{n}}) \leq E^{\TFW}_{a}(\varphi_{R_{n}},m_{R_{n}}) < E^{\TFW}_{a}(0,m_{R_{n}}) = \frac{1}{2} D_{a}( m_{R_{n}} , m_{R_{n}}). \label{eq: phi Rn construction}
\end{align}
It follows from (\ref{eq: phi Rn construction}) that
\begin{align}
I^{\TFW}_{a}(m_{R_{n}}) = E^{\TFW}_{a}(u_{a,R_{n}},m_{R_{n}}) < E^{\TFW}_{a}(0,m_{R_{n}}), \label{eq:uaRn-all-a-0-comparison-est}
\end{align}
which implies that $u_{a,R_{n}} \not\equiv 0$, hence by the Harnack inequality $u_{a,R_{n}} > 0$ on $\R$ \cite{Gilbarg/Trudinger}, hence (\ref{eq: u pos all a}) holds.

Let $\psi_{R_{n}} \in C^{\infty}_{c}(B_{4R_{n}}(0))$ satisfy $\psi_{R_{n}} \geq 0$ and $\psi_{R_{n}} = 1$ on $B_{2R_{n}}(0)$. Then let $\varepsilon > 0$ and consider the difference
\begin{align}
E^{\TFW}_{a} &(\varepsilon \psi_{R_{n}},m_{R_{n}}) - E^{\TFW}_{a} (0,m_{R_{n}}) \nonumber \\ &= \varepsilon^{2} \left( \int |\nabla \psi_{R_{n}}|^{2} - D_{a}( m_{R_{n}}, \psi_{R_{n}}^{2}) \right) + \frac{\varepsilon^{4}}{2} D_{a}( \psi_{R_{n}}^{2}, \psi_{R_{n}}^{2}) + \varepsilon^{10/3} \int \psi_{R_{n}}^{10/3}. \label{eq:E-TFW-psiRn-energy-diff-eq}
\end{align}
Applying (\ref{eq: TL2}) of \Cref{Lemma - Technical Lemma 2}, there exists $R_{0} > 0$ such that for any $R_{n} \geq R_{0}$
\begin{align}
\int_{\R} |\nabla \psi_{R_{n}}|^{2} - D_{a}(m_{R_{n}}, \psi_{R_{n}}^{2}) \leq - C_{0} R_{n}^{3}. \label{eq: TL2 repeat est}
\end{align}

The remaining terms in (\ref{eq:E-TFW-psiRn-energy-diff-eq}) can be estimated for $0 < \varepsilon \leq 1$, using Young's inequality for convolutions and Cauchy-Schwarz, by
\begin{align}
\frac{\varepsilon^{4}}{2} D_{a}( \psi_{R_{n}}^{2}, \psi_{R_{n}}^{2}) + \varepsilon^{4} \int \psi_{R_{n}}^{10/3} &\leq \frac{\varepsilon^{4}}{2} D_{a}( \chi_{B_{2R_{n}}(0)}, \chi_{B_{2R_{n}}(0)}) + \varepsilon^{4} \int_{B_{2R_{n}}(0)} 1 \nonumber \\
&\leq \left( \frac{1}{2} \| Y_{a} \|_{L^{1}(\R)} \| \chi_{B_{2R_{n}}(0)} \|_{L^{2}(\R)}^{2} + \| \chi_{B_{2R_{n}}(0)} \|_{L^{1}(\R)} \right) \varepsilon^{4} \nonumber \\
&\leq C( 1 + a^{-2}) R_{n}^{3} \varepsilon^{4} =: C_{3} \varepsilon^{4} R_{n}^{3}. \label{eq: all a pos est}
\end{align}
Combining the estimates (\ref{eq: TL2 repeat est})--(\ref{eq: all a pos est}) and choosing $0 < \varepsilon \leq \varepsilon_{0} := \min\{1, (\smfrac{C_{0}}{2C_{3}})^{1/2}\}$ ensures that
\begin{align*}
E^{\TFW}_{a} &(\varepsilon \psi_{R_{n}},m_{R_{n}}) - E^{\TFW}_{a} (0,m_{R_{n}}) \leq \left( - C_{0}  + C_{3} \varepsilon^{2} \right) \varepsilon^{2} R_{n}^{3} < 0,
\end{align*}
hence the desired estimate (\ref{eq: phi Rn construction}) holds.
\end{proof}

\begin{proof}[Proof of \Cref{Proposition - Yukawa Regularity Est}]
First suppose that $\spt(m)$ is bounded, then by \Cref{Proposition - Finite Yukawa Regularity Est} there exists $a_{0} > 0$ such that for all $0 < a \leq a_{0}$ and sufficiently large $R_{n}$, $m = m_{R_{n}}$ and hence $(u_{a},\phi_{a}) = (u_{a,R_{n}},\phi_{a,R_{n}})$ solves (\ref{eq: u phi Yuk eq pair}) and satisfies the desired estimate (\ref{eq:Yuk-u-H4-phi-H2-unif-est}).

Now suppose $\spt(m)$ is unbounded, then the estimates (\ref{eq: uRn a H4 unif est})--(\ref{eq: phiRn a H2 unif est}) of \Cref{Proposition - Finite Yukawa Regularity Est} guarantee that for all $0 < a \leq a_{0}$ and $R_{n}$ sufficiently large, the sequences $u_{a,R_{n}}, \phi_{a,R_{n}}$ are bounded uniformly in $H^{2}_{\unif}(\R)$. Consequently, there exist $u_{a},  \phi_{a} \in H^{2}_{\unif}(\R) \cap L^{\infty}(\R)$ such that along a subsequence $u_{a,R_{n}}, \phi_{a,R_{n}}$ converges to $u_{a}, \phi_{a}$, weakly in $H^{2}(B_{R}(0))$, strongly in $H^{1}(B_{R}(0))$ for all $R>0$ and pointwise almost everywhere. It follows from the pointwise convergence that $u_{a} \geq 0$ and
\begin{align*}
\| u_{a} \|_{L^{\infty}(\R)} + \| \phi_{a} \|_{L^{\infty}(\R)} &\leq C(M).
\end{align*}
Passing to the limit of the equations (\ref{eq: u phi finite Yuk pair}) in distribution shows the limit $(u_{a},\phi_{a})$ solves
\begin{align*}
&- \Delta u_{a} + \frac{5}{3} u_{a}^{7/3} - \phi_{a} u_{a} = 0, \\
&- \Delta \phi_{a} + a^{2}\phi_{a} = 4 \pi ( m - u_{a}^{2} ).
\end{align*}
Following the argument used to prove (\ref{eq: uRn a H4 unif est})--(\ref{eq: phiRn a H2 unif est}) in this instance yields the desired estimate (\ref{eq:Yuk-u-H4-phi-H2-unif-est}) holds
\begin{equation*}
\| u_{a} \|_{H^{4}_{\unif}(\R)} + \| \phi_{a} \|_{H^{2}_{\unif}(\R)} \leq C(M). \qedhere
\end{equation*}
\end{proof}

\begin{proof}[Proof of \Cref{Proposition - Yukawa Regularity Est All a}]
This holds from applying \Cref{Proposition - Finite Yukawa Regularity Est All a} and following the proof of \Cref{Proposition - Yukawa Regularity Est} in the unbounded case verbatim.
\end{proof}

\begin{proposition}
\label{Proposition - Uniform Yuk inf u estimate}
There exists $a_{\rm c} = a_{\rm c}(M,\omega) > 0$ and $c_{a_{\rm c},M,\omega} > 0$ such that for all $m \in
\mathcal{M}_{L^{2}}(M,\omega)$ and $0 < a \leq a_{\rm c}$ the corresponding Yukawa ground state $(u_{a},\phi_{a})$ is unique and the electron density $u_{a}$ satisfies
\begin{align}
\inf_{x \in \R} u_{a}(x) &\geq c_{a_{\rm c},M,\omega} > 0. \label{eq:u-Yuk-inf-est}
\end{align}

\end{proposition}

\begin{proof}[Proof of \Cref{Proposition - Uniform Yuk inf u estimate}]
The proof of \Cref{Proposition - Uniform Yuk inf u estimate} closely follows the proof of \cite[Proposition 6.2]{Paper1} and \cite[Theorem 6.10]{C/LB/L}. The estimate (\ref{eq:u-Yuk-inf-est}) is shown by contradiction, so suppose that for any $a_{\rm c} > 0$
\begin{align*}
\inf_{0 < a \leq a_{\rm c}} \inf_{m \in \mathcal{M}_{L^{2}}(M,\omega)} \inf_{x \in \R} u_{a}(x) = 0,
\end{align*}
hence there exists sequences $a_{n} \downarrow 0$ satisfying $a_{n} \leq a_{1}$ for all $n \in \mathbb{N}$, $(m_{n}) \subset \mathcal{M}_{L^{2}}(M,\omega)$ and $(x_{n}) \subset \R$ such that for all $n \in \mathbb{N}$ the ground state $(u_{n}, \phi_{n})$, corresponding to $m_{n}$ with Yukawa parameter $a_{n}$, satisfies
\begin{align}
u_{n}(x_{n}) \leq \frac{1}{n}. \label{eq: uk contr hyp est}
\end{align}
As $\smfrac{5}{3} u_{n}^{4/3} - \phi_{n} u_{n} \in L^{2}_{\loc}(\R)$, $u_{n} \in H^{1}_{\unif}(\R)$ and $u_{n} > 0$ solves
\begin{align*}
L_{n}u_{n} := \left( - \Delta + \frac{5}{3} u_{n}^{4/3} - \phi_{n} \right) u_{n} = 0,
\end{align*}
applying the Harnack inequality \cite{Trudinger_MeasurableCoefficients}, and
observing that the coefficients of $L_{n}$ are uniformly estimated by \Cref{Proposition - Yukawa Regularity Est All a}, this yields a uniform Harnack
constant, hence for all $R > 0$, there exists $C = C(R,a_{1}, M) > 0$ such
that for all $n \in \mathbb{N}$
\begin{align*}
\sup_{x \in B_{R}(x_{n})} u_{n}(x) \leq C \inf_{x \in B_{R}(x_{n})} u_{n}(x) \leq \frac{C}{n}.
\end{align*}
It follows that the sequence of functions $u_{n}(\cdot + x_{n})$ converges
uniformly to zero on compact sets. Consider the ground state $(u_{n},\phi_{n})$
corresponding to the nuclear distribution $m_{n}$.

By the Harnack inequality, it follows that $u_{n}(\cdot + x_{n})$ converges uniformly to 0 on compact subsets. Recall that $\phi_{n}$ satisfies
\begin{align*}
- \Delta \phi_{n} + a_{n}^{2} \phi_{n} = 4 \pi ( m_{n} - u_{n}^{2} )
\end{align*}
in distribution. In addition, $\phi_{n}$ and $m_{n}$ satisfy
\begin{align*}
\| m_{n}(\cdot + x_{n}) \|_{L^{2}_{\unif}(\R)} +
\| \phi_{n}(\cdot + x_{n}) \|_{H^{2}_{\unif}(\R)} &\leq C(a_{1},M).
\end{align*}
It follows that along a subsequence $\phi_{n}(\cdot + x_{n})$ converges to $\wt \phi$, weakly in $H^{2}(B_{R}(0))$, strongly in $H^{1}(B_{R}(0))$ for all $R > 0$ and pointwise almost everywhere. Also, $m_{n}(\cdot + x_{n})$ converges to $\wt m$, weakly in $L^{2}(B_{R}(0))$ for all $R > 0$. By the  Lebesgue-Besicovitch Differentiation Theorem \cite{Evans/Gariepy}, $\wt m \in \mathcal{M}_{L^{2}}(M,\omega)$. As $a_{n} \downarrow 0$, passing to the limit of
\begin{align*}
- \Delta \phi_{n}(\cdot + x_{n}) + a_{n}^{2} \phi_{n}(\cdot + x_{n}) = 4 \pi \left( m_{n}(\cdot + x_{n}) - u_{n}^{2}(\cdot + x_{n}) \right)
\end{align*}
shows that $\wt \phi$ is a distributional solution of
\begin{align}
- \Delta \wt \phi = 4\pi \wt m. \label{eq: phi bar H2 Yuk eq}
\end{align}
The argument of \cite[Theorem 6.10]{C/LB/L} is now used to show that for all $R > 0$
\begin{align}
\int_{B_{R}(0)} \wt m(z) \id z \leq C R. \label{eq: m bar BR contradiction est}
\end{align}
As $\wt m \in \mathcal{M}_{L^{2}}(M,\omega)$, this leads to the contradiction that for all $R > 0$
\begin{align*}
\omega_{0} R^{3} - \omega_{1} \leq \int_{B_{R}(0)} \wt m(z) \id z \leq C R.
\end{align*}
To show (\ref{eq: m bar BR contradiction est}) choose $\varphi \in C^{\infty}_{\textnormal{c}}(B_{2}(0))$ such that $0 \leq \varphi \leq 1$ and $\varphi = 1$ on $B_{1}(0)$. Let $R>0$, then testing (\ref{eq: phi bar H2 Yuk eq}) with $\varphi(\cdot/R)$ gives
\begin{align}
- \frac{1}{R^{2}} \int_{B_{2R}(0)} \wt \phi(z) (\Delta \varphi) (z/R) \id z = 4 \pi \int_{B_{2R}(0)} \wt m(z) \varphi(z/R) \id z. \label{eq: overline m contr est}
\end{align}
The left-hand side can be estimated by
\begin{align}
\frac{1}{R^{2}} \bigg | \int_{B_{2R}(0)} \wt \phi(z) (\Delta \varphi) (z/R) \id z \bigg | \leq \| \wt \phi \|_{L^{\infty}(\R)} \| \Delta \varphi \|_{L^{\infty}} \frac{|B_{2R}(0)|}{R^{2}} \leq C R, \label{eq: wt m R est}
\end{align}
where the constant $C > 0$ is independent of $R$. As $\wt m \geq 0$, combining (\ref{eq: overline m contr est})--(\ref{eq: wt m R est}) yields (\ref{eq: m bar BR contradiction est})
\begin{align*}
\int_{B_{R}(0)} \wt m(z) \id z \leq \int_{B_{2R}(0)} \wt m(z) \varphi(z/R) \id z \leq C R.
\end{align*}
The contradiction ensures that there exists $a_{\rm c} > 0$ and $c_{a_{\rm c},M,\omega} > 0$ such that for all $m \in \mathcal{M}_{L^{2}}(M,\omega)$ and $0 < a \leq a_{\rm c}$, the corresponding Yukawa electron density $u_{a}$ satisfies
\begin{equation*}
\inf_{x \in \R} u_{a}(x) \geq c_{a_{\rm c},M,\omega} > 0. \qedhere
\end{equation*}
Consequently, for $0 < a \leq a_{\rm c}$, the electron density satisfies $\inf u_{a} > 0$, hence the arguments of \cite[Chapter 6]{C/LB/L} can be applied verbatim to guarantee the uniqueness of the ground state $(u_{a},\phi_{a})$.
\end{proof}

\begin{remark}
\Cref{Theorem - Yukawa Coulomb Comparison} provides an additional proof of \Cref{Proposition - Uniform Yuk inf u estimate}. Let $a_{0} > 0$ and $m \in \mathcal{M}_{L^{2}}(M,\omega),$ then for any $0 < a \leq a_{0}$, \cite[Propositions 3.1 and 3.2]{Paper1} and \Cref{Proposition - Yukawa Regularity Est All a} guarantees that there exist corresponding Coulomb and Yukawa ground states $(u,\phi)$, $(u_{a},\phi_{a})$, respectively satisfying $\inf u \geq c_{M,\omega} > 0$ and $u_{a} \geq 0$. Then applying (\ref{eq: Yukawa Coulomb pointwise k pos est}) of \Cref{Theorem - Yukawa Coulomb Comparison} implies
\begin{align*}
u_{a}(x) \geq u(x) - \| u_{a} - u \|_{L^{\infty}(\R)} \geq c_{M,\omega} - C' a^{2},
\end{align*}
hence for all $0 < a \leq a_{\rm c} := \min\{ a_{0}, ( \smfrac{c_{M,\omega}}{2C'} )^{1/2} \}$
\begin{equation*}
\inf_{x \in \R} u_{a}(x) \geq c_{M,\omega} - C' a^{2} \geq \frac{1}{2} c_{M,\omega} > 0. \qedhere 
\end{equation*} \qed

\end{remark}

The proof of \Cref{Proposition - Uniform Yuk inf u estimate all a} requires the following result, which extends the lower bound on $u_{a}$ from $0 < a \leq a_{\rm c}$ to arbitrary $a > 0$.
\begin{restatable}{proposition}{Proposition}
\label{Proposition - Uniform Yuk inf u estimate remainder}
Let $a_{0} > a_{\rm c} > 0$ and $m \in \mathcal{M}_{L^{2}}(M,\omega)$, then for all $0 < a \leq a_{0}$ the corresponding Yukawa ground state $(u_{a},\phi_{a})$ is unique and there exists $c_{a_{0},M,\omega} > 0$ such that the electron density $u_{a}$ satisfies
\begin{align}
\inf_{x \in \R} u_{a}(x) &\geq c_{a_{0},M,\omega} > 0. \label{eq:u-Yuk-inf-est-remainder}
\end{align}
\end{restatable}
Due to the length of the argument, the proof of \Cref{Proposition - Uniform Yuk inf u estimate remainder} is postponed to the Appendix, which can be found on Page \pageref{Appendix}.

\begin{proof}[Proof of \Cref{Proposition - Uniform Yuk inf u estimate all a}]
Combining \Cref{Proposition - Uniform Yuk inf u estimate} and \Cref{Proposition - Uniform Yuk inf u estimate remainder} yields the desired result.
\end{proof}

\begin{proof}[Proof of \Cref{Corollary - General Est Ck Yuk version}]
This is identical to the proof of \cite[Corollary 6.3]{Paper1}, using the estimates (\ref{eq: uRn a H4 unif est})-(\ref{eq: phiRn a H2 unif est}) to provide the initial regularity.
\end{proof}

\subsection{Proof of main results}
\label{Subsection - Proofs of Pointwise Stability Estimates}
The proofs of Theorems \ref{Theorem - Yukawa Coulomb Comparison}, \ref{Theorem - One inf pointwise stability estimate Yuk alt} and \ref{Theorem - Exponential Est Integral Yuk RHS} closely follow the proofs of \cite[Theorems 3.4 and 3.5]{Paper1}, which adapts the uniqueness of the TFW equations \cite{C/LB/L, Blanc_Uniqueness}.

First, two alternative sets of assumptions on nuclear
distributions $m_1, m_2$ are given. In the following, $(u_{0},\phi_{0})$ denotes the corresponding Coulomb ground state solving (\ref{eq:u-phi-eq-pair}), i.e the ground state with Yukawa parameter $a = 0$.
\begin{itemize}
\item[(A)] 
Let $k = 0$, $m_{1} \in \mathcal{M}_{L^{2}}(M,\omega)$, $m_{2} : \R \to \mathbb{R}_{\geq 0}$ satisfy 
\begin{align*}
\| m_{2} \|_{L^{2}_{\unif}(\R)} \leq M',
\end{align*}
then by \Cref{Proposition - Yukawa Regularity Est} there exist $a' = a'(\omega,m_{2}) > 0$ such that for all $0 \leq a_{1} \leq a_{2} \leq a'$ there exists $(u_{1},\phi_{1}) = (u_{1,a_{1}},\phi_{1,a_{1}})$ $(u_{2},\phi_{2}) = (u_{2,a_{2}},\phi_{2,a_{2}})$
solving either (\ref{eq:u-phi-eq-pair}) or (\ref{eq: u phi Yuk eq pair}) corresponding to $m_{2}$, satisfying $\inf u_{1} > 0$, $u_{2} \geq 0$ and
 \begin{align}
 \| u_{2} \|_{H^{4}_{\unif}(\R)} &+ \| \phi_{2} \|_{H^{2}_{\unif}(\R)} \leq C(M'). \label{eq: u2 phi2 reg A* est}
 \end{align} 
 In addition, assume either $m_{2} \not\equiv 0$ and $u_{2} > 0$ or $m_{2} = u_{2} = \phi_{2} = 0$. 

Observe that (A) assumes that $u_{2} > 0$, while Theorems \ref{Theorem - Yukawa Coulomb Comparison} (with $k = 0$) and \ref{Theorem - One inf pointwise stability estimate Yuk alt} only require either $u_{a} \geq 0$ or $u_{2,a} \geq 0$. The
restriction $u_{2} > 0$ will be lifted via a thermodynamic limit argument in the third part of its proof on page \pageref{proof-case3}.

\item[(B)] Let $a_{0} > 0$, $k \in \mathbb{N}_{0}$,
$m_{1}, m_{2} \in \mathcal{M}_{H^{k}}(M, \omega)$, $0 \leq a_{1} \leq a_{2} \leq a_{0}$ and let
$(u_{1},\phi_{1}) = (u_{1,a_{1}},\phi_{1,a_{1}}), (u_{2},\phi_{2}) = (u_{2,a_{2}},\phi_{2,a_{2}})$ denote the corresponding ground states.
(Note that (B) implies (A), with $a' = a_{0}$ and $M' = C(a_{0},M)$.)
\end{itemize}

In addition, for both (A) and (B), define
\begin{align*}
w = u_{1} - u_{2}, \quad \psi = \phi_{1} - \phi_{2},
\end{align*} 
and suppose that there exists $R \in H^{k'}_{\unif}(\R)$, where $k' \in \{k,k+2\}$, such that $(w,\psi)$ solves
\begin{subequations}
\label{eq: w psi R1 R2 pair eq}
\begin{align}
&- \Delta w + \frac{5}{3} \left( {u_{1}}^{7/3} - u_{2}^{7/3} \right) - \phi_{1} u_{1} + \phi_{2} u_{2} = 0, \label{eq: w Yuk eq} \\
&- \Delta \psi + a_{1}^{2} \psi = 4 \pi \left( u_{2}^{2} - u_{1}^{2} \right) + R. \label{eq: psi R Yuk eq}
\end{align}
\end{subequations}

\begin{lemma}
\label{Lemma - Exp Est Yuk Main Est}
Suppose that either \textnormal{(A)} or \textnormal{(B)} holds, then there exist $C = C_{A}(M,M',\omega)$, $\gamma = \gamma_{A}( M,M',\omega) > 0$ or $C = C_{B}(a_{0},k', M,\omega), \gamma = \gamma_{B}(a_{0}, M,\omega) > 0$, independent of both $a_{1},a_{2}$, such that for any $\xi \in H_{\gamma}$
\begin{align}
\int_{\R} \bigg( \sum_{|\alpha_{1}| \leq k+4} |\partial^{\alpha_{1}} w|^{2} &+ \sum_{|\alpha_{2}| \leq k'+2} |\partial^{\alpha_{2}} \psi|^{2} \bigg) \xi^{2} \leq C \int_{\R} \sum_{|\beta| \leq k'} |\partial^{\beta} R|^{2} \xi^{2}. \label{eq: w and psi R Yuk lemma est}
\end{align}
In particular, for any $y \in \R$,
\begin{align}
\sum_{|\alpha_{1}| \leq k+2} &|\partial^{\alpha_{1}} w(y)|^{2} + \sum_{|\alpha_{2}| \leq k'} |\partial^{\alpha_{2}} \psi(y)|^{2} 
\leq C \int_{\R} \sum_{|\beta| \leq k'} |\partial^{\beta} R(x)|^{2} e^{-2\gamma |x - y|} \id x. \label{eq: w and psi pointwise rhs exp integral final lemma est}
\end{align}
Further, 
if both $a_{1} = a_{2} = 0$, then $C = C_{B}(k',M,\omega), \gamma = \gamma_{B}(M,\omega)$.
\end{lemma}

One of the key steps in proving \Cref{Lemma - Exp Est Yuk Main Est} is showing
\begin{align}
\int_{\R} \psi^{2} \xi^{2} \leq C \left( \int_{\R} R \psi \xi^{2} + \int_{\R} ( w^{2} + \psi^{2} ) |\nabla \xi|^{2} \right), \label{eq: psi integral est}
\end{align}
where the constant $C$ is independent of $a_{1}, a_{2}$. However, due to the presence of the additional term in (\ref{eq: psi R Yuk eq}), the argument in \cite[Lemma 6.4]{Paper1} directly yields
\begin{align}
a_{1}^{2} \int_{\R} \psi^{2} \xi^{2} \leq C \left(  \int_{\R} R \psi \xi^{2} + \int_{\R} ( w^{2} + \psi^{2} ) |\nabla \xi|^{2} \right), \label{eq: psi integral est a}
\end{align}
where the left-hand constant tends to $0$ as $a_{1} \to 0$. Instead, (\ref{eq: psi integral est}) is obtained by closely following the proof in the Coulomb setting.

In the following proof, all integrals are taken over $\R$.

\begin{proof}[Proof of \Cref{Lemma - Exp Est Yuk Main Est}]
The argument closely follows the proof of \cite[Lemma 6.7]{Paper1}. This proof describes the key steps of the argument and additional details are provided in \cite{Paper1}.

\emph{Case 1.} Suppose (B) holds, so $m_{1}, m_{2} \in \mathcal{M}_{H^{k}}(M,\omega)$, so by \Cref{Corollary - General Est Ck Yuk version} (or \cite[Corollary 3.3]{Paper1} if either $a_{i} = 0$) for $i \in \{1,2\}$
\begin{align*}
\| u_{i} \|_{H^{k+4}_{\unif}(\R)} + \| \phi_{i} \|_{H^{k+2}_{\unif}(\R)} \leq C(a_{0},k,M,\omega)
\end{align*}
and by \Cref{Proposition - Uniform Yuk inf u estimate} $\inf u_{1}, \inf u_{2} \geq c_{a_{\rm c},M,\omega} > 0$ (if for $i \in \{1,2\}$ $a_{i} = 0$ then by \cite[Proposition 3.2]{Paper1} $\inf u_{i} \geq c_{M,\omega} > 0$). Let $\xi \in H^{1}(\R)$, then testing (\ref{eq: w Yuk eq}) with $w \xi^{2}$ and re-arranging yields
\begin{align}
\int |\nabla (w\xi)|^{2} &+ {\frac{5}{6}} \int ( u_{1}^{4/3} + u_{2}^{4/3}) w^{2} \xi^{2} - {\frac{1}{2}} \int (\phi_{1} + \phi_{2}) w^{2} \xi^{2} + \nu \int w^{2} \xi^{2} \nonumber \\ 
&\leq \int w^{2} |\nabla \xi|^{2} + {\frac{1}{2}} \int \psi ( u_{1}^{2} - u_{2}^{2} ) \xi^{2}, \label{eq: before Yuk l}
\end{align}
where $\nu = \smfrac{1}{2} (u_{1}^{4/3} + u_{2}^{4/3}) \geq \smfrac{1}{2} c_{a_{\rm c}, M, \omega}^{4/3} > 0$ (or $\nu \geq \smfrac{1}{2} c_{M, \omega}^{4/3} > 0$ when $a_{1} = a_{2} = 0$). As $u_{1}, u_{2} > 0$, \cite[Lemma 6.2]{Paper1} implies that
\begin{align*}
L = - \Delta + \smfrac{5}{6} (u_{1}^{4/3} + u_{2}^{4/3}) - \frac{1}{2} (\phi_{1} + \phi_{2})
\end{align*}
is a non-negative operator, hence (\ref{eq: before Yuk l}) can be expressed as
\begin{align}
\langle w \xi, L (w \xi) \rangle + \nu \int w^{2} \xi^{2} \leq \int w^{2} |\nabla \xi|^{2} + {\frac{1}{2}} \int \psi ( u_{1}^{2} - u_{2}^{2} ) \xi^{2}, \label{eq: before Yuk 2}
\end{align}
Then, testing (\ref{eq: psi R Yuk eq}) with $\psi \xi^{2}$ and re-arranging and using $a_{1} \geq 0$ gives
\begin{align}
\int |\nabla (\psi \xi)|^{2} \leq \int |\nabla (\psi \xi)|^{2} + a_{1}^{2} \int \psi^{2} \xi^{2} \leq \int R \psi \xi^{2} + 4\pi \int \psi (u_{2}^{2} - u_{1}^{2}) \xi^{2}. \label{eq: before psi Yuk eq}
\end{align}
Combining (\ref{eq: before Yuk 2}) and (\ref{eq: before psi Yuk eq}) and further re-arrangement yields
\begin{align}
\langle w \xi, L (w \xi) \rangle + \nu \int w^{2} \xi^{2} + \frac{1}{8 \pi} \int |\nabla \psi|^{2} \xi^{2} \leq C \left( \int R \psi \xi^{2} + \int ( w^{2} + \psi^{2} ) |\nabla \xi|^{2} \right). \label{eq: before Yuk 3}
\end{align}
From this point, the proof of \cite[Lemma 6.7]{Paper1} follows verbatim to show the estimate: there exists $C,\gamma > 0$ such that for all $\xi \in H_{\gamma}$ 
\begin{align}
\int_{\R} \bigg( \sum_{|\alpha_{1}| \leq k+4} |\partial^{\alpha_{1}} w|^{2} + \sum_{|\alpha_{2}| \leq k+2} |\partial^{\alpha_{2}} \psi|^{2} \bigg) \xi^{2} \leq C \int_{\R} \sum_{|\beta| \leq k} |\partial^{\beta} R|^{2} \xi^{2}. \label{eq: w and psi Yuk k lemma est}
\end{align}
If $k' = k$, then this is the desired estimate (\ref{eq: w and psi R Yuk lemma est}). Alternatively, if $k' = k+ 2$, the remaining estimate is shown by adapting the proof of \cite[Lemma 6.6]{Paper1}. Recall (\ref{eq: psi R Yuk eq}), that $\psi$ solves
\begin{align}
- \Delta \psi = - a_{1}^{2} \psi + 4 \pi \left( u_{2}^{2} - u_{1}^{2} \right) + R \in H^{k+2}_{\unif}(\R), \label{eq: psi Yuk k' eq}
\end{align}
hence by standard elliptic regularity \cite{Evans} $\psi \in H^{k+4}_{\unif}(\R)$. It follows that
\begin{align}
\int \sum_{|\alpha| \leq k+2} |\partial^{\alpha} \Delta \psi|^{2} \xi^{2} \leq C(k',M,\omega) \int \sum_{|\beta| \leq k+2} \left( |\partial^{\beta} \psi|^{2} + |\partial^{\beta} R|^{2} + |\partial^{\beta} w|^{2} \right) \xi^{2}. \label{eq: k' est 2}
\end{align}
In addition, applying integration by parts, for any $k_{1} \leq k + 2$
\begin{align}
\sum_{|\alpha| = k_{1} + 2} &\int | \partial^{\alpha} \psi|^{2} \xi^{2} \leq C \bigg( \int \sum_{|\beta_{1}| = k_{1}} |\partial^{\beta_{1}} \Delta \psi|^{2} \xi^{2} + \int \sum_{|\beta_{2}| = k_{1}+1} |\partial^{\beta_{2}} \psi|^{2} \xi^{2} \bigg), \label{eq: k' est 3}
\end{align}
hence combining (\ref{eq: w and psi Yuk k lemma est})--(\ref{eq: k' est 3}) for $k_{1} = k+2$ gives
\begin{align}
\sum_{|\alpha| = k + 4} \int | \partial^{\alpha} \psi|^{2} \xi^{2} &\leq C \bigg( \int \sum_{|\beta_{1}| = k+2} |\partial^{\beta_{1}} \Delta \psi|^{2} \xi^{2} + \int \sum_{|\beta_{2}| = k+3} |\partial^{\beta_{2}} \psi|^{2} \xi^{2} \bigg) \nonumber \\
&\leq C \bigg( \int \sum_{|\beta_{1}| = k+2} |\partial^{\beta_{1}} \Delta \psi|^{2} \xi^{2} + \int \sum_{|\beta_{2}| = k+2} |\partial^{\beta_{2}} \psi|^{2} \xi^{2} \bigg) \nonumber \\
&\leq C \int \sum_{|\beta| \leq k+2} \left( |\partial^{\beta} \psi|^{2} + |\partial^{\beta} R|^{2} + |\partial^{\beta} w|^{2} \right) \xi^{2} \nonumber \\ &\leq C \int_{\R} \sum_{|\beta| \leq k+2} |\partial^{\beta} R|^{2} \xi^{2}. \label{eq: k' est final}
\end{align}
Inserting (\ref{eq: k' est final}) into (\ref{eq: w and psi Yuk k lemma est}) yields the desired estimate (\ref{eq: w and psi R Yuk lemma est})
\begin{align*}
\int_{\R} \bigg( \sum_{|\alpha_{1}| \leq k+4} |\partial^{\alpha_{1}} w|^{2} + \sum_{|\alpha_{2}| \leq k'} |\partial^{\alpha_{2}} \psi|^{2} \bigg) \xi^{2} \leq C \int_{\R} \sum_{|\beta| \leq k'} |\partial^{\beta} R|^{2} \xi^{2}.
\end{align*}
Let $y \in \R$, then applying (\ref{eq: w and psi Yuk k lemma est}) with $\xi(x) = e^{-\gamma|x - y|} \in H_{\gamma}$ and following the proof of \cite[Lemma 6.6]{Paper1} yields the remaining estimate (\ref{eq: w and psi pointwise rhs exp integral final lemma est}).

\emph{Case 2.} Suppose (A) holds, then by \Cref{Proposition - Finite Yukawa Regularity Est}
\begin{align*}
\| u_{1} \|_{H^{4}_{\unif}(\R)} + \| \phi_{1} \|_{H^{2}_{\unif}(\R)} &\leq C(M), \\
\| u_{2} \|_{H^{4}_{\unif}(\R)} + \| \phi_{2} \|_{H^{2}_{\unif}(\R)} &\leq C(M'),
\end{align*}
and $\inf u_{1} \geq c_{a',M,\omega} > 0$ (if $a_{1} = 0$ then $\inf u_{1} \geq c_{M,\omega} > 0$) and $u_{2} \geq 0$. Other than this, the argument of Case 1 holds verbatim to obtain (\ref{eq: w and psi R Yuk lemma est})--(\ref{eq: w and psi pointwise rhs exp integral final lemma est}).
\end{proof}

\begin{proof}[Proof of \Cref{Corollary - Yukawa Coulomb Comparison 2}]
As $m \in \mathcal{M}_{H^{k}}(M,\omega)$, applying \Cref{Lemma - Exp Est Yuk Main Est}(B) with $0 < a_{1} \leq a_{2} \leq a_{0}$ and $R = (a_{2}^{2} - a_{1}^{2}) \phi_{2} \in H^{k+2}_{\unif}(\R)$. Then applying \Cref{Lemma - Exp Est Yuk Main Est} case (B) with $\xi(x) = e^{-\gamma |x - y|} \in H_{\gamma}$ yields
\begin{align*}
\sum_{|\alpha| \leq k+2} \left( |\partial^{\alpha} w(y)|^{2} + |\partial^{\alpha} \psi(y)|^{2} \right) \leq C (a_{2}^{2} - a_{1}^{2}) \int_{\R} \sum_{|\beta| \leq k+2} |\partial^{\beta} \phi_{2}(x)|^{2} e^{-2\gamma|x- y|} \id x.
\end{align*}
As $\phi_{2} \in H^{k+2}_{\unif}(\R)$, and for all $z \in \R$ and $A \subset B_{1}(z)$, $\sup_{x \in A} e^{-2\gamma|x|} \leq C \inf_{x \in A} e^{-2\gamma|x|}$, it follows that
\begin{align*}
\sum_{|\alpha| \leq k+2} \left( |\partial^{\alpha} w(y)|^{2} + |\partial^{\alpha} \psi(y)|^{2} \right) &\leq C (a_{2}^{2} - a_{1}^{2}) \int_{\R} \sum_{|\beta| \leq k+2} |\partial^{\beta} \phi_{2}(x)|^{2} e^{-2\gamma|x- y|} \id x \\ &\leq C (a_{2}^{2} - a_{1}^{2}) \| \phi_{2} \|_{H^{k+2}_{\unif}(\R)}^{2} \int_{\R} e^{-2\gamma|x- y|} \id x \leq C (a_{2}^{2} - a_{1}^{2}),
\end{align*}
where the final constant is independent of $y \in \R$, hence the desired estimate (\ref{eq: Yukawa Yukawa pointwise k pos est}) holds.
\end{proof}

\begin{proof}[Proof of \Cref{Theorem - Yukawa Coulomb Comparison}]
For $0 < a \leq a_{0}$, applying \Cref{Corollary - Yukawa Coulomb Comparison 2} with $a_{1} = 0, a_{2} = a$ yields the desired estimate (\ref{eq: Yukawa Coulomb pointwise k pos est}).
\end{proof}

\begin{proof}[Proof of \Cref{Theorem - Exponential Est Integral Yuk RHS}]
Let $0 < a \leq a_{0}$, then as $m_{1}, m_{2} \in \mathcal{M}_{H^{k}}(M,\omega)$ for $k \in \mathbb{N}_{0}$, applying \Cref{Lemma - Exp Est Yuk Main Est}(B) with $a_{1} = a_{2} = a$ and $R = 4\pi (m_{1} - m_{2}) \in H^{k}_{\unif}(\R)$ yields the desired estimate (\ref{eq: w and psi partial xi global est}).
\end{proof}

\begin{proof}[Proof of \Cref{Theorem - One inf pointwise stability estimate Yuk alt}]
The proof closely follows and adapts the argument used to show \cite[Theorem 3.4]{Paper1}.

As $m_{1} \in \mathcal{M}_{L^{2}}(M,\omega)$, by \Cref{Proposition - Uniform Yuk inf u estimate all a} for all $a > 0$ there exists 
a unique ground state $(u_{1,a},\phi_{1,a})$ corresponding to $m_{1}$. It remains to show that $m_{2}$ and its corresponding solution satisfy the conditions of \Cref{Lemma - Exp Est Yuk Main Est}(A).

\emph{Case 1.} Suppose $\spt(m_{2})$ is bounded and $m_{2} \not\equiv 0$. Since $m_{2} \in L^{2}_{\unif}(\R)$, it follows that $m_{2} \in L^{1}(\R)$ and since $m_{2} \geq 0$ and $m_{2} \not\equiv 0$, it follows that $\int m_{2} > 0$. For $a > 0$, consider the minimisation problem
\begin{align*}
I^{\TFW}_{a}(m_{2}) = \inf \left\{ \, E^{\TFW}_{a}(v,m_{2})  \, \bigg| \, v \in H^{1}(\R), v \geq 0\, \right \}.
\end{align*}
By \Cref{Proposition - Finite Yukawa Regularity Est}, there exists $a_{0} = a_{0}(m_{2}) > 0$ such that for all $0 < a \leq a_{0}$, the minimisation problem yields a unique solution $(u_{2,a},\phi_{2,a})$ of (\ref{eq: u phi Yuk eq pair}), satisfying $u_{2,a} > 0$ and (\ref{eq: u2 phi2 reg est})
\begin{align*}
\| u_{2,a} \|_{H^{4}_{\unif}(\R)} + \| \phi_{2,a} \|_{H^{2}_{\unif}(\R)} &\leq C(M'),
\end{align*}
independently of $a$. Consequently, applying \Cref{Lemma - Exp Est Yuk Main Est}(A) with $0 < a_{1} = a_{2} \leq a' \leq 1$ and $R = 4\pi( m_{1} - m_{2}) \in H^{k}_{\unif}(\R)$ yields the desired estimate (\ref{eq: w and psi partial xi global onesided est}).

\emph{Case 2.} Suppose $m_{2} = u_{2} = \phi_{2} = 0$, then by definition $(u_{2},\phi_{2})$ solve (\ref{eq:u-phi-eq-pair}) and (A) is satisfied, so applying \Cref{Lemma - Exp Est Yuk Main Est}(A) with $0 < a_{1} = a_{2} \leq a' = 1$ and $R = 4\pi( m_{1} - m_{2}) \in H^{k}_{\unif}(\R)$ yields the desired estimate (\ref{eq: w and psi partial xi global onesided est}).

\emph{Case 3.} \label{proof-case3} Suppose $\spt(m_{2})$ is unbounded. By \Cref{Proposition - Finite Yukawa Regularity Est}, there exists $a_{0} = a_{0}(m_{2}) > 0$ such that for all $0 < a \leq a_{0}$, there exists $(u_{2,a},\phi_{2,a})$ solving (\ref{eq: u phi Yuk eq pair}) and satisfying $u_{2,a} \geq 0$. As it is not guaranteed that $u_{2,a} > 0$, it is not possible to apply \Cref{Lemma - Exp Est Yuk Main Est}(A) directly to compare $(u_{1,a},\phi_{1,a})$ with $(u_{2,a},\phi_{2,a})$. Instead, by following the proof of \Cref{Proposition - Finite Yukawa Regularity Est}, a thermodynamic limit argument is used to construct a sequence of functions $(u_{2,a,R_{n}},\phi_{2,a,R_{n}})$ which satisfy (A) for sufficiently large $R_{n}$ and converge to $(u_{2,a},\phi_{2,a})$ as $R_{n} \to \infty$.

Let $R_{n} \uparrow \infty$ and define $m_{2,R_{n}} := m_{2} \cdot \chi_{B_{R_{n}}(0)}$, then as $m_{2} \in L^{2}_{\unif}(\R)$, $m_{2} \geq 0$ and $m_{2} \not \equiv 0$, it follows that $m_{2,R_{n}} \in L^{1}(\R)$ and for sufficiently large $R_{n}$, $\int m_{2,R_{n}} > 0$. By \Cref{Proposition - Finite Yukawa Regularity Est}, there exists $R_{0} = R_{0}(m_{2}), a_{0} = a_{0}(m_{2}) > 0$ such that for all $R_{n} \geq R_{0}$ and $0 < a \leq a_{0}$ the minimisation problem
\begin{align*}
I^{\TFW}_{a}(m_{2,R_{n}}) = \inf \left\{ \, E^{\TFW}_{a}(v,m_{2,R_{n}})  \, \bigg| \, v \in H^{1}(\R), v \geq 0, \int_{\R} v^{2} = \int_{\R} m_{2,R_{n}} \, \right \},
\end{align*}
defines a unique solution $(u_{2,a,R_{n}},\phi_{2,a,R_{n}})$ to (\ref{eq: u phi Yuk eq pair}), satisfying $u_{2,a,R_{n}} > 0$ and
\begin{align}
\| u_{2,a,R_{n}} \|_{H^{4}_{\unif}(\R)} + \| \phi_{2,a,R_{n}} \|_{H^{2}_{\unif}(\R)} &\leq C(M'), \label{eq: u 2Rn unif est}
\end{align}
where the constant is independent of $a$, $a_{0}$ and $R_{n}$. Passing to the limit in (\ref{eq: u 2Rn unif est}), there exist $u_{2,a} \in H^{4}_{\unif}(\R), \phi_{2,a} \in H^{2}_{\unif}(\R)$ such that, respectively, along a subsequence $u_{2,a,R_{n}}, \phi_{2,a,R_{n}}$ converges to $u_{2,a}, \phi_{2,a}$, weakly in $H^{4}(B_{R}(0))$ and $H^{2}(B_{R}(0))$, strongly in $H^{2}(B_{R}(0))$ and $L^{2}(B_{R}(0))$ for all $R>0$ and for all $|\alpha| \leq 2$, $\partial^{\alpha} u_{2,a,R_{n}}, \phi_{2,a,R_{n}}$ converges to $\partial^{\alpha} u_{2,a}, \phi_{2,a}$ pointwise. It follows that $(u_{2,a}, \phi_{2,a})$ is a solution of (\ref{eq: u phi Yuk eq pair}) corresponding to $m_{2}$, satisfying $u_{2,a} \geq 0$ and (\ref{eq: u2 phi2 reg est})
\begin{equation*}
\| u_{2,a} \|_{H^{4}_{\unif}(\R)} + \| \phi_{2,a} \|_{H^{2}_{\unif}(\R)} \leq C(M').
\end{equation*}
In addition, for $0 < a \leq a' = a_{0}$, $(u_{1}',\phi_{1}') = (u_{1,a},\phi_{1,a})$ and $(u_{2}',\phi_{2}') = (u_{2,a,R_{n}},\phi_{2,a,R_{n}})$ satisfy (A) for all $R_{n} \geq R_{0}$, so by \Cref{Lemma - Exp Est Yuk Main Est} that there exist $C,\gamma > 0$, independent of $a$, $a_{0}$ and $R_{n}$, such that for $R_{n} \geq R_{0}$ and any $\xi \in H_{\gamma}$
\begin{align}
\int_{\R} \bigg( \sum_{|\alpha_{1}| \leq 4} |\partial^{\alpha_{1}} (u_{1,a} - u_{2,a,R_{n}})|^{2} &+ \sum_{|\alpha_{2}| \leq 2} |\partial^{\alpha_{2}} (\phi_{1,a} - \phi_{2,a,R_{n}})|^{2} \bigg) \xi^{2} \nonumber \\ &\leq C \int_{\R} (m_{1} - m_{2,R_{n}})^{2} \xi^{2}, \label{eq: u1 u2Rn xi est}
\end{align}
and for any $y \in \R$,
\begin{align}
\sum_{|\alpha_{1}| \leq 2} |\partial^{\alpha_{1}} (u_{1,a} - u_{2,a,R_{n}})(y)|^{2} &+ |(\phi_{1,a} - \phi_{2,a,R_{n}})(y)|^{2} \nonumber \\ &\leq C \int_{\R} |(m_{1} - m_{2,R_{n}})(x)|^{2} e^{-2\gamma |x - y|} \id x. \label{eq: u1 u2Rn exp est}
\end{align}
Using the pointwise convergence of $(u_{2,a,R_{n}},\phi_{2,a,R_{n}})$ to $(u_{2,a},\phi_{2,a})$, applying the Dominated Convergence Theorem and sending $R_{n} \to \infty$ in (\ref{eq: u1 u2Rn xi est})--(\ref{eq: u1 u2Rn exp est}) gives the desired estimates (\ref{eq: w and psi partial xi global onesided est})--(\ref{eq: w and psi pointwise rhs exp integral onesided est}).
\end{proof}

\subsection{Proof of Applications}

Proving \Cref{Theorem - Forcing Yuk and Coulomb Comparison} first requires establishing the existence, uniqueness and regularity of solutions to the linearised TFW Yukawa equations.

Fix $Y = (Y_j)_{j \in \mathbb{N}} \in \mathcal{Y}_{L^{2}}(M,\omega)$ and let
$m = m_{Y} \in \mathcal{M}_{L^{2}}(M,\omega)$.  Let
$V \in \R \smallsetminus \{0\}$, $k \in \mathbb{N}$ and for $h \in [0,1]$ define
\begin{align}
Y^{h} = \{ \, Y_{j} + \delta_{jk} h V  \, | \, j \in \mathbb{N} \, \}, \label{eq: Y h def}
\end{align}
and the associated nuclear configuration
\begin{align}
m_{h}(x) = m(x) + \eta(x - Y_{k} - h V) - \eta(x - Y_{k}). \label{eq: mh def}
\end{align}

By \cite[Lemma 6.7]{Paper1}, there exist $(M',\omega')$ such that $m_{h} \in \mathcal{Y}_{L^{2}}(M',\omega')$ for all $h \in [0,1]$, hence by \Cref{Proposition - Yukawa Regularity Est All a} for all $a > 0$ there exists a corresponding ground state $(u_{a,h},\phi_{a,h})$. Also, let $(u_{a},\phi_{a}) = (u_{a,0},\phi_{a,0})$. \Cref{Corollary - Exponential Estimates Yuk Consequences} is now used to compare $(u_{a,h},\phi_{a,h})$ with $(u_{a},\phi_{a})$ to rigorously linearise the TFW Yukawa equations.

\begin{lemma}
\label{Lemma - Linearised TFW Results}
Let $a_{0} > 0$, $Y \in \mathcal{Y}_{L^{2}}(M,\omega)$ and let $m = m_{Y} \in \mathcal{M}_{L^{2}}(M,\omega)$. Also, let $k \in \mathbb{N}$, $V \in \R \smallsetminus\{0\}$ and $h_{0} = \min\{ 1, |V|^{-1}\}$. For $h \in [0,h_{0}]$ define
\begin{align*}
m_{h}(x) = m(x) + \eta(x - Y_{k} - h V) - \eta(x - Y_{k}),
\end{align*}
then for all $0 < a \leq a_{0}$ and $h \in [0,h_{0}]$ there exists a unique Yukawa ground state $(u_{a,h},\phi_{a,h})$ corresponding to $m_{h}$. There exist $C = C(a_{0},M',\omega')$, $\gamma_{0} = \gamma_{0}(a_{0},M',\omega') > 0$, independent of $a$, $h$ and $|V|$, such that for all $0 < a \leq a_{0}$ and $h \in [0,h_{0}]$
\begin{align}
\sum_{|\alpha| \leq 2} \big( |\partial^{\alpha}(u_{a,h} - u_{a})(x)| + |\partial^{\alpha}(\phi_{a,h} - \phi_{a})(x)| \big) + |(m_{h} - m)(x)| &\leq C h e^{-\gamma |x - Y_{k}|}, \label{eq: uh - u0 exp est}
\\ \|u_{a,h} - u_{a}\|_{H^{4}(\R)} + \|\phi_{a,h} - \phi_{a}\|_{H^{2}(\R)} \leq C \| m_{h} - m \|_{L^{2}(\R)} &\leq C h. \label{eq: uh - u0 L2 est}
\end{align}
Moreover, for all $0 < a \leq a_{0}$, the limits
\begin{align*}
\ou_{a} = \lim_{h \to 0} \frac{u_{a,h} - u_{a}}{h}, \quad \ophi_{a} = \lim_{h \to 0} \frac{\phi_{a,h} - \phi_{a}}{h}, \quad \om = \lim_{h \to 0} \frac{m_{h} - m}{h},
\end{align*}
exist and $(\ou_{a},\ophi_{a})$ is the unique solution to the linearised TFW Yukawa equations
\begin{subequations}
\label{eq: ou ophi Yuk eq pair}
\begin{align}
&- \Delta \ou_{a} + \left( \frac{35}{9} u^{4/3}_{a} - \phi_{a} \right) \ou_{a} - u_{a} \ophi_{a} = 0, \label{eq: ou Yuk eq} \\
&- \Delta \ophi_{a} + a^{2} \ophi_{a} = 4 \pi \left( \om - 2u_{a} \ou_{a} \right). \label{eq: ophi Yuk eq}
\end{align}
\end{subequations}
Moreover, $\ou_{a} \in H^{4}(\R), \ophi_{a} \in H^{2}(\R), \om \in C^{\infty}_{\textnormal{c}}(\R)$ and satisfy
\begin{align}
&\sum_{|\alpha| \leq 2} \big( |\partial^{\alpha}\ou_{a}(x)| + |\partial^{\alpha}\ophi_{a}(x)| \big) + |\om(x)| \leq C e^{- \gamma |x - Y_{k}|}, \label{eq: ou ophi Yuk exp est} \\
&\|\ou_{a}\|_{H^{4}(\R)} + \|\ophi_{a}\|_{H^{2}(\R)} \leq C \| \om \|_{L^{2}(\R)}, \label{eq: ou ophi Yuk L2 est}
\end{align}
where $C = C(a_{0},M',\omega')$, $\gamma_{0} = \gamma_{0}(a_{0},M',\omega') > 0$ are independent of $a$ and $|V|$.

\end{lemma}

%
%

\begin{proof}[Proof of \Cref{Lemma - Linearised TFW Results}]

The first step is to show the uniqueness of the linearised Yukawa solution $(\ou_{a},\ophi_{a})$ to (\ref{eq: ou ophi Yuk eq pair}). Let $0 < a \leq a_{0}$ and suppose $(w,\psi) \in H^{1}(\R) \times H^{1}(\R)$ solves
\begin{subequations}
\label{eq: w psi Yuk eq pair}
\begin{align}
&- \Delta w + \left( \frac{35}{9} u^{4/3}_{a} - \phi_{a} \right) w - u_{a} \psi = 0, \label{eq: w Yuk uniq eq} \\
&- \Delta \psi + a^{2} \psi =  - 8 \pi u_{a} \psi. \label{eq: psi Yuk uniq eq}
\end{align}
\end{subequations}
Testing (\ref{eq: w Yuk uniq eq}) with $w$ yields
\begin{align*}
\int_{\R} |\nabla w|^{2} + \int_{\R} \left( \frac{35}{9} u^{4/3}_{a} - \phi_{a} \right) w^{2} = \int_{\R} u_{a} w \psi.
\end{align*}
Then as $u_{a} > 0$, by \cite[Lemma 6.2]{Paper1} $L_{a} = - \Delta + \smfrac{35}{9} u^{4/3}_{a} - \phi_{a}$ is a non-negative operator. In addition, by \Cref{Proposition - Uniform Yuk inf u estimate all a} $\inf u_{a} \geq c_{a_{0}, M', \omega'} > 0$, hence there exists $c_{0} > 0$ such that
\begin{align}
c_{0} \int_{\R} w^{2} &\leq \frac{10}{9} \int_{\R} u^{4/3}_{a} w^{2} \leq \langle w , L_{a} w \rangle + \frac{10}{9} \int_{\R} u^{4/3}_{a} w^{2} \nonumber \\ &= \int_{\R} |\nabla w|^{2} + \int_{\R} \left( \frac{35}{9} u^{4/3}_{a} - \phi_{a} \right) w^{2} = \int_{\R} u_{a} w \psi. \label{eq: w Yuk c0 est}
\end{align}
Then testing (\ref{eq: psi Yuk uniq eq}) with $\smfrac{1}{8\pi} \psi$ gives
\begin{align}
\frac{1}{8\pi} \left( \int_{\R} |\nabla \psi|^{2} + a^{2} \int_{\R} \psi^{2} \right) = - \int_{\R} u_{a} w \psi, \label{eq: psi 8pi Yuk est}
\end{align}
and adding (\ref{eq: w Yuk c0 est})--(\ref{eq: psi 8pi Yuk est}) yields
\begin{align*}
0 \leq c_{0} \int_{\R} w^{2} + \frac{1}{8\pi} \left( \int_{\R} |\nabla \psi|^{2} + a^{2} \int_{\R} \psi^{2} \right) \leq 0,
\end{align*}
hence $w = \psi = 0$ almost everywhere, so (\ref{eq: ou ophi Yuk eq pair}) has a unique solution in $H^{1}(\R) \times H^{1}(\R)$.

Now, \Cref{Proposition - Yukawa Regularity Est} and \Cref{Proposition - Uniform Yuk inf u estimate} imply that for $0 < a \leq a_{0}$ and $h \in [0,h_{0}]$ the ground state $(u_{a,h},\phi_{a,h})$ satisfies
\begin{align}
&\| u_{a,h} \|_{H^{4}_{\unif}(\R)} + \| \phi_{a,h} \|_{H^{2}_{\unif}(\R)} \, \leq C(a_{0},M'), \label{eq: uh phih reg est} \\
&\inf_{x \in \R} u_{a,h}(x) \geq c_{a_{0},M',\omega'} > 0, \label{eq: inf uh est}
\end{align}
independently of $a$, $h$ and $|V|$. Then following the proof of \cite[Lemma 6.8]{Paper1}, for all $0 < a \leq a_{0}$ and $h \in [0,h_{0}]$, the estimates (\ref{eq: uh - u0 exp est})--(\ref{eq: uh - u0 L2 est}) hold. In addition, there exist $\ou_{a} \in H^{4}(\R)$ and $\ophi_{a} \in H^{2}(\R)$ such that along a subsequence $h_{n}$ (which may depend on $a$) such that $\smfrac{u_{a,h_{n}}- u_{a}}{h_{n}}, \smfrac{\phi_{a,h_{n}}- \phi_{a}}{h_{n}}$ converge to
$\ou_{a} \in H^{4}(\R), \ophi_{a} \in H^{2}(\R)$ respectively, weakly in $H^{4}(\R)$ and
$H^{2}(\R)$, strongly in $H^{3}(B_{R}(0))$ and $H^{1}(B_{R}(0))$ for all $R > 0$
and pointwise almost everywhere, along with their derivatives. In addition, it follows that $(\ou_{a},\ophi_{a})$ satisfy (\ref{eq: ou ophi Yuk exp est})--(\ref{eq: ou ophi Yuk L2 est}).

To verify that $(\ou_{a},\ophi_{a})$ are independent of the sequence chosen, passing to the limit in the equations
\begin{align*}
&- \Delta \left( \frac{ u_{a,h_{n}} - u_{a} }{h_{n}}  \right) + \frac{5}{3} \frac{ u_{a,h_{n}}^{7/3} - u_{a}^{7/3} }{h_{n}} - \frac{ \phi_{a,h_{n}} u_{a,h_{n}} - \phi_{a} u_{a} }{h_{n}}  = 0, \\
&- \Delta \left( \frac{ \phi_{a,h_{n}} - \phi_{a} }{h_{n}} \right) + a^{2} \left( \frac{ \phi_{a,h_{n}} - \phi_{a} }{h_{n}} \right) = 4 \pi \left( \frac{ m_{h_{n}} - m }{h_{n}} - \frac{ u_{a,h_{n}}^{2} - u_{a}^{2} }{h_{n}}  \right),
\end{align*}
gives that $(\ou_{a},\ophi_{a})$ solve the linearised Yukawa equations (\ref{eq: ou ophi Yuk eq pair}) pointwise, 
\begin{align*}
&- \Delta \ou_{a} + \left( \frac{35}{9} u_{a}^{4/3} - \phi_{a} \right) \ou_{a} - u_{a} \ophi_{a} = 0, \\
&- \Delta \ophi_{a} + a^{2} \ophi_{a} = 4 \pi \left( \om - 2u_{a} \ou_{a} \right), \\
  \text{where} &\quad
\om(x) = \lim_{h_{n} \to 0} \frac{(m_{h_{n}}- m)(x)}{h_{n}} = - \nabla \eta (x - Y_{k}) \cdot V.
\end{align*}
Clearly $\om$ is independent of the sequence $h_{n}$, so as $(\ou_{a}, \ophi_{a})$ is the unique solution to the linearised Yukawa system (\ref{eq: ou ophi Yuk eq pair}), it is independent of the sequence $(h_{n})$. It then follows that $\smfrac{u_{a,h}- u_{a}}{h}, \smfrac{\phi_{a,h}- \phi_{a}}{h}$ converge to $\ou_{a}, \ophi_{a}$ as $h \to 0$ as stated above.
\end{proof}


\begin{proof}[Proof of \Cref{Theorem - Forcing Yuk and Coulomb Comparison}]
Let $0 < a \leq a_{0}$ and $h \in [0,h_{0}]$, then recall (\ref{eq: TFW energy 2})
\begin{align*}
\mathcal{E}_{2,a}(Y^{h};\cdot) = |\nabla u_{a,h}|^{2} + u_{a,h}^{10/3} + \smfrac{1}{8\pi} \left(  |\nabla \phi_{a,h}|^{2} + a^{2} \phi_{a,h}^{2} \right).
\end{align*}
Applying \Cref{Lemma - Linearised TFW Results} and using the pointwise convergence of $u_{a,h},\phi_{a,h}, \smfrac{u_{a,h} - u_{a}}{h}, \smfrac{\phi_{a,h} - \phi_{a}}{h}$ to $u_{a},\phi_{a},\ou_{a},\ophi_{a}$ as $h \to 0$, along with their derivatives, it follows that
\begin{align*}
\frac{\mathcal{E}_{2,a}(Y^{h};\cdot) - \mathcal{E}_{2,a}(Y;\cdot)}{h} \to 2 \nabla u_{a} \cdot \nabla \ou_{a} + \frac{10}{3} u_{a}^{7/3} \ou_{a} + \frac{1}{4\pi} \left( \nabla \phi_{a} \cdot \nabla \ophi_{a} + a^{2} \phi_{a} \ophi_{a} \right).
\end{align*} 
As $u_{a} \in W^{1,\infty}(\R)$, $\phi_{a} \in L^{\infty}(\R)$ and $\nabla \phi_{a} \in L^{2}_{\unif}(\R)$ and (\ref{eq: ou ophi Yuk exp est}) holds
\begin{align*}
\sum_{|\alpha| \leq 2} \big( |\partial^{\alpha}\ou_{a}(x)| + |\partial^{\alpha}\ophi_{a}(x)| \big) + |\om(x)| \leq C e^{- \gamma_{0} |x - Y_{k}|},
\end{align*}
it follows that $\partial_{Y_{k}} \mathcal{E}_{2,a} \in L^{1}(\R)$ and
\begin{align}
\int_{\R} \frac{\partial \mathcal{E}_{2,a}(Y;x)}{\partial Y_{k}} \id x &= 2 \int_{\R} \nabla u_{a} \cdot \nabla \ou_{a} + \frac{10}{3} \int_{\R} u_{a}^{7/3} \ou_{a} \nonumber \\ & \quad + \frac{1}{4\pi} \int_{\R} \left( \nabla \phi_{a} \cdot \nabla \ophi_{a} + a^{2} \phi_{a} \ophi_{a} \right). \label{eq: E2 Yuk partial int}
\end{align}
An identical argument shows that $\partial_{Y_{k}} \mathcal{E}_{1,a} \in L^{1}(\R)$ and
\begin{align}
\int_{\R} \frac{\partial \mathcal{E}_{1,a}(Y;x)}{\partial Y_{k}} \id x &= 2 \int_{\R} \nabla u_{a} \cdot \nabla \ou_{a} + \frac{10}{3} \int_{\R} u_{a}^{7/3} \ou_{a} \nonumber \\ & \quad + \frac{1}{2} \int_{\R} \left( \phi_{a} ( \om - 2 u_{a} \ou_{a} ) + \ophi_{a} ( m - u_{a}^{2} ) \right). \label{eq: E1 Yuk partial int}
\end{align}
Using that $\phi_{a}$ and $\ophi_{a}$ solve (\ref{eq:phi-Yuk-inf-eq}) and (\ref{eq: ophi Yuk eq}), respectively,
\begin{align}
\frac{1}{2} \int_{\R} \ophi_{a} ( m - u_{a}^{2} ) &= \frac{1}{8 \pi} \int_{\R} \ophi_{a} ( - \Delta \phi_{a} + a^{2} \phi_{a} ) = \frac{1}{8\pi} \int_{\R} \left( \nabla \phi_{a} \cdot \nabla \ophi_{a} + a^{2} \phi_{a} \ophi_{a} \right) \nonumber \\ &= \frac{1}{8 \pi} \int_{\R} \phi_{a} ( - \Delta \ophi_{a} + a^{2} \ophi_{a} ) = \frac{1}{2} \int_{\R} \phi_{a} ( \om - 2u_{a} \ou_{a} ). \label{eq: E1 E2 Yuk equiv}
\end{align}
Combining (\ref{eq: E2 Yuk partial int})--(\ref{eq: E1 E2 Yuk equiv}) and using that $u_{a}$ solves (\ref{eq:u-Yuk-inf-eq}), $- \Delta u_{a} + \frac{5}{3} u_{a}^{7/3} - \phi_{a} u_{a} = 0$, the estimate (\ref{eq: Yuk force equivalence}) follows
\begin{align*}
&\int_{\R} \frac{\partial \mathcal{E}_{1,a}(Y;x)}{\partial Y_{k}} \id x = \int_{\R} \frac{\partial \mathcal{E}_{2,a}(Y;x)}{\partial Y_{k}} \id x \\ &= 2 \left( \int_{\R} \nabla u_{a} \cdot \nabla \ou_{a} + \frac{5}{3} \int_{\R} u_{a}^{7/3} \ou_{a} - \int_{\R} \phi_{a} u_{a} \ou_{a} \right) + \int_{\R} \phi_{a} \om = \int_{\R} \phi_{a} \om.
\end{align*}
Now recall the corresponding result for the Coulomb case \cite[(4.21)]{Paper1}, that $\partial_{Y_{k}} \mathcal{E}_{1}$, $\partial_{Y_{k}} \mathcal{E}_{2} \in L^{1}(\R)$ and
\begin{align*}
\int_{\R} \frac{\partial \mathcal{E}_{1}(Y;x)}{\partial Y_{k}} \id x = \int_{\R} \frac{\partial \mathcal{E}_{2}(Y;x)}{\partial Y_{k}} \id x = \int_{\R} \phi \,\om.
\end{align*}
Applying (\ref{eq: Yukawa Coulomb pointwise k pos est}) of \Cref{Theorem - Yukawa Coulomb Comparison} and (\ref{eq: ou ophi Yuk exp est}) of \Cref{Lemma - Linearised TFW Results} yields the desired estimate (\ref{eq: Yuk and Coulomb force comparison est}), for $i \in \{1,2\}$
\begin{align*}
&\left| \int_{\R} \left( \frac{\partial \mathcal{E}_{i,a}}{\partial Y_{k}} - \frac{\partial \mathcal{E}_{i}}{\partial Y_{k}} \right)(Y;x) \id x \right| \\ &\quad \qquad \leq \int_{\R} | \phi_{a} - \phi | | \om| \leq C \| \phi_{a} - \phi \|_{L^{\infty}(\R)} \int_{\R} e^{- \gamma |x - Y_{k}|} \id x \leq Ca^{2}. \qedhere
\end{align*}

\end{proof}

\begin{proof}[Proof of \Cref{Proposition - Infinite Finite Ground state comparison}]
This holds directly from applying \Cref{Theorem - One inf pointwise stability estimate Yuk alt} and following the proof of \cite[Proposition 4.1]{Paper1} verbatim.
\end{proof}

\begin{proof}[Proof of \Cref{Corollary - Exponential Estimates Yuk Consequences}]
This holds directly from applying \Cref{Theorem - Exponential Est Integral Yuk RHS} and following the proof of \cite[Corollary 4.2]{Paper1} verbatim.
\end{proof}

\begin{proof}[Proof of \Cref{Corollary - Neutrality Yuk Estimate}]
This holds directly from applying \Cref{Theorem - Exponential Est Integral Yuk RHS} with $k = 0$ and following the proof of \cite[Theorem 4.3]{Paper1} verbatim.
\end{proof}

\section{Appendix}
\label{Appendix}

The purpose of this section is to prove \Cref{Proposition - Uniform Yuk inf u estimate remainder}.
%
\Proposition*

%

The proof of \Cref{Proposition - Uniform Yuk inf u estimate remainder} adapts the argument described in \cite[Remark 4.16,  Lemma 4.14]{C/LB/L}, which shows that the periodic Yukawa ground state is bounded below and hence unique. The proof requires the following result.

\begin{lemma}
\label{Lemma - Yukawa Finite Lower Bound}
For any $a_{0} > 0$ and $m \in \mathcal{M}_{L^{2}}(M,\omega)$, there exists $R_{0} = R_{0}(a_{0},\omega)$, $\nu_{a_{0},M,\omega} >0$ such that for all $0 < a \leq a_{0}$ and $R_{n} \geq R_{0}$
%
\begin{align}
\inf_{x \in B_{1}(0)} u_{a,R_{n}}(x) \geq \nu_{a_{0},M,\omega} > 0. \label{eq: Yuk finite inf est}
\end{align}
\end{lemma}

Then, sending $R_{n} \to \infty$ in (\ref{eq: Yuk finite inf est}), it follows that for all $0 < a \leq a_{0}$
\begin{align*}
\inf_{x \in B_{1}(0)} u_{a}(x) \geq \nu_{a_{0},M,\omega} > 0,
\end{align*} 
hence $u_{a} > 0$. Then following the proof of \cite[Lemma 4.14]{C/LB/L} gives the desired estimate (\ref{eq:u-Yuk-inf-est-remainder}). As the argument used in \cite[Lemma 4.14]{C/LB/L} is also necessary to show \Cref{Lemma - Yukawa Finite Lower Bound}, it is followed closely in this instance and for the proof of \Cref{Proposition - Uniform Yuk inf u estimate remainder}, only the necessary changes in the argument are described.

\begin{proof}[Proof of \Cref{Lemma - Yukawa Finite Lower Bound}]
It is first shown that there exists $R_{0}' > 0$ such that for all $0 < a \leq a_{\rm c}$
\begin{align}
\inf_{R_{n} \geq R_{0}'} \inf_{m \in \mathcal{M}_{L^{2}}(M,\omega)} \inf_{x \in B_{1}(0)} u_{a,R_{n},m}(x) \geq \frac{c_{a_{\rm c}, M, \omega}}{2} > 0, \label{eq: u pos contr initial assumption ac}
\end{align}
then it remains to show that there exists $R_{0} > 0$ such that for all $a_{\rm c} < a \leq a_{0}$
\begin{align}
\inf_{R_{n} \geq R_{0}} \inf_{m \in \mathcal{M}_{L^{2}}(M,\omega)} \inf_{x \in B_{1}(0)} u_{a,R_{n},m}(x) = c_{a_{0},M,\omega} > 0. \label{eq: u pos contr initial assumption true}
\end{align}

By \Cref{Proposition - Uniform Yuk inf u estimate}, for any $m \in \mathcal{M}_{L^{2}}(M,\omega)$, $0 < a \leq a_{\rm c}$, the Yukawa ground state electron density $u_{a}$ satisfies
\begin{align*}
\inf_{x \in \R} u_{a}(x) \geq c_{a_{\rm c}, M, \omega} > 0,
\end{align*}
and by \Cref{rem:rdl-estimate} following \Cref{Proposition - Infinite Finite Ground state comparison}
\begin{align*}
\| u_{a} - u_{a,R_{n}} \|_{L^{\infty}(B_1(0))} \leq C' e^{- \gamma ( R_{n} - 1 )}.
\end{align*}
It follows that (\ref{eq: u pos contr initial assumption ac}) holds for $R_{n} \geq R_{0}' := 1 + \gamma^{-1} \log(2C' c_{a_{\rm c},M,\omega}^{-1})$ and any $x \in B_{1}(0)$
\begin{align*}
u_{a, R_{n}}(x) \geq u_{a}(x) - C' e^{- \gamma ( R_{n} - 1 )} \geq c_{a_{\rm c}, M, \omega} - \frac{c_{a_{\rm c}, M, \omega}}{2} \geq \frac{c_{a_{\rm c}, M, \omega}}{2} > 0.
\end{align*}
The estimate (\ref{eq: u pos contr initial assumption true}) is shown by contradiction, so suppose that for all $R_{0} > 0$
\begin{align}
\inf_{a_{\rm c} < a \leq a_{0}} \inf_{R_{n} \geq R_{0}} \inf_{m \in \mathcal{M}_{L^{2}}(M,\omega)} \inf_{x \in B_{1}(0)} u_{a,R_{n},m}(x) = 0, \label{eq: u pos contr initial assumption}
\end{align}
where $u_{a,R_{n},m}$ solves (\ref{eq:u-Yuk-inf-eq}) corresponding to $m_{R_{n}} = m \cdot \chi_{B_{R_{n}}(0)}$. 

 Hence for each $k \in \mathbb{N}$ there exist sequences $(a_{k} ) \subset (a_{\rm c}, a_{0}]$, $R_{n_{k}} \uparrow \infty$, $\widetilde{ m }_{k} \in \mathcal{M}_{L^{2}}(M,\omega)$ and $x_{k} \in B_{1}(0)$ such that $m_{k,R_{n_{k}}} = \widetilde{ m }_{k} \cdot \chi_{B_{R_{n_{k}}}(0)}$ satisfies for all $k \in \mathbb{N}$
\begin{align*}
u_{a_{k},R_{n_{k}},\widetilde{ m}_{k} }(x_{k}) \leq \frac{1}{k}.
\end{align*}
For convenience, in this argument $u_{a_{k},R_{n_{k}},\widetilde{m}_{k}}$ and $m_{k,R_{n_{k}}}$ are referred to as $u_{k}$ and $m_{k}$, respectively. By the Harnack inequality, for fixed $k \in \mathbb{N}$ and any $R' \geq 1$ there exists $C(R',a_{0},M) > 0$ such that
\begin{align}
\sup_{x \in B_{R'}(0)} u_{k}(x) \leq C \inf_{x \in B_{R'}(0)} u_{k}(x) \leq \frac{C(R',a_{0},M)}{k}, \label{eq: uj Harnack est}
\end{align}
so it follows that $u_{k}$ converges uniformly to $0$ on any compact subset as $k \to \infty$. For $R > 0$ and $k \in \mathbb{N}$, define the energy functional acting on $v$ satisfying $\nabla v \in L^{2}(B_{R}(0))$ and $v \in L^{10/3}(B_{R}(0))$ by
\begin{align}
E(v;k,R) &= \int_{B_{R}(0)} |\nabla v|^{2} + \int_{B_{R}(0)} v^{10/3} - \int_{B_{R}(0)} \left( m_{k} * Y_{a_{k}} \right) v^{2}  \bigg. \nonumber \\
& \quad + \frac{1}{2} \int_{B_{R}(0)} \left( v^{2} \cdot \chi_{B_{R}(0)} * Y_{a_{k}} \right) v^{2} + \int_{B_{R}(0)} \left( u_{k}^{2} \cdot \chi_{B_{R}(0)^{\rm c}} * Y_{a_{k}} \right) v^{2}. \label{eq: E k,R def}
\end{align}
Then consider the corresponding variational problem
\begin{align}
I(k,R) &= \inf \bigg \{ E(v;k,R) \, \bigg| \, \, \nabla v \in L^{2}(B_{R}(0)), v \in L^{10/3}(B_{R}(0)),  v|_{\partial B_{R}(0)} = u_{k} \, \bigg \}. \label{eq: I k,j,R variational problem}
\end{align}
The construction of the energy and the boundary condition of (\ref{eq: I k,j,R variational problem}) ensures that $u_{k}$ is the unique minimiser of (\ref{eq: I k,j,R variational problem}) for each $R > 0$. To prove this, observe that $E(v;k,R)$ can be expressed as
\begin{align*}
E(v;k,R) &= \int_{B_{R}(0)} |\nabla v|^{2} + \int_{B_{R}(0)} v^{10/3} + \int_{B_{R}(0)} \left( u_{k}^{2} \cdot \chi_{B_{R}(0)^{\rm c}} * Y_{a_{k}} \right) v^{2} \\
& \quad + \frac{1}{2} D_{a_{k}} \left( m_{k} - v^{2} \chi_{B_{R}(0)}, m_{k} - v^{2} \chi_{B_{R}(0)} \right) - \frac{1}{2} D_{a_{k}} \left( m_{k}, m_{k} \right).
\end{align*}
As $Y_{a_{k}}$ and the Yukawa interaction term are non-negative, it follows that
\begin{align*}
E(v;k,R) \geq \int_{B_{R}(0)} |\nabla v|^{2} + \int_{B_{R}(0)} v^{10/3} - \frac{1}{2} D_{a_{k}} \left( m_{k}, m_{k} \right) \geq - \frac{1}{2} D_{a_{k}} \left( m_{k}, m_{k} \right)
> - \infty,
\end{align*}
so as $E(v;k,R)$ is bounded below, $I(k,R)$ is well-defined. Any minimising sequence $v_{n}$ satisfies
\begin{align*}
\| \nabla  v_{n} \|_{L^{2}(B_{R}(0))}^{2} + \| v_{n} \|_{L^{10/3}(B_{R}(0))}^{10/3} \leq C(k,R,a_{0},M),
\end{align*}
hence there exists $v_{k,R}$ such that $\nabla v_{k,R} \in L^{2}(\R), v_{k,R} \in L^{10/3}(\R)$. Moreover, along a subsequence $\nabla v_{n}$ converges to $\nabla v_{k,R}$ weakly in $L^{2}(\R)$, $v_{n}$ converges to $v_{k,R}$, weakly in $L^{6}(\R)$ and $L^{10/3}(\R)$, strongly in $L^{p}(B_{R}(0))$ for all $p \in [1,6)$ and $R > 0$ and pointwise almost everywhere. Moreover, $v_{k,R}$ satisfies
\begin{align*}
E(v_{k,R};k,R) = I(k,R),
\end{align*}
and solves
\begin{align}
- \Delta v_{k,R} + \frac{5}{3} v_{k,R}^{7/3} + &\left( m_{k} - v_{k,R}^{2} \cdot \chi_{B_{R}(0)} - u_{k}^{2} \cdot \chi_{B_{R}(0)^{\rm c}}  \right) v_{k,R} = 0, \label{eq: v k,j,R EL eq} \\
& \quad v_{k,R} = u_{k} \quad \text{ on } \partial B_{R}(0). \nonumber
\end{align}
It is straightforward to verify that $u_{k}$ solves (\ref{eq: v k,j,R EL eq}). Define the alternate minimisation problem
\begin{align}
\inf \left\{ \, E(\sqrt \rho;k,R)  \, \bigg| \, \nabla \sqrt \rho \in L^{2}(\R), \rho \in L^{5/3}(\R), \rho \geq 0 \, \right \}. \label{eq: rho min problem 2}
\end{align}
Due to the strict convexity of $\rho \mapsto E(\sqrt \rho;k,R)$, it follows that $\rho_{k} = u_{k}^{2}$ is the unique minimiser of (\ref{eq: rho min problem 2}), hence $u_{k}$ is the unique minimiser of (\ref{eq: I a Rn min problem}).

As $u_{k} \to 0$ uniformly as $k \to \infty$, it follows that for any fixed $R > 0$
\begin{align}
E(u_{k};k,R) \to 0 \quad \text{ as } \quad k \to \infty. \label{eq: E u k,j j limit zero est}
\end{align}
To verify (\ref{eq: E u k,j j limit zero est}), observe that
\begin{align*}
E(u_{k};k,R) &= \int_{B_{R}(0)} |\nabla u_{k}|^{2} + \int_{B_{R}(0)} u_{k}^{10/3} - \int_{B_{R}(0)} \left( m_{k} * Y_{a_{k}}  \right) u_{k}^{2}  \bigg. \nonumber \\
& \quad + \frac{1}{2} \int_{B_{R}(0)} \left( u_{k}^{2} \cdot \chi_{B_{R}(0)} * Y_{a_{k}} \right) u_{k}^{2} + \int_{B_{R}(0)} \left( u_{k}^{2} \cdot \chi_{B_{R}(0)^{c}} * Y_{a_{k}} \right) u_{k}^{2}.
\end{align*}
Clearly
\begin{align}
0 \leq \int_{B_{R}(0)} u_{k}^{10/3} \leq C R^{3} \| u_{k} \|_{L^{\infty}(B_{R}(0))}^{10/3} \to 0 \text{ as } k \to \infty. \label{eq: E to 0 1}
\end{align}
The term $m_{k} * Y_{a_{k}}$ can be estimated by
\begin{align}
\left \| m_{k} * Y_{a_{k}} \right \|_{L^{\infty}(\R)} \leq C(a_{\rm c},M), \label{eq: m k,j conv est}
\end{align}
where the constant $C(a_{\rm c},M)$ is independent of $k \in \mathbb{N}$. From (\ref{eq: m k,j conv est}) it follows that
\begin{align}
\left| \int_{B_{R}(0)} \left( m_{k,j} * Y_{a_{k}} \right) u_{k}^{2} \right| &\leq \left \| m_{k} * Y_{a_{k}} \right \|_{L^{\infty}(\R)} \int_{B_{R}(0)} u_{k}^{2} \nonumber \\
&\leq C a_{\rm c}^{-3} M R^{3}  \| u_{k} \|_{L^{\infty}(B_{R}(0))}^{2} \to 0 \text{ as } k \to \infty. \label{eq: E to 0 2}
\end{align}
To show (\ref{eq: m k,j conv est}), let $\Gamma \subset \R$ be a semi-open unit cube centred at the origin, so $\R = \{ \, \Gamma + i \, | \, i \in \mathbb{Z}^{3} \, \}$. For any  $x \in \R$
\begin{align}
\left |  \left(m_{k} * Y_{a_{k}} \right)(x) \right | &\leq  \int_{\R} |m_{k}(x-y)| \frac{ e^{-a_{k}|y|} }{|y|} \id y = \sum_{i \in \mathbb{Z}^{3}} \int_{\Gamma + i} |m_{k}(x-y)| \frac{ e^{-a_{k}|y|} }{|y|} \id y \nonumber \\
&\leq C \sum_{i \in \mathbb{Z}^{3}} \| m_{k} \|_{L^{2}_{\unif}(\R)} \left \| \smfrac{ e^{-a_{k}|\cdot|} }{|\cdot|} \right \|_{L^{2}(\Gamma + i)} \leq C M \sum_{i \in \mathbb{Z}^{3}} \left \| \smfrac{ e^{-a_{k}|\cdot|} }{|\cdot|} \right \|_{L^{2}(\Gamma + i)} \nonumber \\
&\leq C M \sum_{i \in \mathbb{Z}^{3}} e^{-a_{k}|i|} \leq \frac{CM}{a_{k}^{3}} \leq \frac{CM}{a_{\rm c}^{3}}. \label{eq: C ac M est}
\end{align} 
As the estimate (\ref{eq: C ac M est}) is independent of $k \in \mathbb{N}$ and $x \in \R$, (\ref{eq: m k,j conv est}) holds. Estimating the remaining terms gives
\begin{align}
\frac{1}{2} \int_{B_{R}(0)} \left( u_{k}^{2} \cdot \chi_{B_{R}(0)} * Y_{a_{k}} \right) u_{k}^{2} &\leq \| u_{k} \|_{L^{\infty}(B_{R}(0))}^{4} D_{a_{k}}(\chi_{B_{R}(0)},\chi_{B_{R}(0)}) \\ &\leq C a_{\rm c}^{-2} R^{3} \| u_{k} \|_{L^{\infty}(B_{R}(0))}^{4} \to 0 \text{ as } k \to \infty, \label{eq: uk Yak est}
\end{align}
\begin{align}
\int_{B_{R}(0)} \left( u_{k}^{2} \cdot \chi_{B_{R}(0)^{c}} * Y_{a_{k}} \right) u_{k}^{2} &\leq \left \|  u_{k}^{2} \cdot \chi_{B_{R}(0)^{c}} * Y_{a_{k}} \right \|_{L^{\infty}(\R)} \int_{B_{R}(0)} u_{k}^{2} \nonumber \\
&\leq C R^{3} \| u_{k} \|_{L^{\infty}(\R)}^{2} \left \| Y_{a_{k}} \right \| _{L^{1}(\R)} \| u_{k} \|_{L^{\infty}(B_{R}(0))}^{2} \nonumber \\
&\leq \frac{C(a_{0},M) R^{3}}{a_{\rm c}^{2}} \| u_{k} \|_{L^{\infty}(B_{R}(0))}^{2} \to 0 \text{ as } k \to \infty. \label{eq: E to 0 3}
\end{align}
For the final term, integration by parts yields
\begin{align}
\int_{B_{R}(0)} |\nabla u_{k} |^{2} &= - \int_{B_{R}(0)} u_{k} \Delta u_{k} + \int_{\partial B_{R}(0)} u_{k} \frac{\partial u_{k}}{\partial n} \nonumber \\
&\leq C \| u_{k} \|_{W^{2,\infty}(\R)} ( R^{3} + R^{2} ) \| u_{k} \|_{L^{\infty}\left(\overline{B_{R}(0)}\right)} \nonumber \\
&\leq C(a_{0},M) R^{3} \| u_{k} \|_{L^{\infty}(B_{R}(0))} \to 0 \text{ as } k \to \infty. \label{eq: E to 0 4}
\end{align}
Collecting (\ref{eq: E to 0 1})--(\ref{eq: E to 0 4}), it follows that for fixed $R > 0$, $E(u_{k};k,R) \to 0$ as $k \to \infty$.
A family of test functions $\varphi_{\varepsilon,k} \in H^{1}(B_{R}(0))$ is now constructed, satisfying the boundary condition $\varphi_{\varepsilon,k}|_{\partial B_{R}(0)} = u_{k}$ of (\ref{eq: I k,j,R variational problem}) such that for sufficiently large $R > 0$ and small $\varepsilon > 0$, there exists a constant $C_{1} > 0$ such that for all large $k \in \mathbb{N}$
\begin{align}
E(\varphi_{\varepsilon,k};k,R) \leq -C_{1} < 0, \label{eq: varphi epsilon contr statement}
\end{align}
contradicting the fact that $E(u_{k};k,R) \to 0$ as $k \to \infty$, as (\ref{eq: varphi epsilon contr statement}) implies
\begin{align*}
E(u_{k};k,R) \leq E(\varphi_{\varepsilon,k};k,R) \leq - C_{1} < 0.
\end{align*}
\Cref{Lemma - Technical Lemma 2} will be used to prove (\ref{eq: varphi epsilon contr statement}) by showing that there exists $R_{0}' \geq R_{0}$ and $k_{1} \in \mathbb{N}$ such that choosing $R_{n} = R_{0}'$ and $k \geq k_{1}$ ensures
\begin{align}
\int_{B_{4R_{0}'}(0)} |\nabla \psi_{R_{0}'}|^{2} + \int_{B_{4R_{0}'}(0)} \left( \left( u_{k}^{2} \cdot \chi_{B_{4R_{0}'}(0)^{\rm c}} - m_{k} \right) * Y_{a_{k}} \right) \psi_{R_{0}'}^{2} \leq -1. \label{eq: psi R1 eps est}
\end{align}
Recall \Cref{Lemma - Technical Lemma 2}, that there exists $C_{0} = C_{0}(a_{\rm c}, a_{0},\omega) > 0$ and $R_{0} = R_{0}(a_{\rm c}, a_{0},\omega) > 0$ such that for any $a_{\rm c} < a \leq a_{0}$ and $R_{n} \geq R_{0}$ 
\begin{align}
\int_{\R} |\nabla \psi_{R_{n}}|^{2} - D_{a}(m_{R_{n}}, \psi_{R_{n}}^{2}) \leq - C_{0} R_{n}^{3}, \label{eq: TL2-repeat}
\end{align}
The following term can be estimated and decomposed as
\begin{align}
\int_{B_{4R_{n}}(0)} \left( \left( u_{k}^{2} \cdot \chi_{B_{4R_{n}}(0)^{\rm c}} \right) * Y_{a_{k}} \right) \psi_{R_{n}}^{2} &\leq \int_{B_{4R_{n}}(0)} \left( \left( u_{k}^{2} \cdot \chi_{B_{4R_{n}}(0)^{\rm c}} \right) * Y_{a_{k}} \right) \nonumber \\
= \int_{B_{4R_{n}}(0)} \left( \left( u_{k}^{2} \cdot \chi_{B_{8R_{n}}(0)^{\rm c}} \right) * Y_{a_{k}} \right) &+ \int_{B_{4R_{n}}(0)} \left( \left( u_{k}^{2} \cdot \chi_{B_{8R_{n}}(0) \smallsetminus B_{4R_{n}}(0)} \right) * Y_{a_{k}} \right). \label{eq: split est}
\end{align}
The first term of (\ref{eq: split est}) can be expressed as
\begin{align*}
\int_{B_{4R_{n}}(0)} \left( \left( u_{k}^{2} \cdot \chi_{B_{8R_{n}}(0)^{\rm c}} \right) * Y_{a_{k}} \right) &= \int_{B_{8R_{n}}(0)^{\rm c}} u_{k}^{2}(y) \left( \int_{B_{4R_{n}}(0)} \frac{e^{-a_{k}|x-y|}}{|x-y|} \id x \right) \id y.
\end{align*}
By the triangle inequality $|x-y| \geq \smfrac{|y|}{2}$, hence
\begin{align*}
\int_{B_{4R_{n}}(0)} &\left( \left( u_{k}^{2} \cdot \chi_{B_{8R_{n}}(0)^{\rm c}} \right) * Y_{a_{k}} \right) \\ &\leq \| u_{k} \|_{L^{\infty}(\R)}^{2} \int_{B_{8R_{n}}(0)^{\rm c}} \left( \int_{B_{4R_{n}}(0)} \frac{e^{-a_{\rm c}|y|/2}}{|y|} \id x \right) \id y = C R_{n}^{3} \int_{B_{8R_{n}}(0)^{\rm c}} \frac{e^{-a_{\rm c}|y|/2}}{|y|} \id y \\
&= Ca_{\rm c}^{-2} R_{n}^{3} \left( 1+ 4a_{\rm c}R_{n} \right) e^{-4a_{\rm c}R_{n}} \leq Ca_{\rm c}^{-2} R_{n}^{3} e^{-2a_{\rm c}R_{n}}.
\end{align*}
As $e^{-2a_{\rm c}R_{n}} \to 0$ as $R_{n} \to \infty$, there exists $R_{2} > 0$ such that for $R_{n} \geq R_{2}$
\begin{align}
\int_{B_{4R_{n}}(0)} &\left( \left( u_{k}^{2} \cdot \chi_{B_{8R_{n}}(0)^{\rm c}} \right) * Y_{a_{k}} \right) \leq Ca_{\rm c}^{-2} R_{n}^{3} e^{-2a_{\rm c}R_{n}} \leq \frac{C_{0}}{4} R_{n}^{3}. \label{eq: split est part 1}
\end{align}
Now define $R_{0}' = \max\{R_{0},R_{2}, (2C_{0})^{-1/3}\}$ and choose $R_{n} = R_{0}'$. The second term of (\ref{eq: split est}) can be estimated using Young's inequality for convolutions
\begin{align*}
\int_{B_{4R_{0}'}(0)} &\left( \left( u_{k}^{2} \cdot \chi_{B_{8R_{0}'}(0) \smallsetminus B_{4R_{0}'}(0)} \right) * Y_{a_{k}} \right) \leq \int_{B_{4R_{0}'}(0)} \left( \left( u_{k}^{2} \cdot \chi_{B_{8R_{0}'}(0)} \right) * Y_{a_{k}} \right) \\ &\leq CR_{0}'^{3} \| Y_{a_{k}} \|_{L^{1}(\R)} \| u_{k} \|_{L^{\infty}(B_{8R_{0}'}(0))}^{2} \leq C a_{\rm c}^{-2} R_{0}'^{3} \| u_{k} \|_{L^{\infty}(B_{8R_{0}'}(0))}^{2}.
\end{align*}
As $u_{k} \to 0$ on compact sets, there exists $k_{1} \in \mathbb{N}$ such that $k \geq k_{1}$ ensures that
\begin{align}
\int_{B_{4R_{0}'}(0)} &\left( \left( u_{k}^{2} \cdot \chi_{B_{8R_{0}'}(0) \smallsetminus B_{4R_{0}'}(0)} \right) * Y_{a_{k}} \right) \leq C a_{\rm c}^{-2} R_{0}'^{3} \| u_{k} \|_{L^{\infty}(B_{8R_{0}'}(0))}^{2} \leq \frac{C_{0}}{4} R_{0}'^{3}. \label{eq: split est part 2}
\end{align}
Choose $R_{n} = R_{0}'$ and recall that $R_{n_{k}} \uparrow \infty$, hence there exists $k_{2} \in \mathbb{N}$ such that $R_{n_{k}} \geq R_{0}'$ for all $k \geq k_{2}$, so it follows that $m_{k} \geq m_{R_{n}}$. Collecting the estimates (\ref{eq: TL2-repeat}), (\ref{eq: split est})--(\ref{eq: split est part 2}) with $R_{n} = R_{0}'$ and observing that $\smfrac{C_{0}}{4} R_{0}'^{3} \geq 1$ yields the desired estimate (\ref{eq: psi R1 eps est})
\begin{align*}
\int_{B_{4R_{0}'}(0)} &|\nabla \psi_{R_{0}'}|^{2} + \int_{B_{4R_{0}'}(0)} \left( \left( u_{k}^{2} \cdot \chi_{B_{4R_{0}'}(0)^{\rm c}} - m_{k} \right) * Y_{a_{k}} \right) \psi_{R_{0}'}^{2} \\ &\leq \int_{\R} |\nabla \psi_{R_{0}'}|^{2} - D_{a}(m_{R_{0}'}, \psi_{R_{0}'}^{2}) + \int_{B_{4R_{0}'}(0)} \left( \left( u_{k}^{2} \cdot \chi_{B_{8R_{0}'}(0)^{\rm c}} \right) * Y_{a_{k}} \right) \\ & \qquad + \int_{B_{4R_{0}'}(0)} \left( \left( u_{k}^{2} \cdot \chi_{B_{8R_{0}'}(0) \smallsetminus B_{4R_{0}'}(0)} \right) * Y_{a_{k}} \right) \\
&\leq - C_{0}R_{0}'^{3} + \frac{C_{0}}{4}R_{0}'^{3} + \frac{C_{0}}{4}R_{0}'^{3} = - \frac{C_{0}}{2}R_{0}'^{3} \leq -1.
\end{align*}
Now choose $R = 4 R_{0}' + 2$ such that $\psi = \psi_{R_{0}'} \in C^{\infty}_{c}(B_{R - 2}(0))$ satisfies the estimate (\ref{eq: psi R1 eps est}) for all $a_{\rm c} < a \leq a_{0}$. Then let $\xi \in C^{\infty}(\R)$ satisfy $0 \leq \xi \leq 1$, $\xi = 1$ on $B_{R-1}^{c}(0)$, $\xi = 0$ on $B_{R-2}(0)$ and for $\varepsilon > 0$, define $\varphi_{\varepsilon,k} \in H^{1}(\R)$ by
\begin{align}
\varphi_{\varepsilon,k}(x) = \varepsilon \psi(x) + \xi(x) u_{k}(x). \label{eq: varphi eps def}
\end{align}
It follows from the definition that $\varphi_{\varepsilon,k}$ satisfies the boundary condition from (\ref{eq: I k,j,R variational problem}), that $\varphi_{\varepsilon,k}|_{\partial B_{R}(0)} = u_{k}$. Observe that as $\psi$ and $\xi \cdot u_{k}$ have disjoint support, the energy $E(\varphi_{\varepsilon,k};k,R)$ can be decomposed as
\begin{align*}
E(\varphi_{\varepsilon,k};k,R) &= E(\varepsilon \psi;k,R) + E( \xi u_{k} ;k,R) \\ & \quad + \varepsilon^{2} \int_{B_{R}(0)} \left( (\xi u_{k})^{2} \cdot \chi_{B_{R}(0)} * Y_{a_{k}} \right) \psi^{2}.
\end{align*}
Recall that $\psi$ satisfies (\ref{eq: psi R1 eps est}), so for $0 <\varepsilon \leq 1$
\begin{align*}
E(\varepsilon \psi;k,R) + \varepsilon^{4} &= \varepsilon^{2} \left( \int_{B_{R}(0)} |\nabla \psi|^{2} + \int_{B_{R}(0)} \left( \left( u_{k}^{2} \cdot \chi_{B_{R}(0)^{\rm c}} - m_{k} \right) * Y_{a_{k}} \right) \psi^{2} \right) \nonumber \\
& \quad + \varepsilon^{10/3} \int_{B_{R}(0)} \psi^{10/3} + \frac{\varepsilon^{4}}{2} \int_{B_{R}(0)} \left( \psi^{2} \cdot \chi_{B_{R}(0)} * Y_{a_{k}} \right) \psi^{2} + \varepsilon^{4} \\
& \leq - \varepsilon^{2} + C \varepsilon^{10/3} R^{3} + C \varepsilon^{4} a_{k}^{-2} R^{3} + \varepsilon^{4} \\
& \leq - \varepsilon^{2} + C \varepsilon^{4} =: - \varepsilon^{2} + C_{3} \varepsilon^{4}.
\end{align*}
Choosing $\varepsilon = \varepsilon_{0} = \min\{1,(2C_{3})^{-1/2}\}$ implies that (\ref{eq: E phi eps neg est}) holds
\begin{align}
E(\varepsilon_{0} \psi;k,R) + \varepsilon_{0}^{4} \leq - \varepsilon_{0}^{2} + C_{3} \varepsilon_{0}^{4} \leq - \frac{ \varepsilon_{0}^{2} }{2} =: -C_{1} < 0. \label{eq: E phi eps neg est}
\end{align}
Now consider
\begin{align*}
E( \xi u_{k} ;k,R) &= \int_{B_{R}(0)} |\nabla (\xi u_{k})|^{2} + \int_{B_{R}(0)} (\xi u_{k})^{10/3} - \int_{B_{R}(0)} \left( m_{k} * Y_{a_{k}}  \right) (\xi u_{k})^{2}  \bigg. \nonumber \\
& \quad + \frac{1}{2} \int_{B_{R}(0)} \left( (\xi u_{k})^{2} \cdot \chi_{B_{R}(0)} * Y_{a_{k}} \right) (\xi u_{k})^{2} + \int_{B_{R}(0)} \left( u_{k}^{2} \cdot \chi_{B_{R}(0)^{\rm c}} * Y_{a_{k}} \right) (\xi u_{k})^{2}.
\end{align*}
Using that $0 \leq \xi \leq 1$, $|\nabla \xi| \in L^{\infty}(\R)$, $u_{k} \to 0$ as $k \to \infty$ and following the proof of (\ref{eq: E u k,j j limit zero est}), it follows that $E( \xi u_{k} ;k,R) \to 0$ as $k \to \infty$. For the remaining term
\begin{align}
0 \leq \varepsilon_{0}^{2} \int_{B_{R}(0)} \left( (\xi u_{k})^{2} \cdot \chi_{B_{R}(0)} * Y_{a_{k}} \right) \psi^{2} &\leq C \varepsilon_{0}^{2} \| u_{k} \|_{L^{\infty}(B_{R}(0))}^{2} \left \| Y_{a_{k}} \right \|_{L^{1}(\R)}  \int_{B_{R}(0)} \psi^{2} \nonumber \\
&= \frac{C \varepsilon_{0}^{2}}{a_{\rm c}^{2}} \| u_{k} \|_{L^{\infty}(B_{R}(0))}^{2} \to 0 \text{ as } k \to \infty. \label{eq:e E u k,j j limit zero}
\end{align}
It follows that there exists $k_{2} \in \mathbb{N}$ such that for all $k \geq k_{2}$
\begin{align}
E( \xi u_{k} ;k,R) + \varepsilon_{0}^{2} \int_{B_{R}(0)} \left( (\xi u_{k})^{2} \cdot \chi_{B_{R}(0)} * Y_{a_{k}} \right) \psi^{2} \leq \varepsilon_{0}^{4}. \label{eq: xi u k,j epsilon est}
\end{align}
Combining (\ref{eq: E phi eps neg est}) and (\ref{eq: xi u k,j epsilon est}), for $k \geq \max \{k_{1}, k_{2} \}$ yields the desired estimate (\ref{eq: varphi epsilon contr statement}).
\begin{align*}
E(\varphi_{\varepsilon_{0},k};k,R) &= E(\varepsilon_{0} \psi;k,R) + E( \xi u_{k} ;k,R) \\ & \quad + \varepsilon_{0}^{2} \int_{B_{R}(0)} \left( (\xi u_{k})^{2} \cdot \chi_{B_{R}(0)} * Y_{a_{k}} \right) \psi^{2} \\
&\leq E(\varepsilon_{0} \psi;k,R) + \varepsilon_{0}^{4} \leq - C_{1} < 0,
\end{align*}
which contradicts the initial assumption (\ref{eq: u pos contr initial assumption}).
\end{proof}

\begin{proof}[Proof of \Cref{Proposition - Uniform Yuk inf u estimate remainder}]
The estimate (\ref{eq:u-Yuk-inf-est-remainder}) is shown by contradiction, so suppose there exists $a_{0} > a_{\rm c}$ such that
\begin{align}
\inf_{a_{\rm c} < a \leq a_{0}} \inf_{m \in \mathcal{M}_{L^{2}}(M,\omega)}\inf_{x \in \R} u_{a,m}(x) = 0, \label{eq: main result contr hyp}
\end{align}
hence for each $k \in \mathbb{N}$, there exists $a_{k} \in (a_{\rm c},a_{0}], m_{k} \in \mathcal{M}_{L^{2}}(M,\omega)$ and $x_{k} \in \R$ such that $u_{a_{k},m_{k}}(x_{k}) \leq \smfrac{1}{k}.$ Without loss of generality, assume that $x_{k} = 0$ for all $k \in \mathbb{N}$, otherwise translate $u_{a_{k},m_{k}}$. For convenience, $u_{a_{k},m_{k}}$ will be referred to as $u_{k}$ in this argument. By the Harnack inequality, it follows that $u_{k}$ converges uniformly to $0$ on compact sets. 

For $R > 0$ and $k \in \mathbb{N}$, define the energy functional acting on $v$ satisfying $\nabla v \in L^{2}(B_{R}(0))$ and $v \in L^{10/3}(B_{R}(0))$ by
\begin{align}
E(v;k,R) &= \int_{B_{R}(0)} |\nabla v|^{2} + \int_{B_{R}(0)} v^{10/3} - \int_{B_{R}(0)} \left( m_{k} * Y_{a_{k}} \right) v^{2}  \bigg. \nonumber \\
& \quad + \frac{1}{2} \int_{B_{R}(0)} \left( v^{2} \cdot \chi_{B_{R}(0)} * Y_{a_{k}} \right) v^{2} + \int_{B_{R}(0)} \left( u_{k}^{2} \cdot \chi_{B_{R}(0)^{\rm c}} * Y_{a_{k}} \right) v^{2}. \label{eq: E k,R def 2}
\end{align}
Then consider the corresponding variational problem
\begin{align}
I(k,R) &= \inf \bigg \{ E(v;k,R) \, \bigg| \, \, \nabla v \in L^{2}(B_{R}(0)), v \in L^{10/3}(B_{R}(0)),  v|_{\partial B_{R}(0)} = u_{k} \, \bigg \}. \label{eq: I k,j,R variational problem 2}
\end{align}
The construction of the energy (\ref{eq: E k,R def 2}) and the boundary condition of (\ref{eq: I k,j,R variational problem 2}) ensures that $u_{k}$ is the unique minimiser of (\ref{eq: I k,j,R variational problem 2}) for each $R > 0$. It follows that for any fixed $R > 0$, $I(k,R) \to 0$ as $k \to \infty$. Then by following the construction used in the proof of \Cref{Lemma - Yukawa Finite Lower Bound}, there exists $R > 0$ and $\varphi_{\varepsilon,k}$ such that for sufficiently small $\varepsilon > 0$ and sufficiently large $k \in \mathbb{N}$
\begin{align*}
I(k,R) = E(u_{k};k,R) \leq E(\varphi_{\varepsilon,k};k,R) \leq - C_{1} < 0,
\end{align*}
which contradicts the fact that $I(k,R) \to 0$ as $k \to \infty$, hence the desired estimate (\ref{eq:u-Yuk-inf-est-remainder}) holds.

Consequently, as for all $a > 0$ and $m \in \mathcal{M}_{L^{2}}(M,\omega)$, the electron density satisfies $\inf u_{a} > 0$, the argument presented in \cite[Chapter 6]{C/LB/L} can be applied verbatim to guarantee the uniqueness of the corresponding ground state $(u_{a},\phi_{a})$.
\end{proof}

%

\bibliographystyle{plain}
\bibliography{sample2}

\begin{thebibliography}{10}

\bibitem{Ashcroft}
N.~Ashcroft and N.~Mermin.
\newblock {\em {Solid State Physics}}.
\newblock Saunders College, 1976.

\bibitem{Aubin_HLS}
T.~Aubin.
\newblock {\em Nonlinear Analysis on Manifolds. Monge--Ampere Equations},
  volume 252.
\newblock Springer Science \& Business Media, 1982.

\bibitem{Blanc_Uniqueness}
X.~Blanc.
\newblock Unique solvability of a system of nonlinear elliptic pdes arising in
  solid state physics.
\newblock {\em SIAM Journal on Mathematical Analysis}, 38(4):1235--1248, 2006.

\bibitem{Blanc-DefinitionofGroundState}
X.~Blanc, C.~Le Bris, and P.-L. Lions.
\newblock A definition of the ground state energy for systems composed of
  infinitely many particles.
\newblock {\em Communications in Partial Differential Equations},
  28(1-2):439--475, 2003.

\bibitem{Lahbabi-MeanField}
E.~Canc{\`e}s, S.~Lahbabi, and M.~Lewin.
\newblock Mean-field models for disordered crystals.
\newblock {\em Journal de Math\'ematiques Pures et Appliqu\'ees}, 100(2):241 --
  274, 2013.

\bibitem{C/LB/L}
I.~Catto, C.~Le~Bris, and P.-L. Lions.
\newblock {\em The Mathematical Theory of Thermodynamic Limits:
  {T}homas--{F}ermi Type Models}.
\newblock Oxford Mathematical Monographs. The Clarendon Press Oxford University
  Press, New York, 1998.

\bibitem{Paper2}
H.~Chen, F.~Nazar, and C.~Ortner.
\newblock Models for crystalline defects.
\newblock In preparation.

\bibitem{EOS}
V.~Ehrlacher, C.~Ortner, and A.~Shapeev.
\newblock Analysis of boundary conditions for crystal defect atomistic
  simulations.
\newblock {\em ArXiv e-prints}, 1306.5334, 2013.

\bibitem{Evans}
L.~Evans.
\newblock {\em Partial Differential Equations}, volume~19 of {\em Graduate
  Studies in Mathematics}.
\newblock American Mathematical Society, Providence, RI, second edition, 2010.

\bibitem{Evans/Gariepy}
L.~Evans and R.~Gariepy.
\newblock {\em Measure Theory and Fine Properties of Functions}, volume~5.
\newblock CRC press, 1991.

\bibitem{Paper1}
Nazar F. and C.~Ortner.
\newblock {L}ocality of the {T}homas-{F}ermi-von {W}eizs\"{a}cker {E}quations.
\newblock {\em arXiv preprint arXiv:1509.06753}, 2015.

\bibitem{Gilbarg/Trudinger}
D.~Gilbarg and N.~Trudinger.
\newblock {\em Elliptic Partial Differential Equations of Second Order}.
\newblock Classics in Mathematics. Springer-Verlag, Berlin, 2001.
\newblock Reprint of the 1998 edition.

\bibitem{Kaxiras}
E.~Kaxiras.
\newblock {\em Atomic and electronic structure of solids}.
\newblock Cambridge University Press, 2003.

\bibitem{Kittel}
C.~Kittel and C.-Y. Fong.
\newblock {\em Quantum Theory of Solids}.
\newblock Wiley New York, 1963.

\bibitem{Lahbabi-Defects}
S.~Lahbabi.
\newblock The reduced {H}artree--{F}ock model for short-range quantum crystals
  with nonlocal defects.
\newblock {\em Annales Henri Poincar\'e}, 15(7):1403--1452, 2014.

\bibitem{Lieb/Simon_TF}
E.~Lieb and B.~Simon.
\newblock The {T}homas--{F}ermi theory of atoms, molecules and solids.
\newblock {\em Advances in Mathematics}, 23(1):22--116, 1977.

\bibitem{Rowlinson}
J.S. Rowlinson.
\newblock The {Y}ukawa potential.
\newblock {\em Physica A: Statistical Mechanics and its Applications},
  156(1):15 -- 34, 1989.

\bibitem{Solovej_Universality}
J.~Solovej.
\newblock Universality in the {T}homas--{F}ermi-von {W}eizs{\"a}cker model of
  atoms and molecules.
\newblock {\em Communications in Mathematical Physics}, 129(3):561--598, 1990.

\bibitem{Trudinger_MeasurableCoefficients}
N.~Trudinger.
\newblock Linear elliptic operators with measurable coefficients.
\newblock {\em Annali della Scuola Normale Superiore di Pisa-Classe di
  Scienze}, 27(2):265--308, 1973.

\bibitem{YuTrinkleBaderVolumes}
M.~Yu, D.~Trinkle, and R.~Martin.
\newblock Energy density in density functional theory: Application to
  crystalline defects and surfaces.
\newblock {\em Physical Review B}, 83(11):115113, 2011.

\end{thebibliography}

\end{document}